%% file: main.tex
\newif\ifels
\elsfalse
\ifels
	\documentclass[1p, sort&compress]{elsarticle}
	\newproof{proof}{Proof}
	\newtheorem{theorem}{Theorem}
	\journal{Applied and Computational Harmonic Analysis}
	\usepackage{lineno}
	\usepackage{color}
\else
	\documentclass[10pt, a4paper, onecolumn, twoside]{amsart}
	\input{misc/format}
	\usepackage{cite} 
	\newtheorem{theorem}{Theorem}[section]
\fi
\usepackage{amsmath} \usepackage{algorithm, algpseudocode}
\usepackage{amssymb, bm, graphicx, hyperref, enumerate, rotating}
\usepackage{verbatim} 
\usepackage{subfigure}
\input{misc/macro+variables}
\newcommand{\longtitle}{{Random sampling of bandlimited signals on graphs}}
\newcommand{\GPlong}{{Gilles~Puy}}
\newcommand{\PVlong}{{Pierre~Vandergheynst}}
\newcommand{\RGlong}{{R\'emi Gribonval}}
\newcommand{\NTlong}{{Nicolas Tremblay}}
\newcommand{\GPshort}{{G.~Puy}}
\newcommand{\PVshort}{{P.~Vandergheynst}}
\newcommand{\RGshort}{{R.~Gribonval}}
\newcommand{\NTshort}{{N.~Tremblay}}
\newcommand{\Technicolor}{{Technicolor, 975 Avenue des Champs Blancs, 35576 Cesson-S\'evign\'e, France}}
\newcommand{\INRIA}{{INRIA Rennes - Bretagne Atlantique, Campus de Beaulieu, FR-35042 Rennes Cedex, France}}
\newcommand{\EPFL}{{Institute of Electrical Engineering, Ecole Polytechnique F{\'e}d{\'e}rale de Lausanne (EPFL), CH-1015 Lausanne, Switzerland}}
\newcommand{\ERC}{{This work was partly funded by the European Research Council, PLEASE project (ERC-StG-2011-277906), and by the Swiss National Science Foundation, grant 200021-154350/1 - Towards Signal Processing on Graphs}}
\newcommand{\INRIAGilles}{The first version of this manuscript was submitted when \GPshort\ was at \INRIA}
\graphicspath{{figs/}}
\newcommand{\MODIF}[1]{{\color{black}{#1}}}
\newcommand{\ADDED}[1]{{\color{black}{#1}}}

\ifels
	\title{\longtitle\tnoteref{t1}}
	\tnotetext[t1]{\INRIAGilles. \ERC.}
	\author[technicolor]{\GPlong}
	\author[inria,epfl]{\NTlong}
	\author[inria]{\RGlong}
	\author[inria,epfl]{\PVlong}
	\address[technicolor]{\Technicolor}
	\address[inria]{\INRIA}
	\address[epfl]{\EPFL}
\else
	\title[\longtitle]{\longtitle}
	\author[\GPshort]{\GPlong}
	\author[\NTshort]{\NTlong}
	\author[\RGshort]{\RGlong}
	\author[\PVshort]{\PVlong}
	\thanks{\GPshort\ is with \Technicolor. \NTshort, \RGshort\ and \PVshort\ are with \INRIA. \NTshort\ and \PVshort\ are also with the \EPFL. \INRIAGilles. \ERC.}
\fi

\begin{document}

\ifels
	\input{manuscript/abstract}
	\maketitle
	\input{manuscript/core}
	\section*{Acknowledgements}
	\ERC. \INRIAGilles. 
	\input{manuscript/appendix}
\else
	\maketitle

\input{manuscript/abstract}\input{manuscript/core}\input{manuscript/appendix}
\fi

\ifels
	\bibliographystyle{elsarticle-num}
	\section*{References}
\else
	\bibliographystyle{IEEEtran}
\fi
\bibliography{biblio}

\end{document}

%% file: misc/format.tex
\usepackage{titlesec, color, amsthm}
\definecolor{blue1}{rgb}{0,0,0}
\renewcommand\thesection{\arabic{section}}
\renewcommand{\theequation}{\arabic{equation}}

\titleformat{\section}[hang]{\color{blue1}\large\bfseries\sffamily}{\thesection}{0mm}{. }[]
\titleformat{\subsection}[hang] {\color{blue1}\bfseries\sffamily}{\thesubsection}{0em}{. }[]
\titleformat{\subsubsection}[hang] {\color{blue1}\sffamily}{\thesubsubsection}{0em}{. }[]
\titlespacing*{\section}{1em}{3.5ex plus .2ex minus .2ex}{1ex plus .2ex}
\titlespacing*{\subsection}{0em}{3ex plus .2ex minus .2ex}{1ex plus .2ex}
\titlespacing*{\subsubsection}{0em}{3ex plus .2ex minus .2ex}{1ex plus .2ex}

\renewenvironment{abstract}{{\color{blue1}\small\bfseries Abstract.}\footnotesize}{\par \vskip .1in}
\makeatletter
\def\@setauthors{
\begingroup 
\def \thanks{\protect\thanks@warning}
\trivlist \centering\footnotesize \@topsep30\p@\relax \advance\@topsep by -\baselineskip
\item\relax \author@andify \authors \def\\{\protect\linebreak} {\color{blue1}\large\authors} \endtrivlist \endgroup}
\def\@settitle{\centering{\color{blue1} \Large \bfseries \bfseries \@title \par}}
\makeatother
\usepackage[normal, small, up, textfont=it]{caption}

%% file: misc/macro+variables.tex
\newcommand{\Graph}{\ensuremath{\set{G}}}

\newcommand{\CCal}{\ensuremath{\set{C}}}

\newcommand{\Lap}{\ensuremath{\ma{L}}}
\newcommand{\Fou}{\ensuremath{\ma{U}}}
\newcommand{\Eig}{\ensuremath{\ma{\Lambda}}}
\newcommand{\Meas}{\ensuremath{\ma{M}}}

\newcommand{\fou}{\ensuremath{\vec{u}}}
\newcommand{\eig}{\ensuremath{\vec{\lambda}}}
\newcommand{\err}{\ensuremath{\vec{n}}}
\newcommand{\sig}{\ensuremath{\vec{x}}}
\newcommand{\meas}{\ensuremath{\vec{y}}}
\newcommand{\prob}{\ensuremath{\vec{p}}}

\newcommand{\nbVert}{\ensuremath{n}}
\newcommand{\nbVertRed}{\ensuremath{m}}
\newcommand{\nbClass}{\ensuremath{k}}

\newcommand{\cumCoh}{\ensuremath{\nu}}
\newcommand{\reg}{\ensuremath{\gamma}}

\newcommand{\refeq}[1]{(\ref{#1})}

\newcommand{\Rbb}{\ensuremath{\mathbb{R}}} 
\newcommand{\Nbb}{\ensuremath{\mathbb{N}}}
\newcommand{\Ebb}{\ensuremath{\mathbb{E}}}
\newcommand{\Pbb}{\ensuremath{\mathbb{P}}}

\renewcommand{\leq}{\ensuremath{\leqslant}}
\renewcommand{\geq}{\ensuremath{\geqslant}}

\newcommand{\adjoint}{\ensuremath{{\intercal}}}
\newcommand{\inv}[1]{\ensuremath{\frac{1}{#1}}}

\newcommand{\norm}[1]{\ensuremath{\left\| #1\right\|}}
\newcommand{\abs}[1]{\ensuremath{\left| #1 \right|}}

\newcommand{\ma}[1]{\ensuremath{\mathsf{#1}}}
\renewcommand{\vec}[1]{\ensuremath{\bm{#1}}}
\newcommand{\diag}{\ensuremath{{\rm diag}}}

\newcommand{\set}[1]{\ensuremath{\mathcal{#1}}}

\newcommand{\spann}{\ensuremath{{\rm span}}}

\newcommand{\ie}{\textit{i.e.}}
\newcommand{\eg}{\textit{e.g.}}
\renewcommand{\th}{\ensuremath{\text{th}}}
\newcommand{\ee}{\ensuremath{{\rm e}}}

\newtheorem{definition}[theorem]{Definition}
\newtheorem{lemma}[theorem]{Lemma}
\newtheorem{corollary}[theorem]{Corollary}

\algnewcommand\algorithmicinput{\textbf{Input:}}
\algnewcommand\Input{\item[\algorithmicinput]}

\algnewcommand\algorithmicoutput{\textbf{Output:}}
\algnewcommand\Output{\item[\algorithmicoutput]}

%% file: manuscript/abstract.tex
\begin{abstract}
We study the problem of sampling $k$-bandlimited signals on graphs. We propose two sampling strategies that consist in selecting a small subset of nodes at random. The first strategy is non-adaptive, \ie, independent of the graph structure, and its performance depends on a parameter called the graph coherence. On the contrary, the second strategy is adaptive but yields optimal results. Indeed, no more than $O(\nbClass \log(\nbClass))$ measurements are sufficient to ensure an accurate and stable recovery of all $\nbClass$-bandlimited signals. This second strategy is based on a careful choice of the sampling distribution, which can be estimated quickly. Then, we propose a computationally efficient decoder to reconstruct $k$-bandlimited signals from their samples. We prove that it yields accurate reconstructions and that it is also stable to noise. Finally, we conduct several experiments to test these techniques.
\end{abstract}

%% file: manuscript/core.tex

\section{Introduction}

Graphs are a central modelling tool for network-structured data~\cite{newman_book2010}. Depending on the application, the nodes of a graph may represent people in social networks, brain regions in neuronal networks, or stations in transportation networks. Data on a graph, such as individual hobbies, activity of brain regions, traffic at a station, may be represented by scalars defined on each node, which form a graph signal. Extending classical signal processing methods to graph signals is the purpose of the emerging field of graph signal processing~\cite{shuman_SPMAG2013, sandryhaila_SPMAG2014}. 

Within this framework, a cornerstone is sampling, \ie, measuring a graph signal on a reduced set of nodes carefully chosen to enable stable reconstructions. Classically, sampling a continuous signal $x(t)$ consists in measuring a countable sequence of its values, $\{ x(t_j) \}_{j \in \mathbb{Z}}$, that ensures its recovery under a given smoothness model~\cite{unser_shannon2000}. Smoothness assumptions are often defined in terms of the signal's Fourier transform. For example, Shannon's famous sampling theorem~\cite{shannon_IRE1949} states that any $\omega$-bandlimited signal can be recovered \textit{exactly} from its values at $t_j=j/2\omega$. Similar theorems exist for other classes of signals, \eg, signals on the sphere \cite{mcewen11}; and other types of sampling schemes, \eg, irregular sampling~\cite{grochenig1992reconstruction,benedetto1992irregular} or compressive sampling~\cite{candes2006compressive}. Extending these theorems to graph signals requires to decide on a smoothness model and to design a sampling scheme that enables stable recovery. 

Natural choices of smoothness models build upon, \eg, the graph's adjacency matrix, the combinatorial Laplacian matrix, the normalised Laplacian, or the random walk Laplacian. The sets of eigenvectors of these operators define different graph Fourier bases. Given such a Fourier basis, the equivalent of a classical $\omega$-bandlimited signal is a $k$-bandlimited graph signal whose $k$ first Fourier coefficients are non-null~\cite{chen_TSP2015,anis_ICASSP2014}. 

Unlike continuous time signal processing, the concept of regular sampling itself is not applicable for graph signals, apart for very regular graphs such as bipartite graphs \cite{narang_TSP2013}. We are left with two possible choices for sampling: irregular or random sampling. Irregular sampling of $k$-bandlimited graph signals has been studied first by Pesenson~\cite{pesenson_paley, Pesenson:2011tu} who introduced the notion of uniqueness set associated to the subspace of $k$-bandlimited graph signals. If two $k$-bandlimited graph signals are equal on a uniqueness set, they are necessarily equal on the whole graph. Building upon this first work, and using the fact that the sampling matrix applied to the first $k$ Fourier modes should have rank $k$ in order to guarantee recovery of $k$-bandlimited signals, Anis \emph{et al.}~\cite{anis_ICASSP2014,Anis_ARXIV2015} 
and Chen \emph{et al.}~\cite{chen_ICASSP2015,chen_TSP2015} showed that a sampling set of size $k$ that perfectly embeds $k$-bandlimited signals always exists. To find such an optimal set, the authors need to compute the first $k$ eigenvectors of the Laplacian, which is computationally prohibitive for large graphs. A recent work~\cite{anis_arxiv2015_long} bypasses the partial diagonalisation of the Laplacian by using graph spectral proxies, but the procedure to find an optimal sampling set still requires a search over all possible subsets of nodes of a given size. This is a very large combinatorial problem. In practice, approximate results are obtained using a greedy heuristic that enables the authors to efficiently perform experiments on graphs of size up to few thousands nodes. 

Several other sampling schemes exist in the literature, such as schemes based on a bipartite decomposition of the graph~\cite{narang_TSP2013}, on a decomposition via maximum spanning trees~\cite{nguyen_TSP2015}, on the sign of the last Fourier mode~\cite{shuman_ARXIV2013}, on the sign of the Fiedler vector~\cite{irion_SPIE2015}, or on a decomposition in communities~\cite{tremblay_arxiv2015}. All these propositions are however specifically designed for graph multiresolution analysis with filterbanks, and are not suited to find optimal or close-to-optimal sets of nodes for sampling $k$-bandlimited graph signals.

\subsection{Main Contributions} 

In this paper, we propose a very different approach to sampling on graphs. Instead of trying to find an \textit{optimal} sampling set, (\ie, a set of size $k$) for $k$-bandlimited signals, we relax this optimality constraint in order to tackle graphs of very large size. We allow ourselves to sample slightly more than $k$ nodes and, inspired by compressive sampling, we propose two random sampling schemes that ensure recovery of graph signals with high probability.

A central graph characteristic that appears from our study is the \emph{graph weighted coherence} of order $k$ (see Definition~\ref{def:Gr_w_coh}). This quantity is a measure of the localisation of the first $k$ Fourier modes on the nodes of the graph. Unlike the classical Fourier modes, some graph Fourier modes have the surprising potential of being localised on very few nodes. The farther a graph is from a regular grid, the higher the chance to have a few localised Fourier modes. This particularity in graph signal processing is studied in~\cite{agaskar_TIT2013,rabbat_icassp2014, Nakatsukasa_laa2013} but is still largely not understood. 

First, we propose a non-adaptive sampling technique that consists in choosing a few nodes at random to form the sampling set. In this setting, we show that the number of samples ensuring the reconstruction of all $k$-bandlimited signals scales with the square of the graph weighted coherence. For regular or almost-regular graphs, \ie, graphs whose coherence is close to $\sqrt{k}$, this result shows that $O(k\log{k})$ samples selected using the uniform distribution are sufficient to sample $k$-bandlimited signals. We thus obtain an almost optimal sampling condition.

Second, for arbitrary graphs with a coherence potentially tending to $\sqrt{\nbVert}$, where $\nbVert \gg k$ is the total number of nodes, we propose a second sampling strategy that compensates the undesirable consequences of mode localisation. The technique relies on the variable density sampling strategy widely used in compressed sensing \cite{puy11, krahmer12, adcock14}. We prove that there always exists a sampling distribution such that no more than $O(k\log{k})$ samples are sufficient to ensure exact and stable reconstruction of all $k$-bandlimited signals, whatever the graph structure. Unfortunately, computing the optimal sampling distribution requires the partial diagonalisation of the first $k$ eigenvectors of the Laplacian. To circumvent this issue, we propose a fast technique to estimate this optimal sampling distribution accurately. 

Finally, we propose an efficient method to reconstruct any $k$-bandlimited signal from its samples. We prove that the method recovers $k$-bandlimited signals exactly in the absence of noise. We also prove that the method is robust to measurement noise and model errors.

Note that our sampling theorems are applicable to any symmetrical Laplacian or adjacency matrix, \ie, any weighted undirected graph. Nevertheless, the efficient recovery method we propose is specifically designed to take 
advantage of the semi-definite positivity of the Laplacian operator. In the following, we therefore concentrate on such symmetrical positive semi-definite Laplacians, such as the combinatorial or normalized Laplacians. 

Let us acknowledge that the idea of random sampling for $k$-bandlimited graph signals is mentioned in \cite{chen_TSP2015} and \cite{chen_sampta2015}. In \cite{chen_TSP2015}, the authors prove that the space of $k$-bandlimited graph signals can be stably embedded using a uniform sampling but for the Erd\H{o}s-R\'enyi graph only. The idea of using a non-uniform sampling appears in \cite{chen_sampta2015}. However, the authors do not prove that this sampling strategy provides a stable embedding of the space of $k$-bandlimited graph signals. We prove this result in Section~\ref{sec:rip} but also show that there always exists a sampling distribution that yields optimal results. Finally, the reconstruction methods proposed in \cite{chen_sampta2015} requires a partial diagonalisation of the Laplacian matrix, unlike ours. We also have much stronger recovery guarantees than the ones presented in \cite{chen_sampta2015}, which are expected recovery guarantees.

\subsection{Notations and definitions} 

For any matrix $\ma{X} \in \Rbb^{m \times \nbVert}$, $\norm{\ma{X}}_2$ denotes the spectral norm of $\ma{X}$, $\lambda_{\rm max}(\ma{X})$ denotes the largest eigenvalue of $\ma{X}$, and $\lambda_{\rm min}(\ma{X})$ denotes the smallest eigenvalue of $\ma{X}$. For any vector $\sig \in \Rbb^\nbVert$, $\norm{\sig}_2$ denotes the Euclidean norm of $\sig$. Depending on the context, $\sig_j$ may represent the $j^\th$ entry of the vector $\sig$ or the $j^\th$ column-vector of the matrix $\ma{X}$. The identity matrix is denoted by $\ma{I}$ - its dimensions are determined by the context - and $\vec{\delta}_j$ is its $j^\th$ column vector.

We consider an undirected, connected, weighted graph $\Graph = \{\mathcal{V}, \mathcal{E}, \ma{W} \}$, where $\mathcal{V}$ is the set of $\nbVert$ nodes, $\mathcal{E}$ is the set of edges, and $\ma{W} \in \Rbb^{\nbVert \times \nbVert}$ is the weighted adjacency matrix. The entries of $\ma{W}$ are nonnegative. We denote the graph Laplacian by $\Lap \in \Rbb^{\nbVert \times \nbVert}$. As said before, we assume that $\Lap$ is real, symmetric, and positive semi-definite. For example, the matrix $\Lap$ can be the combinatorial graph Laplacian $\Lap := \ma{D} - \ma{W}$, or the normalised one $\Lap := \ma{I} - \ma{D}^{-1/2} \ma{W} \ma{D}^{-1/2}$, where $\ma{D} \in \Rbb^{\nbVert \times \nbVert}$ is the diagonal degree matrix and $\ma{I}$ is the identity matrix~\cite{chung_book1997}. The diagonal degree matrix $ \ma{D}$ has entries $d_i := \sum_{i \neq j}\ma{W}_{ij}$.

As the matrix $\Lap$ is real symmetric, there exists a set of orthonormal eigenvectors $\Fou \in \Rbb^{\nbVert \times \nbVert}$ and real eigenvalues $\eig_1, \ldots, \eig_n$ such that $\Lap = \Fou \Eig \Fou^\adjoint$, where $\Eig := \diag(\eig_1, \ldots, \eig_n) \in \Rbb^{\nbVert \times \nbVert}$. Furthermore, semi-definite positivity of $\Lap$ implies that all eigenvalues are nonnegative. Without loss of generality, we assume that $\eig_1 \leq \ldots \leq \eig_n$.

The matrix $\Fou$ is often viewed as the graph Fourier transform~\cite{shuman_SPMAG2013}. For any signal $\sig \in \Rbb^{\nbVert}$ defined on the nodes of the graph $\Graph$, $\hat{\sig} = \Fou^\adjoint \sig$ contains the Fourier coefficients of $\sig$ ordered in increasing frequencies. As explained before, it is thus natural to consider that a $k$-bandlimited (smooth) signal $\sig \in \Rbb^{\nbVert}$ on $\Graph$ with band-limit $\nbClass>0$ is a signal that satisfies 
\begin{align*}
\sig = \Fou_{\nbClass} \hat{\sig}^\nbClass
\end{align*}
where $\hat{\sig}^\nbClass \in \Rbb^{\nbClass}$ and
\begin{align*}
\Fou_{\nbClass} := \left( \fou_1, \ldots, \fou_\nbClass \right) \in \Rbb^{\nbVert \times \nbClass},
\end{align*}
\ie, $\Fou_{\nbClass}$ is the restriction of $\Fou$ to its first $\nbClass$ vectors. This yields the following formal definition of a $k$-bandlimited signal.

\begin{definition}[$k$-bandlimited signal on \Graph] A signal $\sig \in \Rbb^{\nbVert}$ defined on the nodes of the graph $\Graph$ is $k$-bandlimited with $\nbClass \in \Nbb\setminus\{0\}$ if $\sig \in \spann(\ma{U}_\nbClass)$.
\end{definition}

Note we use $\spann(\ma{U}_\nbClass)$ in our definition of $k$-bandlimited signals to handle the case where the eigendecomposition is not unique. To avoid any ambiguity in the definition of $\nbClass$-bandlimited signals, we assume that $\eig_{\nbClass} \neq  \eig_{\nbClass+1}$ for simplicity.

\subsection{Outline} 

In Section~\ref{sec:rip}, we detail our sampling strategies and provide sufficient sampling conditions that ensure a stable embedding of $k$-bandlimited graph signals. We also prove that there always exists an optimal sampling distribution that ensures an embedding of $k$-bandlimited signals for $O(k \log(k))$ measurements. In Section~\ref{sec:reconstruction}, we propose decoders able to recover $k$-bandlimited signals from their samples. In Section~\ref{sec:estimation_distribution}, we explain how to obtain an estimation of the optimal sampling distribution quickly, \textit{without} partial diagonalisation of the Laplacian matrix. In Section~\ref{sec:experiments}, we conduct several experiments on different graphs to test our methods. Finally, we conclude and discuss perspectives in Section~\ref{sec:conclusion}.

\section{Sampling $k$-bandlimited signals}
\label{sec:rip}

In this section, we start by describing how we select a subset of the nodes to sample $k$-bandlimited signals. Then, we prove that this sampling procedure stably embeds the set of $k$-bandlimited signals. We describe how to reconstruct such signals from these measurements in Section~\ref{sec:reconstruction}.

\subsection{The sampling procedure}

In order to select the subset of nodes that will be used for sampling, we need a probability distribution $\mathcal{P}$ on $\{ 1, \ldots, \nbVert \}$. This probability distribution is used as a sampling distribution. We represent it by a vector $\prob \in \Rbb^\nbVert$. We assume that $\prob_i > 0$ for all $i = 1, \ldots, \nbVert$. We obviously have $\norm{\prob}_1 = \sum_{i=1}^\nbVert \prob_i = 1$. We associate the matrix 
\begin{align*}
\ma{P} := \diag(\prob) \in \Rbb^{\nbVert \times \nbVert}
\end{align*}
to $\prob$.

The subset of nodes $\Omega := \{\omega_1, \ldots, \omega_\nbVertRed \}$ used for sampling is constructed by drawing independently (with replacements) $\nbVertRed$ indices from the set $\{1, \ldots, \nbVert \}$ according to the probability distribution $\prob$. We thus have 
\begin{align*}
\Pbb(\omega_j = i) = \prob_i,\quad \forall j \in \{1, \ldots, \nbVertRed \} \text{ and } \forall i \in \{1, \ldots, \nbVert \}.
\end{align*}
For any signal $\sig \in \Rbb^n$ defined on the nodes of the graph, its sampled version $\meas \in \Rbb^\nbVertRed$ satisfies
\begin{align*}
\meas_j := \sig_{\omega_j}, \quad \forall j \in \{1, \ldots, \nbVertRed \}.
\end{align*}
Note that we discuss the case of sampling without replacement in Section~\ref{sec:sampling_without_replacement}.

Let us pause for a moment and highlight few important facts. First, the sampling procedure allows each node to be selected multiple times. The number of measurements $\nbVertRed$ includes these duplications. In practice, one can sample each selected node only once and add these duplications ``artificially'' afterwards. Second, \emph{the set of nodes $\Omega$ needs to be selected only once to sample all $k$-bandlimited signals on $\Graph$.} One does not need to construct a set $\Omega$ each time a signal has to be sampled. Third, note that the sampling procedure is so far completely independent of the graph $\Graph$. This is a non-adaptive sampling strategy.

Let us define the sampling matrix $\Meas \in \Rbb^{\nbVertRed \times \nbVert}$. This matrix satisfies
\begin{align}
\label{eq:subsampling_matrix_def}
\Meas_{ij} := 
\left\{
\begin{array}{ll}
1 & \text{if } j = \omega_i\\
0 & \text{otherwise},
\end{array}
\right.
\end{align}
for all $i \in \{1, \ldots, \nbVert \}$ and $j \in \{1, \ldots, \nbVertRed \}$. Note that $\meas = \Meas \sig$. In the next section, we show that, with high probability, $\Meas$ embeds the set of $\nbClass$-bandlimited signals for a number of measurements $\nbVertRed$ essentially proportional to $\nbClass \log(\nbClass)$ times a parameter called the graph weighted coherence.

\subsection{The space of $k$-bandlimited signals is stably embedded}

Similarly to many compressed sensing results, the number of measurements required to stably sample $k$-bandlimited signals will depend on a quantity, called the \emph{graph weighted coherence}, that represents how the energy of these signals spreads over the nodes. Before providing the formal definition of this quantity, let us give an intuition of what it represents and why it is important.

Consider the signal $\vec{\delta}_i \in \Rbb^n$ with value $1$ at node $i$ and $0$ everywhere else. This signal has its energy concentrated entirely at the $i^\th$ node. Compute $\Fou_\nbClass^\adjoint \vec{\delta}_i$, \ie, the first $\nbClass$ Fourier coefficients of $\vec{\delta}_i$. The ratio
\begin{align*}
\frac{\norm{\Fou_\nbClass^\adjoint \vec{\delta}_i}_2}{\norm{\Fou^\adjoint \vec{\delta}_i}_2}
= 
\frac{\norm{\Fou_\nbClass^\adjoint \vec{\delta}_i}_2}{\norm{\vec{\delta}_i}_2}
= 
\norm{\Fou_\nbClass^\adjoint \vec{\delta}_i}_2
\end{align*}
characterises how much the energy of $\vec{\delta}_i$ is concentrated on the first $\nbClass$ Fourier modes. This ratio varies between $0$ and $1$. When it is equal to $1$, this indicates that there exists $k$-bandlimited signals whose energy is solely concentrated at the $i^\th$ node; not sampling the $i^\th$ node jeopardises the chance of reconstructing these signals. When this ratio is equal to $0$, then no $k$-bandlimited signal has a part of its energy on the $i^\th$ node; one can safely remove this node from the sampling set. We thus see that the quality of our sampling method will depend on the interplay between the sampling distribution $\prob$ and the quantities $\norm{\Fou_\nbClass^\adjoint \vec{\delta}_i}_2$ for $i \in \{1, \ldots, \nbVert\}$. Ideally, we should have $\prob_i$ large wherever $\norm{\Fou_\nbClass^\adjoint \vec{\delta}_i}_2$ is large and $\prob_i$ small wherever $\norm{\Fou_\nbClass^\adjoint \vec{\delta}_i}_2$ is small. The interplay between $\prob_i$ and $\norm{\Fou_\nbClass^\adjoint \vec{\delta}_i}_2$ is characterised by the \emph{graph weighted coherence}.
%
\begin{definition}[Graph weighted coherence]
\label{def:Gr_w_coh}
Let $\prob \in \Rbb^n$ represent a sampling distribution on $\{ 1, \ldots, \nbVert \}$. The graph weighted coherence of order $\nbClass$ for the pair $(\Graph, \prob)$ is
\begin{align*}
\cumCoh_{\prob}^{\nbClass} := \max_{1 \leq i \leq \nbVert} \left\{ \prob_i^{-1/2} \norm{\Fou_\nbClass^\adjoint \vec{\delta}_i}_2 \right\}.
\end{align*}
\ADDED{The quantity $\norm{\Fou_\nbClass^\adjoint \vec{\delta}_i}_2$ is called the local graph coherence at node $i$.}
\end{definition}

Let us highlight two fundamental properties of $\cumCoh_{\prob}^{\nbClass}$. First, we have 
\begin{align*}
\cumCoh_{\prob}^{\nbClass} \; \geq \; \sqrt{\nbClass}.
\end{align*}
Indeed, as the columns of $\Fou_\nbClass$ are normalised to $1$, we have
\begin{align*}
\nbClass = \norm{\Fou_\nbClass}_{\rm Frob}^2 
= \sum_{i=1}^\nbVert \norm{\Fou_\nbClass^\adjoint \vec{\delta}_i}_2^2 
= \sum_{i=1}^\nbVert \prob_i \frac{\norm{\Fou_\nbClass^\adjoint \vec{\delta}_i}_2^2}{\prob_i} 
\; \leq \; 
\max_{1 \leq i \leq \nbVert} \left\{ \frac{\norm{\Fou_\nbClass^\adjoint \vec{\delta}_i}_2^2}{\prob_i} \right\} \cdot \sum_{i=1}^\nbVert \prob_i
= 
(\cumCoh_{\prob}^{\nbClass})^2.
\end{align*}
Second, $\cumCoh_{\prob}^{\nbClass}$ is a quantity that depends solely on $\prob$ and $\spann(\Fou_\nbClass)$. The choice of the basis for $\spann(\Fou_\nbClass)$ does not matter in the definition of $\cumCoh_{\prob}^{\nbClass}$. Indeed, it suffices to notice that $\norm{\Fou_\nbClass^\adjoint \vec{\delta}_i}_2^2 = \norm{{\rm P}_{\nbClass} (\vec{\delta}_i)}_2^2$, where ${\rm P}_{\nbClass}(\cdot): \Rbb^{\nbVert} \rightarrow \Rbb^{\nbVert}$ is the orthogonal projection onto $\spann(\Fou_\nbClass)$, whose definition is independent of the choice of the basis of $\spann(\Fou_\nbClass)$. The graph weighted coherence is thus a characteristic of the interaction between the signal model, \ie, $\spann(\Fou_\nbClass)$, and the sampling distribution $\prob$.

We are now ready to introduce our main theorem which shows that $\nbVertRed^{-1} \Meas\ma{P}^{-1/2}$ satisfies a restricted isometry property on the space of $k$-bandlimited signals.

\begin{theorem}[Restricted isometry property]
\label{th:rip}
Let $\Meas$ be a random subsampling matrix constructed as in \refeq{eq:subsampling_matrix_def} with the sampling distribution $\prob$. For any $\delta, \epsilon \in (0, 1)$, with probability at least $1-\epsilon$,
\begin{align}
\label{eq:RIP}
(1 - \delta) \norm{\sig}_2^2 \leq \inv{\nbVertRed} \norm{\Meas \ma{P}^{-1/2} \; \sig}_2^2 \leq (1 + \delta) \norm{\sig}_2^2
\end{align}
for all $\sig \in \spann(\Fou_{\nbClass})$ provided that
\begin{align}
\label{eq:sampling_condition}
\nbVertRed \geq \frac{3}{\delta^2} \; (\cumCoh_{\prob}^{\nbClass})^2 \; \log\left( \frac{2\nbClass}{\epsilon} \right).
\end{align}
\end{theorem}
\begin{proof}
See \ifels\else Appendix \fi \ref{app:proof_rip}.
\end{proof}

There are several important comments to make about the above theorem.
\begin{itemize}
\item First, this theorem shows that the matrix $\Meas \ma{P}^{-1/2}$ embeds the set of $k$-bandlimited signals into $\Rbb^\nbVertRed$. Indeed, for any $\sig \neq \vec{z} \in \spann(\Fou_\nbClass)$, we have $\norm{\Meas \ma{P}^{-1/2} \; (\sig - \vec{z})}_2^2 \geq \nbVertRed (1 - \delta) \norm{\sig - \vec{z}}_2^2 > 0$, as $(\sig - \vec{z}) \in \spann(\Fou_\nbClass)$. The matrix $\Meas \ma{P}^{-1/2}$ can thus be used to sample $k$-bandlimited signals.
\item Second, we notice that $\Meas \ma{P}^{-1/2} \; \sig = \ma{P}_{\Omega}^{-1/2} \Meas \sig$ where $\ma{P}_{\Omega} \in \Rbb^{\nbVertRed \times \nbVertRed}$ is the diagonal matrix with entries $(\ma{P}_{\Omega})_{ii} = \prob_{\omega_i}$. Therefore, one just needs to measure $\Meas \sig$ in practice; the re-weighting by $\ma{P}_{\Omega}^{-1/2}$ can be done off-line.
\item Third, as $(\cumCoh_{\prob}^{\nbClass})^2 \geq \nbClass$, we need to sample at least $\nbClass$ nodes. Note that $\nbClass$ is also the minimum number of measurements that one must take to hope to reconstruct $\sig \in \spann(\Fou_\nbClass)$. 
\end{itemize}

The above theorem is quite similar to known compressed sensing results in bounded orthonormal systems \cite{candes07, foucart13}. The proof actually relies on the same tools as the ones used in compressed sensing. However, in our case, the setting is simpler. Unlike in compressed sensing where the signal model is a union of subspaces, the model here is a single known subspace. In the proof, we exploit this fact to refine and tighten the sampling condition. In this simpler setting and thanks to our refined result, we can propose a sampling procedure that is \emph{always} optimal in terms of the number of measurements.

In order to minimise the number of measurements, the idea is to choose a sampling distribution that minimises $\cumCoh_{\prob}^{\nbClass}$. Luckily, it occurs that it is always possible to reach the lower bound of $\cumCoh_{\prob}^{\nbClass}$ with a proper choice of the sampling distribution. The sampling distribution $\vec{p}^* \in \Rbb^\nbVert$ that minimises the graph weighted coherence is 
\begin{align}
\label{eq:optimal_distribution}
\vec{p}^*_i := \frac{\norm{\Fou_\nbClass^\adjoint \vec{\delta}_i}_2^2}{\nbClass},\quad i = 1, \ldots, \nbVert,
\end{align}
for which $(\cumCoh_{\vec{p}^*}^{\nbClass})^2 =  \nbClass$. The proof is simple. One just need to notice that $\sum_{i=1}^\nbVert \vec{p}^*_i = \nbClass^{-1} \; \sum_{i=1}^\nbVert \norm{\Fou_\nbClass^\adjoint \vec{\delta}_i}_2^2 = \nbClass^{-1} \norm{\Fou_\nbClass}_{\rm Frob} = \nbClass^{-1} \nbClass = 1$ so that $\vec{p}^*$ is a valid probability distribution. Finally, it is easy to check that $(\cumCoh_{\vec{p}^*}^{\nbClass})^2 = \nbClass$. This yields the following corollary to Theorem \ref{th:rip}.

\begin{corollary}
Let $\Meas$ be a random subsampling matrix constructed as in \refeq{eq:subsampling_matrix_def} with the sampling distribution $\vec{p}^*$ defined in \refeq{eq:optimal_distribution}. For any $\delta, \epsilon \in (0, 1)$, with probability at least $1-\epsilon$,
\begin{align*}
(1 - \delta) \norm{\sig}_2^2 \leq \inv{\nbVertRed} \norm{\Meas \ma{P}^{-1/2} \; \sig}_2^2 \leq (1 + \delta) \norm{\sig}_2^2
\end{align*}
for all $\sig \in \spann(\Fou_{\nbClass})$ provided that
\begin{align*}
\nbVertRed \geq \frac{3}{\delta^2} \; \nbClass \; \log\left( \frac{2\nbClass}{\epsilon} \right).
\end{align*}
\end{corollary}

The sampling distribution $\vec{p}^*$ is optimal in the sense that the number of measurements needed to embed the set of $\nbClass$-bandlimited signals is essentially reduced to its minimum value. Note that, unlike Theorem \ref{th:rip} where the sampling is non-adaptive, the sampling distribution is now adapted to the structure of the graph and \emph{a priori} requires the knowledge of a basis of $\spann(\Fou_k)$. We present a fast method that does not require the computation of a basis of $\spann(\Fou_k)$ to estimate $\vec{p}^*$ in Section~\ref{sec:estimation_distribution}.

\ADDED{It is important to mention that variable density sampling techniques are also popular in compressed sensing to reduce the sampling rate. We have been inspired by the works in this field to develop our sampling technique on graphs. In compressed sensing, the high efficiency of variable density sampling was first observed empirically in magnetic resonance imaging where the goal was to speed up the acquisition by reducing the amount of acquired data~\cite{lustig07}. Theoretical evidence of the efficiency of this technique then appeared in \cite{puy11, krahmer12, adcock14} as well as in~\cite{chauffert13, boyer15, bigot16} where additional measurement constraints and structured sparsity patterns are considered. There are similarities between the theoretical results existing in the compressed sensing literature and the ones presented in this work as we use similar proof techniques. However, we refine the proofs to take into account our specific signal model: bandlimited signals on graphs. One specificity of this setting is that there always exists a sampling distribution for which sampling $O(\nbClass\log(\nbClass))$ nodes is enough to capture all $\nbClass$-bandlimited signals. We recall that up to the log factor, one cannot hope to reduce the number of measurements much further. A second originality is that we can rapidly estimate this optimal distribution using fast filtering techniques on graphs (see Section~\ref{sec:estimation_distribution}).

Finally, we would like to highlight the similarity between the concept of local graph coherence and the concept of leverage scores of a matrix. Let $\ma{G}$ be a matrix and $\ma{V}$ be the left singular vectors of $\ma{G}$, 
the leverage scores related to the best rank-$r$ approximation of $\ma{G}$ are $l_i := \norm{\ma{V}_r^\adjoint \delta_i}_2^2$, where $\ma{V}_r$ contains the $r$ eigenvectors of $\ma{G}$ with largest eigenvalues 
(see, \eg,~\cite{drineas12} for more details). The only difference between the leverage scores and the local graph coherences of $\Lap$ is thus that the latter involve the smallest eigenvalues while the leverage scores involve the 
largest eigenvalues. For the normalised graph Laplacian $\Lap = \ma{I} - \ma{D}^{-1/2} \ma{W} \ma{D}^{-1/2}$, one can notice that the leverage scores of $\ma{D}^{-1/2} \ma{W} \ma{D}^{-1/2}$ correspond exactly 
to the local graph coherences of $\Lap$. In machine learning, the leverage scores have been used, \eg, to improve the performance of randomised algorithms that solve overdetermined least-square 
problems~\cite{drineas06, mahoney11, chen16} or that compute a low-rank approximation of a given matrix~\cite{drineas08, mahoney11, gittens13, sun15}. These algorithms work by building a sketch of the matrix 
of interest from a subset of its rows and/or columns. The leverage scores represent the importance of each row/columns in the dataset. The optimised sampling distribution for the rows/columns is then constructed 
by normalising the vector of leverage scores, as we do it here with the local graph coherence. Note also that fast algorithms to estimate the leverage scores have been developed in~\cite{drineas12}. In the future, 
it would be interesting to compare this method with ours, which explicitly uses the graph structure in the computations.}

\subsection{Sampling without replacement}
\label{sec:sampling_without_replacement}

We have seen that the proposed sampling procedure allows one node to be sampled multiple times. In the case of a uniform sampling distribution, we can solve the issue by considering a sampling of the nodes without replacement and still prove that the set of $k$-bandlimited signals is stably embedded. We denote by $\vec{\pi} \in \Rbb^\nbVert$ the uniform distribution on $\{1, \ldots, \nbVert \}$, $\vec{\pi}_i = 1/\nbVert$ for all $i \in \{1, \ldots, \nbVert \}$.

\begin{theorem}
\label{th:rip_bis}
Let $\Meas$ be a random subsampling matrix constructed as in \refeq{eq:subsampling_matrix_def} with $\Omega$ built by drawing $\nbVertRed$ indices $\{ \omega_1, \ldots, \omega_\nbVertRed \}$ from $\{ 1, \ldots, \nbVert \}$ uniformly at random \textbf{without} replacement. For any $\delta, \epsilon \in (0, 1)$, with probability at least $1-\epsilon$,
\begin{align*}
(1 - \delta) \norm{\sig}_2^2 \leq \frac{\nbVert}{\nbVertRed} \norm{\Meas \, \sig}_2^2 \leq (1 + \delta) \norm{\sig}_2^2
\end{align*}
for all $\sig \in \spann(\Fou_{\nbClass})$ provided that
\begin{align*}
\nbVertRed \geq \frac{3}{\delta^2} \; (\cumCoh_{\vec{\pi}}^{\nbClass})^2 \; \log\left( \frac{2\nbClass}{\epsilon} \right).
\end{align*}
\end{theorem}
\begin{proof}
See \ifels\else Appendix \fi \ref{app:proof_rip}.
\end{proof}

The attentive reader will notice that, unfortunately, the condition on $\nbVertRed$ is identical to the case where the sampling is done with replacement. This is because the theorem that we use to prove this result is obtained by ``coming back'' to sampling with replacement. Yet, we believe that it is still interesting to mention this result for applications where one wants to avoid any duplicated lines in the sampling matrix $\Meas$, which, for example, ensures that $\norm{\Meas}_2 = 1$.

In the general case of non-uniform distributions, we are unfortunately not aware of any result allowing us to handle the case of a sampling without replacement. Yet it would be interesting to study this scenario more carefully in the future as sampling without replacement seems more natural for practical applications.

\subsection{\ADDED{Intuitive links between a graph's structure and its local coherence}}
\label{sec:intuitive_coh}
\ADDED{Recall that the local graph coherence at node $i$ reads $\norm{\Fou_\nbClass^\adjoint\vec{\delta}_i}_2$. We give here some examples showing how this quantity changes for different graphs. 

Consider first the {\em $d$-dimensional grid with periodic boundary conditions}. In this case, the eigenvectors of its Laplacian are simply the $d$-dimensional classical Fourier modes. For simplicity, we suppose that $\lambda_k\neq\lambda_{k+1}$. In this case, one can show that the local coherence is independent of $i$: $\forall i\in\mathcal{V}\quad \norm{\Fou_\nbClass^\adjoint\vec{\delta}_i}_2=\sqrt{\nbClass/\nbVert}$. The optimal probability $\prob^*$ is therefore uniform for the $d$-dimensional grid with periodic boundary conditions. Without the periodic boundary conditions, the optimal probability is mostly constant with an increase when going to the boundary nodes.
An example in 1 dimension is shown in Fig.~\ref{fig:sampling_distribution} for the path graph. 

Let us now consider {\em a graph made of $k$ disconnected components} of size $\nbVert_1, \ldots, \nbVert_k$. We have $\sum_{j=1}^k \nbVert_j = \nbVert$. Considering the combinatorial graph Laplacian, one can show that a basis of $\spann{(\Fou_\nbClass)}$ is the concatenation of the indicator vectors of each component. Moreover, $\spann{(\Fou_\nbClass)}$ is the eigenspace associated to the eigenvalue $0$. The local coherence associated to node $i$ in component $j$ is $1/\sqrt{\nbVert_j}$, and the probability to choose this node reads $p_i^*=1/(\nbClass \nbVert_j)$. If all components have the same size, \ie, $\nbVert_1~=~\ldots~=~\nbVert_k$, the optimal sampling is the uniform sampling. If the components have different sizes, the smaller is a component, the larger is the probability of sampling one of its nodes. With the optimal sampling distribution, each component is sampled with probability $1/k$, no matter the size of the component. The probability that each component is sampled at least once - a necessary condition for perfect recovery - is thus higher than when using uniform sampling. One may relax this strictly disconnected component example into a loosely defined community-structured graph, where $k$ sets of nodes (forming a partition of the $n$ nodes) are more connected with themselves than with the rest of the graph. In this case, one also expects that the probability to sample a node is inversely proportional to the size of the community it belongs to.}

\section{Signal recovery}
\label{sec:reconstruction}

In the last section, we proved that it is possible to embed the space of $k$-bandlimited signals into $\Rbb^\nbVertRed$ using a sparse matrix $\Meas \in \Rbb^{\nbVertRed \times \nbVert}$. We now have to design a procedure to estimate accurately any $\sig \in \spann(\Fou_\nbClass)$ from its, possibly \emph{noisy}, $\nbVertRed$ samples. Let us consider that the samples $\meas \in \Rbb^\nbVertRed$ satisfy
\begin{align*}
\meas = \Meas \sig + \err,
\end{align*}
where $\err \in \Rbb^\nbVertRed$ models a noise. Note that $\err$ can be \emph{any} vector $\Rbb^\nbVertRed$. We do not restrict our study to a particular noise structure. The vector $\err$ can be used to represent, \eg, errors 
relative to the signal model or correlated noise. \ADDED{Such a recovery problem of a graph signal given few measured nodes and under a smoothness model is reminiscent to the litterature on 
semi-supervised learning. We link our approach to the appropriate literature in Section~\ref{subsec:eff_decoder}.}

\subsection{Standard decoder}

In a situation where one knows a basis of $\spann{(\Fou_\nbClass)}$, the standard method to estimate $\sig$ from $\meas$ is to compute the best approximation to $\meas$ from $\spann(\Fou_\nbClass)$, \ie, to solve
\begin{align}
\label{eq:optimal_decoder}
\min_{\vec{z} \in \spann(\Fou_\nbClass)} \norm{ \ma{P}^{-1/2}_\Omega \left(\Meas \vec{z} - \meas \right)}_2.
\end{align}

Note that we introduced a weighting by the matrix $\ma{P}^{-1/2}_\Omega$ in \refeq{eq:optimal_decoder} to account for the fact that $\nbVertRed^{-1} \, \ma{P}^{-1/2}_\Omega \Meas = \nbVertRed^{-1} \, \Meas \ma{P}^{-1/2}$ satisfies the RIP, not $\Meas$ alone. The following theorem proves that the solution of \refeq{eq:optimal_decoder} is a faithful estimation of $\sig$.

\begin{theorem}
\label{th:ideal_decoder}
Let $\Omega$ be a set of $\nbVertRed$ indices selected independently from $\{1, \ldots, \nbVert \}$ using a sampling distribution $\prob \in \Rbb^\nbVert$, and $\Meas$ be the sampling matrix associated to $\Omega$ (see \refeq{eq:subsampling_matrix_def}). Let $\epsilon, \delta \in (0, 1)$ and suppose that $\nbVertRed$ satisfies \refeq{eq:sampling_condition}. With probability at least $1-\epsilon$, the following holds for all $\sig \in \spann(\Fou_\nbClass)$ and all $\err \in \Rbb^\nbVertRed$.
\begin{enumerate}[i)]
\item Let $\sig^*$ be the solution of Problem \refeq{eq:optimal_decoder} with $\meas = \Meas \sig + \err$. Then,
\begin{align}
\label{eq:error_reconstruction_optimal}
\norm{\sig^* -  \sig}_2 \leq \frac{2}{\sqrt{\nbVertRed \, (1 - \delta)}} \norm{\ma{P}^{-1/2}_\Omega \err}_2.
\end{align}
\item There exist particular vectors $\err_0 \in \Rbb^\nbVertRed$ such that the solution $\sig^*$ of Problem \refeq{eq:optimal_decoder} with $\meas = \Meas \sig + \err_0$ satisfies
\begin{align}
\label{eq:lower_error_reconstruction_optimal}
\norm{\sig^* -  \sig}_2 \geq \frac{1}{\sqrt{\nbVertRed \, (1 + \delta)}} \norm{\ma{P}^{-1/2}_\Omega \err_0}_2.
\end{align}
\end{enumerate}
\end{theorem}
\begin{proof}
See \ifels\else Appendix \fi \ref{app:proof_decoder}.
\end{proof}

We notice that in the absence of noise $\sig^* = \sig$, as desired. In the presence of noise, the upper bound on the error between $\sig^*$ and $\sig$ increases linearly with $\| \ma{P}^{-1/2}_\Omega \err \|_2$. For a uniform sampling, we have $\| \ma{P}^{-1/2}_\Omega \err \|_2 = \sqrt{n} \; \| \err \|_2$. For a non-uniform sampling, we may have $\| \ma{P}^{-1/2}_\Omega \err \|_2 \gg \sqrt{n} \; \norm{\err}_2$ for some \emph{particular} draws of $\Omega$ and noise vectors $\err$. Indeed, some weights $\prob_{\omega_i}$ might be arbitrarily close to $0$. Unfortunately, one cannot in general improve the upper bound in \refeq{eq:error_reconstruction_optimal} as proved by the second part of the theorem with \refeq{eq:lower_error_reconstruction_optimal}. Non-uniform sampling can thus be very sensitive to noise unlike uniform sampling. However, this is a worst case scenario. First, it is unlikely to draw an index $\omega_i$ where $\prob_{\omega_i}$ is small by construction of the sampling procedure. Second, 
\begin{align*}
\Ebb \, \norm{ \ma{P}^{-1/2}_\Omega \err }_2^2 = \nbVert \norm{\err}_2^2,
\end{align*}
so that $\| \ma{P}^{-1/2}_\Omega \err \|_2$ is not too large on average over the draw of $\Omega$. Furthermore, in our numerical experiments, we noticed that we have $\min_i \prob_i = 1/(\alpha^2 \, \nbVert)$, where $\alpha>1$ is a small constant\footnote{In the numerical experiments presented below, we have $\alpha$ smaller or equal to $3$ in all cases tested  with the optimal sampling distribution for the graphs presented in Fig.~\ref{fig:graphs}.}, for the optimal sampling distributions $\prob = \prob^*$ obtained in practice. This yields $\| \ma{P}^{-1/2}_\Omega \err \|_2 \leq \alpha \, \sqrt{n} \, \norm{\err}_2$, which shows that non-uniform sampling is just slightly more sensitive to noise than uniform sampling in practical settings, with the advantage of reducing the number of measurements. Non-uniform sampling is thus still a beneficial solution.

We have seen a first method to estimate $\sig$ from its measurements. This method has however a major drawback: it requires the estimation of a basis of $\Fou_\nbClass$, which can be computationally very expensive for large graphs. To overcome this issue, we propose an alternative decoder which is computationally much more efficient. This algorithm uses techniques developed to filter graph signals rapidly. We thus briefly recall the principle of these filtering techniques.

\subsection{Fast filtering on graphs}

A filter is represented by a function $h: \Rbb \rightarrow \Rbb$ in the Fourier (spectral) domain. The signal $\sig$ filtered by $h$ is
\begin{align*}
\sig_{h} := \Fou \, \diag(\hat{\vec{h}}) \, \Fou^\adjoint \sig \in \Rbb^\nbVert,
\end{align*}
where $\hat{\vec{h}} = (h(\eig_1), \ldots, h(\eig_\nbVert))^\adjoint \in \Rbb^{\nbVert}$. Filtering thus consists in multiplying point by point the Fourier transform of $\sig$ with $\hat{\vec{h}}$ and then computing the inverse Fourier transform of the resulting signal.

According to the above definition, filtering \emph{a priori} requires the knowledge of the matrix $\Fou$. To avoid the computation of $\Fou$, one can approximate the function $h$ by a polynomial 
\begin{align*}
\kappa(t) = \sum_{i=0}^d \alpha_i \, t^i
\end{align*}
of degree $d$ and compute $\sig_{\kappa}$, which will approximate $\sig_{h}$. This computation can be done rapidly as it only requires matrix-vector multiplications with $\Lap$, which is sparse in most applications. Indeed,
\begin{align*}
\sig_{\kappa} = \Fou \, \diag(\hat{\vec{\kappa}}) \, \Fou^\adjoint \sig = \sum_{i=0}^d \alpha_i \; \Fou \, \diag(\eig_1^i, \ldots, \eig_\nbVert^i) \, \Fou^\adjoint \sig = \sum_{i=0}^d \alpha_i \; \Lap^i \sig.
\end{align*}
%
We let the reader refer to \cite{hammond11} for more information on this fast filtering technique. 

To simplify notations, for any polynomial function $\kappa(t) = \sum_{i=0}^d \; \alpha_i \, t^i$ and any matrix $\ma{A} \in \Rbb^{\nbVert \times \nbVert}$, we define
\begin{align}
\label{eq:poly_matrix}
\kappa(\ma{A}) := \sum_{i=0}^d \alpha_i \; \ma{A}^i .
\end{align}
Remark that $\kappa(\Lap) = \Fou \, \kappa(\Eig) \, \Fou^\adjoint$.

\subsection{Efficient decoder}
\label{subsec:eff_decoder}

Instead of solving \refeq{eq:optimal_decoder}, we propose to estimate $\sig$ by solving the following problem
\begin{align}
\label{eq:practical_decoder}
\min_{\vec{z} \in \Rbb^\nbVert} \norm{ \ma{P}^{-1/2}_\Omega \left(\Meas \vec{z} - \meas \right)}_2^2 + \reg \; \vec{z}^\adjoint g(\Lap) \vec{z},
\end{align}
where $\gamma>0$ and $g: \Rbb \rightarrow \Rbb$ is a nonnegative and nondecreasing polynomial function. These assumptions on $g$ imply that $g(\Lap)$ is positive semi-definite - hence \refeq{eq:practical_decoder} is convex - and that $0 \leq g(\eig_1) \leq \ldots \leq g(\eig_\nbVert)$.

The intuition behind decoder~\eqref{eq:practical_decoder} is quite simple. Consider, for simplicity, that $g$ is the identity. The regularisation term becomes $\vec{z}^\adjoint \Lap \vec{z}$. Remember that a $\nbClass$-bandlimited signal is a signal that lives in the span of the first $\nbClass$ eigenvector of $\Fou$, \ie, where the eigenvalues of $\Lap$ are the smallest. The regularisation term satisfies $\vec{z}^\adjoint \Lap \vec{z} = (\vec{z}^\adjoint \Fou) \Eig (\Fou^\adjoint\vec{z})$, where $\Eig$ is the diagonal matrix containing the eigenvalues of $\Lap$. Therefore, this term penalises signals with energy concentrated at high frequencies more than signals with energy concentrated at low frequencies. In other words, this regularisation term favours the reconstruction of low-frequency signals, \ie, signals approximately bandlimited. Notice also that one can recover the standard decoder defined in \refeq{eq:optimal_decoder} by substituting the function $i_{\lambda_\nbClass} : \Rbb \rightarrow \Rbb \cup \{ + \infty \}$, defined as
\begin{align*}
i_{\lambda_\nbClass}(t) := 
\left\{
\begin{array}{ll}
0 & \text{if } t \in [0, \lambda_\nbClass],\\
+\infty & \text{otherwise},
\end{array}
\right.
\end{align*}
for $g$ in \refeq{eq:practical_decoder}.

We argue that solving \refeq{eq:practical_decoder} is computationally efficient because $\Lap$ is sparse in most applications. Therefore, any method solving \refeq{eq:practical_decoder} that requires only matrix-vector multiplications with $g(\Lap)$ can be implemented efficiently, as it requires multiplications with $\Lap$ only (recall the definition of $g(\Lap)$ in \refeq{eq:poly_matrix}). Examples of such methods are the conjugate gradient method or any gradient descend methods. Let us recall that one can find a solution to \refeq{eq:practical_decoder} by solving
\begin{align}
\label{eq:exact_solution}
\left(\Meas^\adjoint \ma{P}_\Omega^{-1} \Meas + \reg \, g(\Lap) \right) \vec{z} = \Meas^\adjoint \ma{P}_\Omega^{-1} \vec{y}.
\end{align}

The next theorem bounds the error between the original signal $\sig$ and the solution of \refeq{eq:practical_decoder}.

\begin{theorem}
\label{th:practical_decoder}
Let $\Omega$ be a set of $\nbVertRed$ indices selected independently from $\{1, \ldots, \nbVert \}$ using a sampling distribution $\prob \in \Rbb^\nbVert$, $\Meas$ be the sampling matrix associated to $\Omega$ (see \refeq{eq:subsampling_matrix_def}), and $M_{\rm max}>0$ be a constant such that $\norm{\Meas \ma{P}^{-1/2}}_2 \leq M_{\rm max}$. Let $\epsilon, \delta \in (0, 1)$ and suppose that $\nbVertRed$ satisfies \refeq{eq:sampling_condition}. With probability at least $1-\epsilon$, the following holds for all $\sig \in \spann(\Fou_\nbClass)$, all $\err \in \Rbb^{\nbVert}$, all $\reg > 0$, and all nonnegative and nondecreasing polynomial functions $g$ such that $g(\eig_{\nbClass+1}) > 0$.

Let $\sig^*$ be the solution of \refeq{eq:optimal_decoder} with $\meas = \Meas \sig + \err$. Then,
\begin{align}
\label{eq:bound_alpha}
\norm{ \vec{\alpha}^* - \sig }_2
& \; \leq \;
\inv{\sqrt{\nbVertRed (1-\delta)}} 
\left[\left( 2 +  \frac{M_{\rm max}}{\sqrt{\reg g(\eig_{\nbClass+1})}}\right)\norm{\ma{P}^{-1/2}_\Omega \err}_2\right. \nonumber \\
& \hspace{4cm} \left. + \left( M_{\rm max} \sqrt{\frac{g(\eig_\nbClass)}{g(\eig_{\nbClass+1})}} + \sqrt{\reg g(\eig_\nbClass)}\right) \norm{ \sig  }_2 \right],
\end{align}
and
\begin{align}
\label{eq:bound_beta}
\norm{ \vec{\beta}^*}_2
\; \leq \;
 \frac{1}{\sqrt{\reg g(\eig_{\nbClass+1})}} \norm{ \ma{P}^{-1/2}_\Omega \err }_2 +  \sqrt{\frac{g(\eig_\nbClass)}{g(\eig_{\nbClass+1})}} \norm{ \sig  }_2,
\end{align}
where $\vec{\alpha}^* := \Fou_\nbClass \Fou_\nbClass^\adjoint \, \sig^*$ and $\vec{\beta}^* := (\ma{I} - \Fou_\nbClass \Fou_\nbClass^\adjoint) \, \sig^*$.
\end{theorem}
\begin{proof}
See \ifels\else Appendix \fi \ref{app:proof_decoder}.
\end{proof}

In the above theorem, $\vec{\alpha}^*$ is the orthogonal projection of $\sig^*$ onto $\spann(\Fou_\nbClass)$ and $\vec{\beta}^*$ onto the orthogonal complement of $\spann(\Fou_\nbClass)$. To obtain a bound on $\norm{ \sig^* - \sig }_2$, one can simply use the triangle inequality and the bounds \refeq{eq:bound_alpha} and \refeq{eq:bound_beta}. 

In the absence of noise, we thus have
\begin{align*}
\norm{ \sig^* - \sig }_2
\; \leq \;
\inv{\sqrt{\nbVertRed (1-\delta)}} 
\left( M_{\rm max} \sqrt{\frac{g(\eig_\nbClass)}{g(\eig_{\nbClass+1})}} + \sqrt{\reg g(\eig_\nbClass)}\right) \norm{ \sig  }_2 +  \sqrt{\frac{g(\eig_\nbClass)}{g(\eig_{\nbClass+1})}} \norm{ \sig  }_2.
\end{align*}
If $g(\eig_\nbClass) = 0$, we notice that we obtain a perfect reconstruction. Note that as $g$ is supposed to be nondecreasing and nonnegative, $g(\eig_\nbClass) = 0$ implies that we also have $g(\eig_1) = ... = g(\eig_{\nbClass-1}) = 0$. If $g(\eig_\nbClass) \neq 0$, the above bound shows that we should choose $\reg$ as close as possible to\footnote{Notice that if $\meas$ is in the range of $\Meas$, then the solution of Problem~\refeq{eq:practical_decoder} tends to the solution of $\min_{\vec{z}} \vec{z}^\adjoint g(\Lap) \vec{z} \quad \text{s.t.}\quad \meas = \Meas \vec{z}$ in the limit where $\reg \rightarrow 0^+$.} $0$ and seek to minimise the ratio $g(\eig_\nbClass)/g(\eig_{\nbClass+1})$ to minimise the upper bound on the reconstruction error. Notice that if $g(\Lap) = \Lap^l$, with $l \in \Nbb^*$, then the ratio $g(\eig_\nbClass)/g(\eig_{\nbClass+1})$ decreases as $l$ increases. Increasing the power of $\Lap$ and taking $\reg$ sufficiently small to compensate the potential growth of $g(\eig_\nbClass)$ is thus a simple solution to improve the reconstruction quality in the absence of noise.

In the presence of noise, for a fixed function $g$, the upper bound on the reconstruction error is minimised for a value of $\reg$ proportional to $\| \ma{P}^{-1/2}_\Omega \err \|_2 /\norm{ \sig  }_2$. To optimise the result further, one should seek to have $g(\eig_{\nbClass})$ as small as possible and $g(\eig_{\nbClass+1})$ as large as possible.

\ADDED{Decoder~\eqref{eq:practical_decoder} has close links to several ``decoders'' used in the semi-supervised learning litterature~\cite{Chapelle_SSL_book} that attempt to estimate the label of unlabeled data from a small number of labeled data, by supposing that the label functions are smooth either  1) in the data space or  2) in a suitable transformed space --using similarity kernels that define graphs modeling the underlying manifold for instance--~\cite{zhu_ICML2003,Zhou_NIPS2004}, or 3) in both~\cite{Belkin_JMLR2006,bengio_chapter2006}. In our work, smoothness is defined solely using the graph (case 2) which we suppose given; there is no equivalent of a data space (case 1) on which to define another smoothness constraint. Nevertheless, other types of smoothness could be considered instead of the Laplacian smoothness $\vec{z}^\adjoint g(\Lap) \vec{z}$. For instance, one could decide to use an $l_1$ penalisation of the 
graph difference operator, as in~\cite{trend_filtering_AISTATS2015}, to allow the signal to depart from the smoothness prior at some nodes. The closest semi-supervised framework to our naive decoder~\eqref{eq:optimal_decoder} is found in~\cite{Belkin_NIPS2003}, where the authors constrain the solution to be in $\spann{(\Fou_k)}$ without specifying precisely the value of $k$; and the closest technique to our efficient decoder~\eqref{eq:practical_decoder} is found in~\cite{Zhou_NIPS2004} even though their cost function has an additional term of the form $\sum_{i\not\in\Omega} \vec{z}_i^2$ compared to ours. Another decoding method may be found in~\cite{Fergus_NIPS2009}, where the authors have a similar cost function to ours. However, they work in the data space (case 1 above) and try to optimise this cost function directly in this space, \ie, without explicitly constructing and using a graph.

Even though we use similar decoders than in the semi-supervised learning literature, let us stress that an important difference is that we choose beforehand which nodes to sample/label. In this sense, our work also has connections with the literature in active learning~\cite{settles_active_learning}, more precisely with the works that concentrate on the offline (all nodes to label are chosen from the start), single-batch (the nodes to sample are drawn simultaneously) selection problem~\cite{Fu2012}, such as in~\cite{Gu_ICDM2012,bilgic_AAAI2010,ji_AISTATS2012,ma_NIPS2013}. Yet another connection may be found in the area of Gaussian Random Fields~\cite{zhang2015graph,gadde_ICASSP2015}. The originality of our method compared to these works comes from our specific smoothness model ($\spann{(\Fou_\nbClass)}$) for which we devise an original sampling scenario which ensures stable reconstruction when coupled with the decoder~\eqref{eq:practical_decoder} (see Theorem \ref{th:practical_decoder}).}

\section{Estimation of the optimal sampling distribution}
\label{sec:estimation_distribution}

In this section, we explain how to estimate the optimal sampling distribution $\vec{p}^*$ efficiently. This distribution is entirely defined by the values 
$\norm{\Fou_\nbClass^\adjoint \vec{\delta}_i}_2^2$, $i = 1, \ldots, \nbVert$ (see \refeq{eq:optimal_distribution}). In order to be able to deal with large graphs and potentially large $\nbClass$, we want to avoid 
the computation of a basis of $\spann(\Fou_\nbClass)$ to estimate this distribution. Instead, we take another route that consists in filtering a small number of random signals. Note that the idea of filtering few random signals to estimate the number of eigenvalues of a Hermitian matrix in a given interval is already proposed and studied in \cite{napoli13}. We show here that this technique can be used to estimate $\vec{p}^*$.

For this estimation, we will need to use low-pass filters. For any $\lambda>0$, the filter ideal low-pass filter $b_\lambda : \Rbb \rightarrow \Rbb$ with cut-off frequency $\lambda$ satisfies 
\begin{align*}
b_\lambda(t) = 
\left\{
\begin{array}{ll}
1 & \text{if } t \in [0, \lambda],\\
0 & \text{otherwise}.
\end{array}
\right.
\end{align*}
%

\subsection{Principle of the estimation}

We recall that our goal is to estimate $\| \Fou_\nbClass^\adjoint \vec{\delta}_i \|_2^2$ for all $i \in \{ 1, \ldots, \nbVert \}$. To understand how our method works, consider that $\eig_\nbClass$ is known for the moment. Let $\vec{r} \in \Rbb^\nbVert$ be a vector with independent random entries that follow a standard normal distribution. By filtering $\vec{r}$ with $b_{\eig_\nbClass}$, we obtain
\begin{align*}
\vec{r}_{b_{\eig_\nbClass}} 
\; = \;
\Fou \; \diag(\eig_1, \ldots, \eig_\nbClass, 0, \ldots, 0) \; \Fou^\adjoint \; \vec{r}
\; = \;
\Fou_\nbClass \Fou_\nbClass^\adjoint \; \vec{r}.
\end{align*}
The estimation of the optimal sampling distribution is based on the following property. The $i^\th$ entry of $\vec{r}_{b_{\eig_\nbClass}}$ is
\begin{align*}
(\vec{r}_{b_{\eig_\nbClass}})_i 
\; = \;
\vec{r}_{b_{\eig_\nbClass}}^\adjoint \vec{\delta}_i
\; = \;
\vec{r}^\adjoint \Fou_\nbClass \Fou_\nbClass^\adjoint \vec{\delta}_i,
\end{align*}
and the mean of $(\vec{r}_{b_{\eig_\nbClass}})_i^2$ satisfies
\begin{align*}
\Ebb \, (\vec{r}_{b_{\eig_\nbClass}})_i^2
\; = \;
\vec{\delta}_i^\adjoint \Fou_\nbClass \Fou_\nbClass^\adjoint \; \Ebb (\vec{r} \vec{r}^\adjoint) \; \Fou_\nbClass \Fou_\nbClass^\adjoint \vec{\delta}_i
\; = \;
\vec{\delta}_i^\adjoint \Fou_\nbClass \Fou_\nbClass^\adjoint \Fou_\nbClass \Fou_\nbClass^\adjoint \vec{\delta}_i
\; = \;
\vec{\delta}_i^\adjoint \Fou_\nbClass \Fou_\nbClass^\adjoint \vec{\delta}_i
\; = \;
\norm{\Fou_\nbClass^\adjoint \vec{\delta}_i}_2^2.
\end{align*}
This shows that $(\vec{r}_{b_{\eig_\nbClass}})_i^2$ is an unbiased estimation of $\| \Fou_\nbClass^\adjoint \vec{\delta}_i \|_2^2$, the quantity we want to evaluate. Therefore, a possibility to estimate the optimal sampling distribution consists in filtering $L$ random signals $\vec{r}^{1}, \ldots, \vec{r}^{L}$ with the same distribution as $\vec{r}$ and average $(\vec{r}_{b_{\eig_\nbClass}}^1)_i^2, \ldots, (\vec{r}_{b_{\eig_\nbClass}}^L)_i^2$ for each $i \in \{ 1, \ldots, \nbVert \}$. The next theorem shows that if $\eig_k$ is known, then $L \geq O(\log(\nbVert))$ random vectors are sufficient to have an accurate estimation of $\| \Fou_\nbClass^\adjoint \vec{\delta}_i \|_2^2$.

In the theorem below, we consider a realistic scenario where we filter the signals with a polynomial approximation of $b_{\lambda}$. This theorem shows how this approximation affects the estimation of $\| \Fou_\nbClass^\adjoint \vec{\delta}_i \|_2^2$. We denote the polynomial filter approximating $b_{\lambda}$ by $c_{\lambda}: \Rbb \rightarrow \Rbb$. It satisfies
\begin{align}
\label{eq:def_c_lambda}
c_{\lambda} = b_{\lambda} + \hat{e}_{\lambda},
\end{align}
where $\hat{e}_{\lambda}: \Rbb \rightarrow \Rbb$ models the approximation error. We define 
\begin{align*}
\ma{E}_{\lambda} := \diag(\hat{e}_{\lambda}(\eig_1), \ldots, \hat{e}_{\lambda}(\eig_\nbVert)) \in \Rbb^{\nbVert \times \nbVert}.
\end{align*}
\begin{theorem}
\label{th:jl}
Let $\vec{r}^{1}, \ldots, \vec{r}^{L} \in \Rbb^\nbVert$ be $L$ independent zero-mean Gaussian random vectors with covariance $L^{-1} \, \ma{I}$. Denote by $\vec{r}_{c_{\lambda}}^1, \ldots, \vec{r}_{c_{\lambda}}^L \in \Rbb^\nbVert$ the signals $\vec{r}^{1}, \ldots, \vec{r}^{L}$ filtered by $c_{\lambda}$ with $\lambda>0$. Let $j^*$ be the largest integer such that $\eig_{j^*} \leq \lambda$. There exists an absolute constant $C>0$ such for any $\epsilon, \delta \in (0, 1)$, with probability at least $1-\epsilon$, the filtered signals  satisfy
\begin{align*}
(1-\delta) \, \abs{ \norm{\Fou_{j^*}^\adjoint \vec{\delta}_i }_2 - \norm{\ma{E}_{\lambda} \Fou^\adjoint \vec{\delta}_i }_2 }^2
\; \leq \;
\sum_{l=1}^L \; (\vec{r}_{c_\lambda}^l)_i^2
\; \leq \;
(1+\delta) \, \abs{ \norm{\Fou_{j^*}^\adjoint \vec{\delta}_i }_2 + \norm{\ma{E}_{\lambda} \Fou^\adjoint \vec{\delta}_i }_2 }^2,
\end{align*}
for all $i \in \{1, \ldots, \nbVert \}$, provided that
\begin{align*}
L \; \geq \; \frac{C}{\delta^2} \log \left( \frac{2\nbVert}{\epsilon} \right).
\end{align*}
\end{theorem}
\begin{proof}
See \ifels\else Appendix \fi \ref{app:proof_jl}.
\end{proof}

The above theorem indicates that if $\lambda \in [\eig_\nbClass, \eig_{\nbClass+1})$ and $\hat{e}$ is null, then $\sum_{l=1}^L \; (\vec{r}_{c_{\lambda_k}}^l)_i^2$ estimates $\norm{\Fou_\nbClass^\adjoint \vec{\delta}_i }_2^2$ with an error at most $\delta$ on each entry $i \in \{1, \ldots, \nbVert \}$. Recalling that the optimal sampling distribution has entries
\begin{align*}
\vec{p}^*_i = \frac{\norm{\Fou_\nbClass^\adjoint \vec{\delta}_i }_2^2}{k} = \frac{\norm{\Fou_\nbClass^\adjoint \vec{\delta}_i }_2^2}{\sum_{i=1}^\nbVert \norm{\Fou_\nbClass^\adjoint \vec{\delta}_i }_2^2},
\end{align*}
we see that $\tilde{\vec{p}} \in \Rbb^\nbVert$ with entries
\begin{align*}
\tilde{\vec{p}}_i := \frac{\sum_{l=1}^L \; (\vec{r}_{c_{\lambda_k}}^l)_i^2}{\sum_{i=1}^\nbVert\sum_{l=1}^L \; (\vec{r}_{c_{\lambda_k}}^l)_i^2}
\end{align*}
approximates the optimal sampling distribution. If we know $\eig_{\nbClass}$ and $\eig_{\nbClass+1}$, we can thus approximate $\vec{p}^*$. In order to complete the method, we now need a solution to estimate $\eig_j$ with $j=k$ or $j = k+1$.

\subsection{Estimating $\lambda_\nbClass$ and $\lambda_{\nbClass+1}$}

Let $\lambda \in (0, \eig_\nbVert)$. Theorem \ref{th:jl} shows that, with probability $1-\epsilon$,
\begin{align*}
(1-\delta) \, \sum_{i=1}^\nbVert \norm{\Fou_{j^*}^\adjoint \vec{\delta}_i }_2^2
\; \leq \;
\sum_{i=1}^\nbVert \sum_{l=1}^L \; (\vec{r}_{b_\lambda}^l)_i^2
\; \leq \;
(1+\delta) \, \sum_{i=1}^\nbVert \norm{\Fou_{j^*}^\adjoint \vec{\delta}_i }_2^2,
\end{align*}
when using the filter $b_\lambda$. Noticing that
\begin{align*}
\sum_{i=1}^\nbVert \norm{\Fou_{j^*}^\adjoint \vec{\delta}_i }_2^2 = \norm{\Fou_{j^*}}_{\rm Frob}^2 = j^*,
\end{align*}
as the columns of $\Fou$ are normalised, yields
\begin{align*}
(1-\delta) \; j^*
\; \leq \;
\sum_{i=1}^\nbVert \sum_{l=1}^L \; (\vec{r}_{b_\lambda}^l)_i^2
\; \leq \;
(1+\delta) \; j^*.
\end{align*}
In other words, the total energy of the filtered signals is tightly concentrated around $j^*$, which is the largest integer such that $\eig_{j^*} \leq \lambda$. Therefore, the total energy of the filtered signals provides an estimation of the number of eigenvalues of $\Lap$ that are below $\lambda$. 

Using this phenomenon, one can obtain, by dichotomy, an interval $(\underline{\lambda}, \bar{\lambda})$ such that $\nbClass-1$ eigenvalues are below $\underline{\lambda}$ and $\nbClass$ eigenvalues are below $\bar{\lambda}$ and thus obtain an estimation of $\eig_\nbClass$. The same procedure can be used to estimate $\eig_{\nbClass+1}$. Note that we cannot filter the signals using an ideal low-pass filter in practice, so that an additional error will slightly perturb the estimation.

\subsection{The complete algorithm}

We now have all the tools to design an algorithm that estimates the optimal sampling distribution. This  is summarised in Algorithm \ref{alg:optimal_distribution}. In practice, we noticed that using $L = 2 \, \log(\nbVert)$ signals is usually enough to obtain a reasonable approximation of the sampling distribution. We also only estimate $\eig_{\nbClass}$ and do not estimate $\eig_{\nbClass+1}$. Steps \ref{step:start} to \ref{step:end} of the algorithm concern the estimation of $\eig_\nbClass$ by dichotomy. The estimated optimal sampling distribution $\tilde{\vec{p}} \in \Rbb^\nbVert$ is defined in Step \ref{step:result}. Finally, we would like to mention that a better estimation of $\eig_k$ and of the sampling distribution could be obtained by running multiple times Algorithm \ref{alg:optimal_distribution} and averaging the results. In the following experiments, this algorithm is run only once but already yields good results.

\begin{algorithm}
\caption{Estimation of the optimal sampling distribution}
\label{alg:optimal_distribution}
\begin{algorithmic}[1]
\Input Precision parameter $\varepsilon \in (0, 1)$, and bandlimit $\nbClass$.
\State Set $L = 2 \, \log(\nbVert)$ and draw $L$ random vectors $\vec{r}^{1}, \ldots, \vec{r}^{L} \in \Rbb^\nbVert$ as in Theorem \ref{th:jl}.
\State Estimate $\eig_\nbVert$ and set $\underline{\lambda} = 0$, $\bar{\lambda} = \eig_\nbVert$, $\lambda = \eig_\nbVert/2$, and compute $c_\lambda$ that approximates the ideal low-pass filter $b_\lambda$.
\While{${\rm round} \left(\sum_{i=1}^\nbVert \sum_{l=1}^L \; (\vec{r}_{c_\lambda}^l)_i^2 \right) \neq \nbClass$ or $\abs{\underline{\lambda}-\bar{\lambda}}>\varepsilon \cdot \bar{\lambda}$}
\label{step:start}
	\If {${\rm round} \left(\sum_{i=1}^\nbVert \sum_{l=1}^L \; (\vec{r}_{c_\lambda}^l)_i^2 \right) \geq \nbClass$}
		\State Set $\bar{\lambda} = \lambda$.
	\Else 
		\State Set $\underline{\lambda} = \lambda$.
\EndIf
\State Set $\lambda = (\underline{\lambda}+\bar{\lambda})/2$, and compute $c_\lambda$ that approximates the ideal low-pass filter $b_\lambda$.
\EndWhile \label{step:end}
\Output Set $\tilde{\vec{p}}_i = \left(\sum_{l=1}^L \; (\vec{r}_{c_\lambda}^l)_i^2 \right)/\left(\sum_{i=1}^\nbVert \sum_{l=1}^L \; (\vec{r}_{c_\lambda}^l)_i^2 \right)$. \label{step:result}
\end{algorithmic}
\end{algorithm}
%

\section{Experiments}
\label{sec:experiments}

In this section, we run several experiments to illustrate the above theoretical findings. First we show how the sampling distribution affects the number of measurements required to ensure that the RIP holds. Then, we show how the reconstruction quality is affected with the choice of $g$ and $\reg$ in \refeq{eq:practical_decoder}.

All our experiments are done using \MODIF{five} different types of graph, all available in the GSP toolbox \cite{perraudin14} and presented in Fig.~\ref{fig:graphs}. We use a) different community-type graphs of size $\nbVert = 1000$, b) the graph representing the Minnesota road network of size $\nbVert = 2642$, c) the graph of the Stanford bunny of size $\nbVert = 2503$, \ADDED{d) the unweighted path graph of size $\nbVert = 1000$ and e) a binary tree of depth $9$ and size $\nbVert = 1023$. We recall that each node in the unweighted path graph is connected to its left and right neighbours with weight $1$, except for the two nodes at the boundaries which have only one neighbour.} We use the combinatorial Laplacian in all experiments. All samplings are done in the conditions of Theorem~\ref{th:rip}, \ie, with replacement. Finally, the reconstructions are obtained by solving~\refeq{eq:exact_solution} using the {\sf mldivide} function of Matlab. For the graphs and functions $g$ considered, we noticed that it was faster to use this function than solving~\refeq{eq:exact_solution} by conjugate gradient.

\begin{figure}
\centering
\begin{minipage}{.19\linewidth} \centering \small Community graph \end{minipage}
\begin{minipage}{.19\linewidth} \centering \small Minnesota graph \end{minipage}
\begin{minipage}{.19\linewidth} \centering \small Bunny graph \end{minipage}
\begin{minipage}{.19\linewidth} \centering \small Path graph \end{minipage}
\begin{minipage}{.19\linewidth} \centering \small Binary tree \end{minipage}\\
\includegraphics[width=.19\linewidth, keepaspectratio]{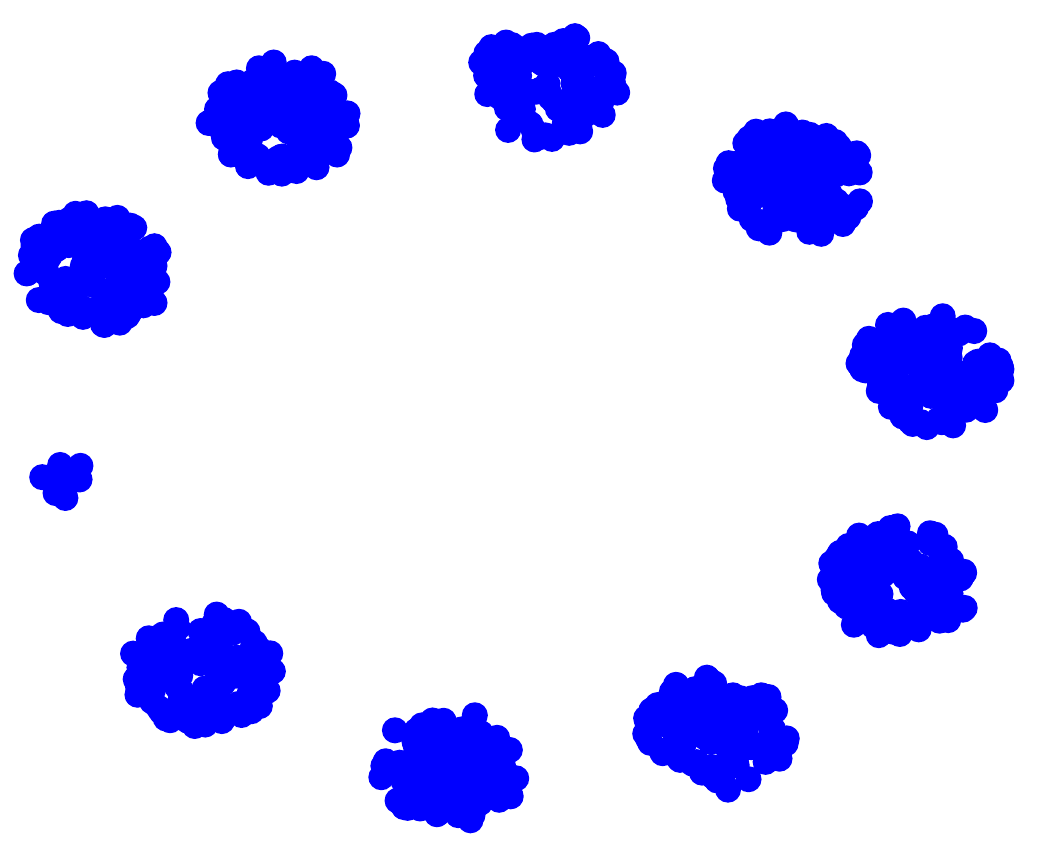}
\includegraphics[width=.19\linewidth, keepaspectratio]{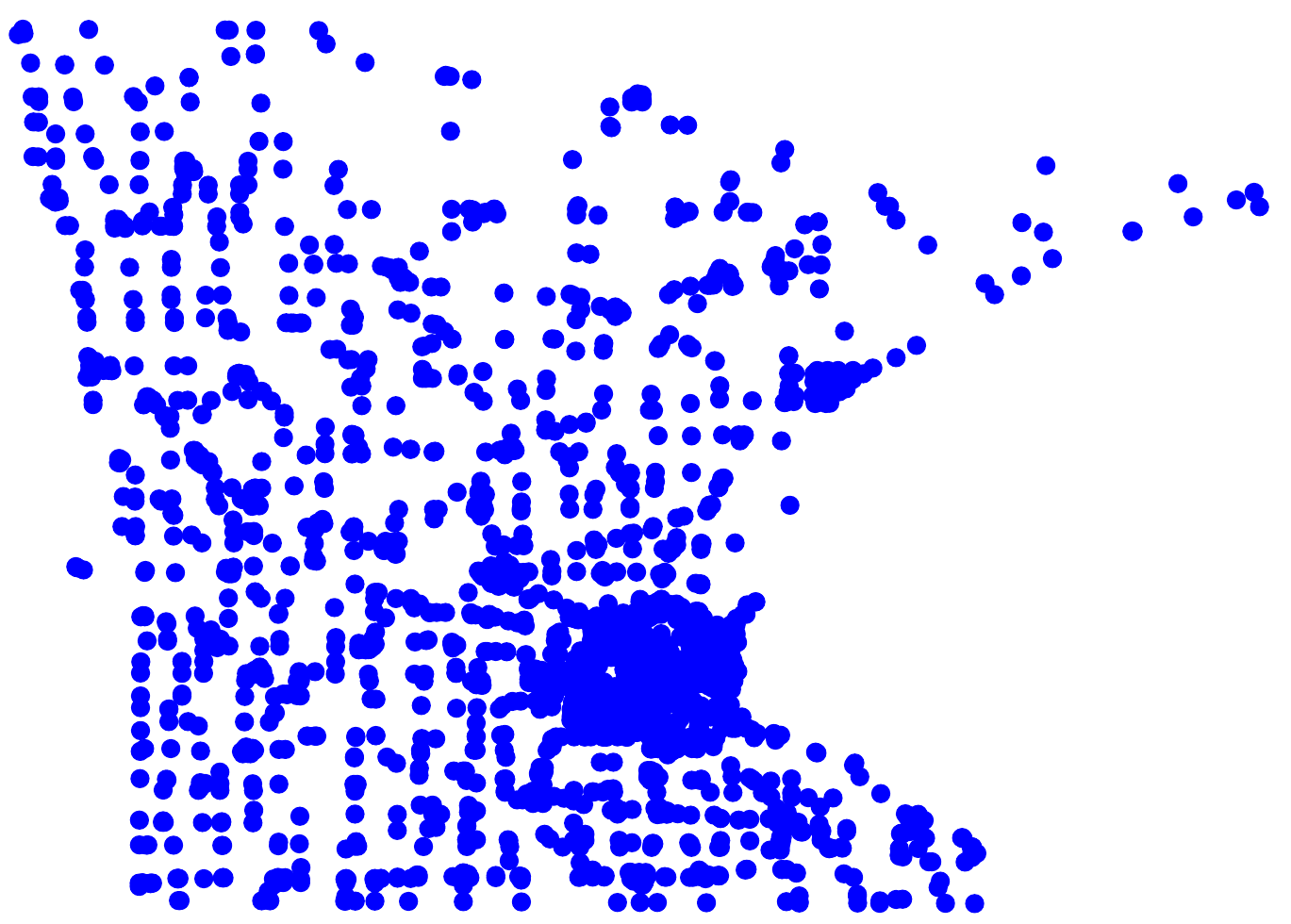}
\includegraphics[width=.19\linewidth, keepaspectratio]{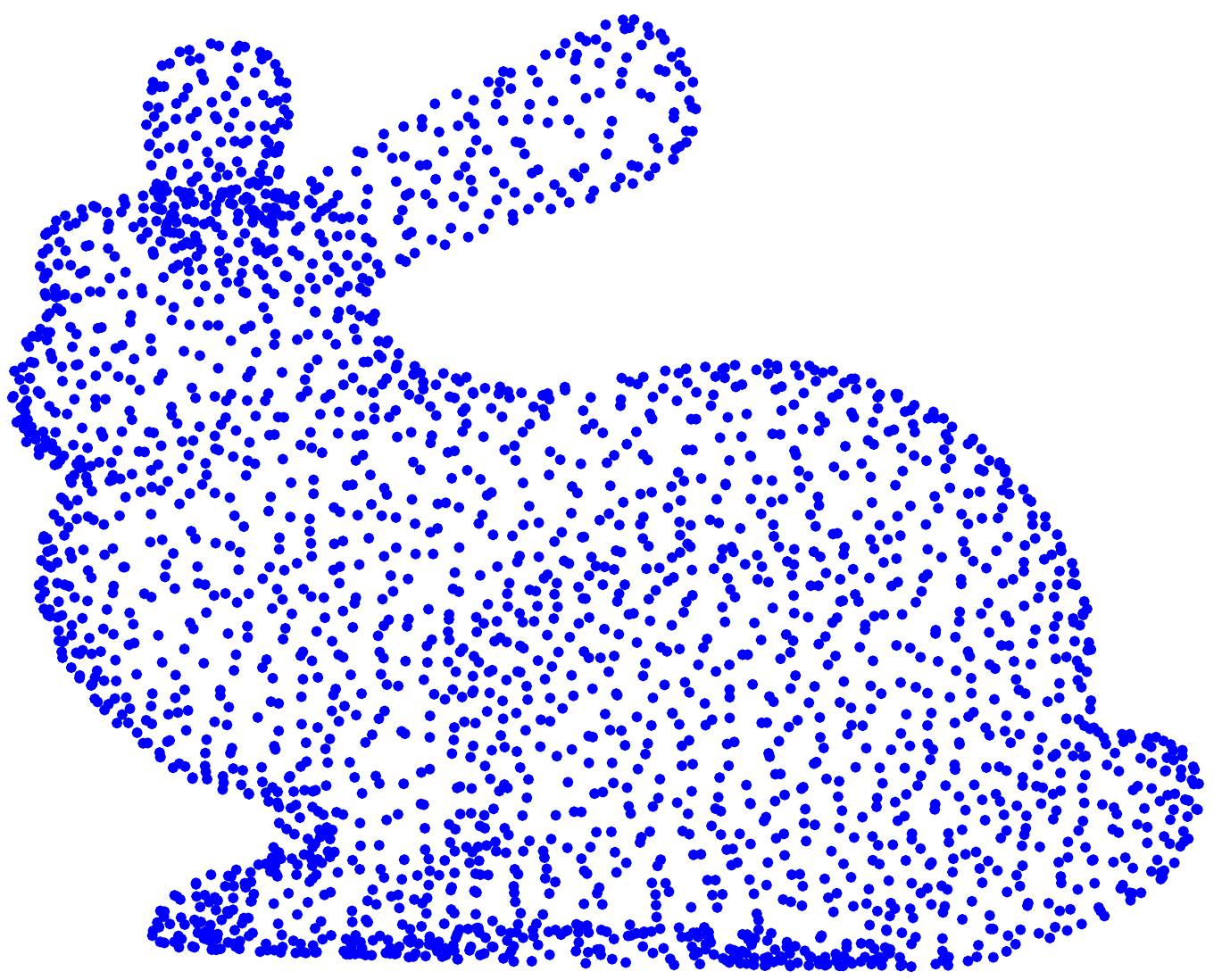}
\includegraphics[width=.19\linewidth, keepaspectratio]{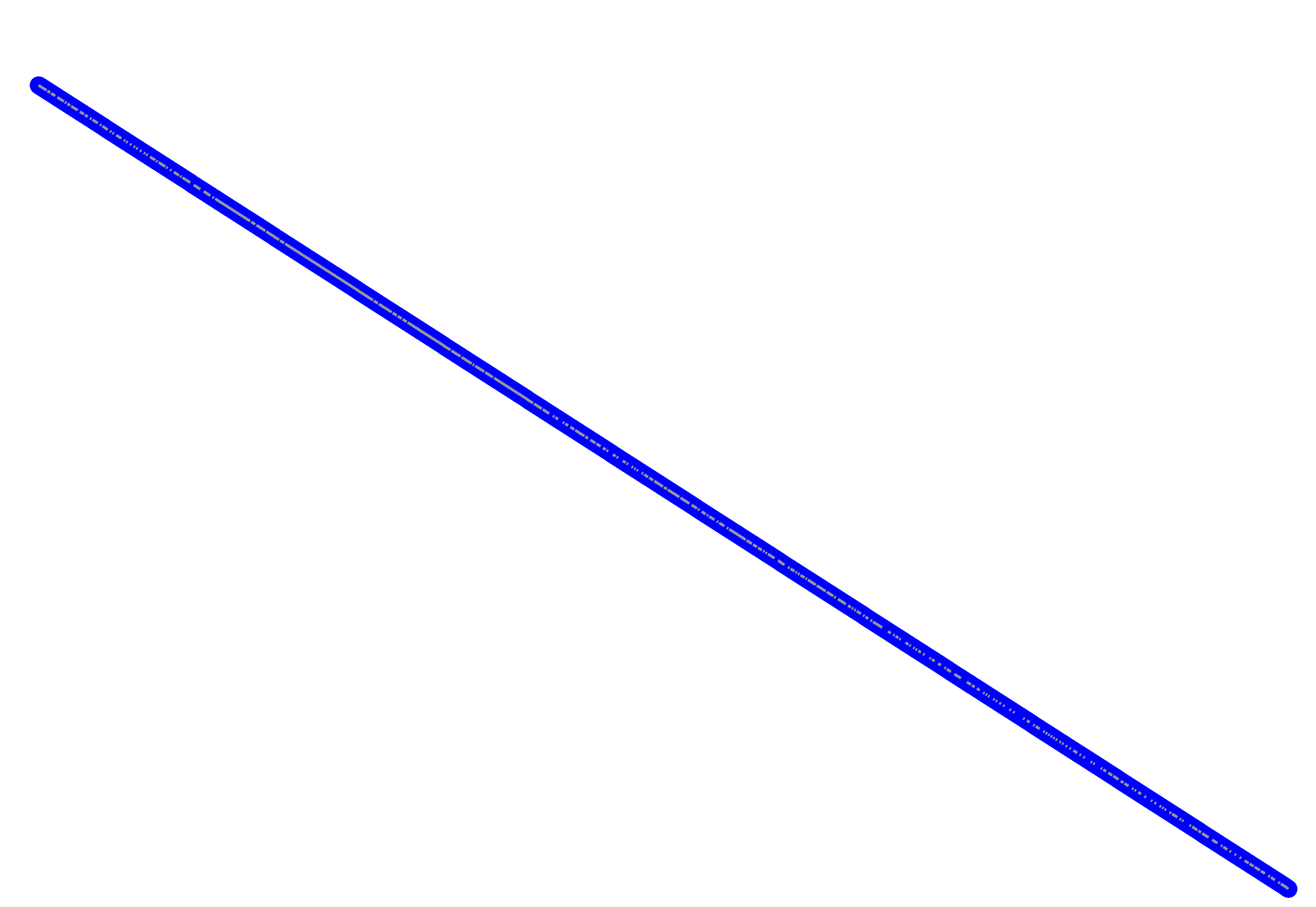}
\includegraphics[width=.19\linewidth, keepaspectratio]{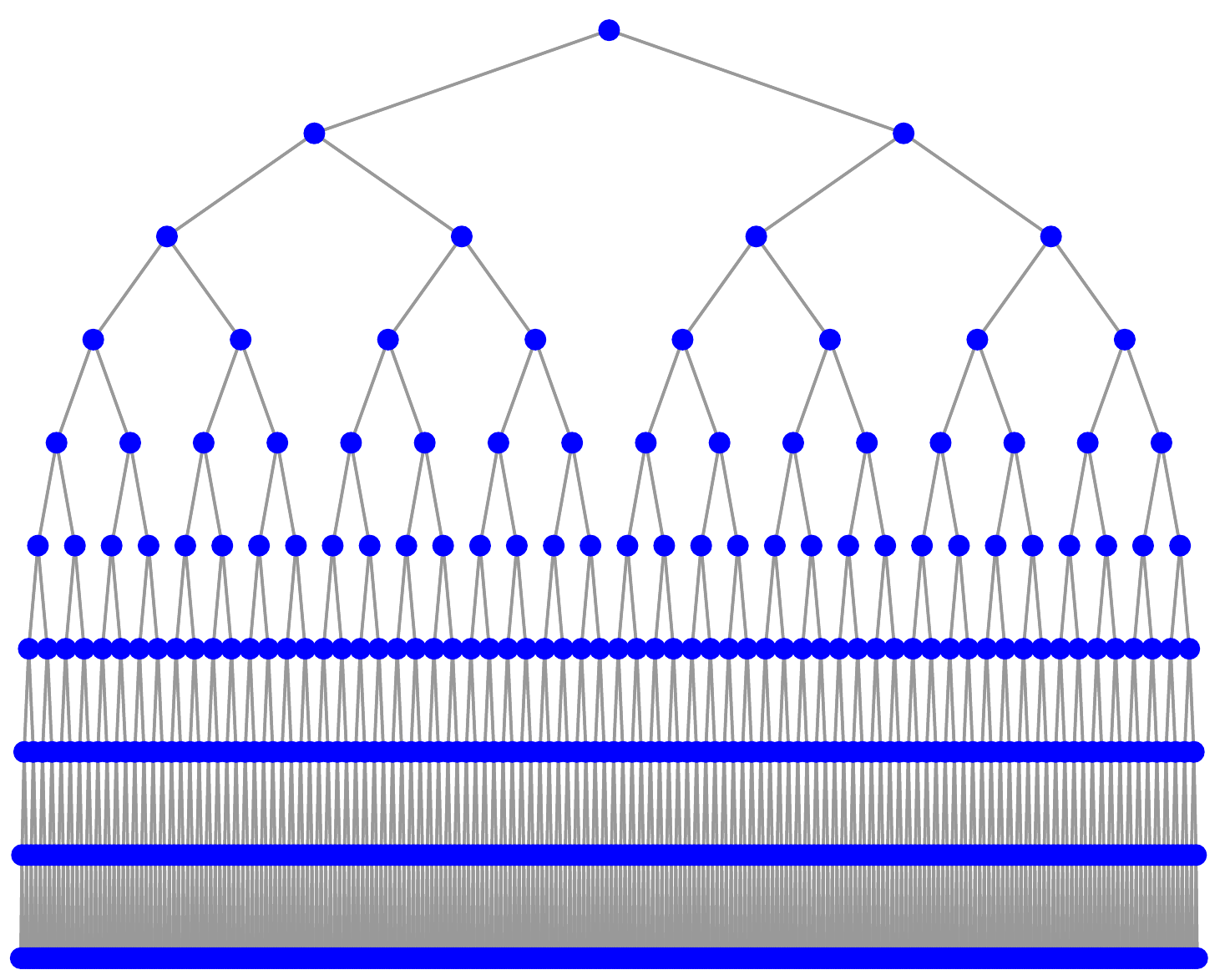}
\caption{\label{fig:graphs} The \MODIF{five} different graphs used in the simulations.}
\end{figure}
%

\subsection{Effect of the sampling distribution on $\nbVertRed$}

In this first part, we study how the sampling distributions affect the minimum number of measurements required to satisfy the RIP. All experiments are repeated for three different sampling distributions: a) the uniform distribution $\vec{\pi}$, b) the optimal distribution $\vec{p}^*$, and c) the estimated optimal distribution $\tilde{\vec{p}} \in \Rbb^\nbVert$ computed using Algorithm \ref{alg:optimal_distribution}.

\subsubsection{Using community graphs}

We conduct a first set of experiments using five types of community graph, denoted by $\CCal_1, \ldots, \CCal_5$. They all have $10$ communities. To study the effect of the size of the communities on the sampling distribution, we choose to build these graphs with $9$ communities of (approximately) equal size and reduce the size of last community: 
\begin{itemize}
\item the graphs of type $\CCal_1$ have $10$ communities of size $100$;
\item the graphs of type $\CCal_2$ have $1$ community of size $50$, $8$ communities of size $105$, and $1$ community of size $110$;
\item the graphs of type $\CCal_3$ have $1$ community of size $25$, $8$ communities of size $108$, and $1$ community of size $111$;
\item the graphs of type $\CCal_4$ have $1$ community of size $17$, $8$ communities of size $109$, and $1$ community of size $111$;
\item the graphs of type $\CCal_5$ have $1$ community of size $13$, $8$ communities of size $109$, and $1$ community of size $115$.
\end{itemize}

For each pair of graph-type, $j \in \{1, \ldots, 5\}$, and sampling distribution, $\prob \in \{\vec{\pi}, \vec{p}^*, \tilde{\vec{p}}\}$, we generate a graph of type $\CCal_j$, compute $\Fou_{10}$ and the lower RIP constant
\begin{align}
\label{eq:lower_rip_constant}
\underline{\delta}_{10} := 1 - \inf_{\substack{\sig \in \spann(\Fou_{10}) \\ \norm{\sig}_2=1}} \left\{ \inv{m} \; \norm{\Meas \ma{P}^{-1/2} \sig}_2^2 \right\},
\end{align}
for different numbers of measurements $\nbVertRed$. Note that to compute $\underline{\delta}_{10}$, one just needs to notice that 
\begin{align*}
\underline{\delta}_{10} = 1 - \inv{\nbVertRed} \; \lambda_{\rm min} \left( \Fou_{10}^\adjoint \ma{P}^{-1/2} \Meas^\adjoint \Meas \ma{P}^{-1/2} \Fou_{10}\right).
\end{align*}
We compute $\underline{\delta}_{10}$ for $500$ independent draws of the matrix $\Meas$. When conducting the experiments with the estimated optimal distribution $\tilde{\vec{p}}$, we re-estimate this distribution at each of the $500$ trials.

We present in Fig.~\ref{fig:eig_community} the probability that $\underline{\delta}_{10}$ is less than $0.995$, estimated over the $500$ trials, as a function of $\nbVertRed$. Let $m_{j, \prob}^*$ be the number of measurements required to reach a probability of, \eg, $\Pbb(\underline{\delta}_{10} \leq 0.995) = 0.9$ for the pair $(j, \prob)$ of graph-type and sampling distribution. Theorem~\ref{th:rip} predicts that $m_{j, \prob}^*$ scales linearly with $(\cumCoh_{\prob}^{10})^2$. 
\begin{itemize}
\item For the uniform distribution $\vec{\pi}$, the first figure from the left in Fig.~\ref{fig:eig_community} indicates the value of $(\cumCoh_{\vec{\pi}}^{10})^2(j)$, $j=1, \ldots, 5$ for the five different types of graph. We have $(\cumCoh_{\vec{\pi}}^{10})^2(1) \leq \ldots \leq (\cumCoh_{\vec{\pi}}^{10})^2(5)$ and $m_{1, \vec{\pi}}^* \leq m_{2, \vec{\pi}}^* \leq \ldots \leq m_{5, \vec{\pi}}^*$, in accordance with Theorem~\ref{th:rip}. 
\item For the optimal sampling distribution $\vec{p}^*$, we have $(\cumCoh_{\vec{p}^*}^{10})^2 = 10$. Therefore $m_{j, \vec{p}^*}^*$ must be identical for all graph-types, as observed in the second panel of Fig.~\ref{fig:eig_community}. 
\item For the estimated optimal sampling distribution $\tilde{\vec{p}}$, the last figure in Fig.~\ref{fig:eig_community} shows that the performance is identical for all graph-types, as with $\vec{p}^*$. Furthermore, we attained almost the same performance with $\tilde{\vec{p}}$ and $\vec{p}^*$, confirming the quality of the estimation provided by Algorithm~\ref{alg:optimal_distribution}
\end{itemize}
\begin{figure}
\centering
\begin{minipage}{.32\linewidth} \centering \small \hspace{2mm} Uniform distribution $\vec{\pi}$ \end{minipage}
\begin{minipage}{.32\linewidth} \centering \small \hspace{2mm} Optimal distribution $\vec{p}^*$ \end{minipage}
\begin{minipage}{.32\linewidth} \centering \small \hspace{2mm} Estimated distribution $\tilde{\vec{p}}$ \end{minipage}\\
\includegraphics[width=.32\linewidth, keepaspectratio]{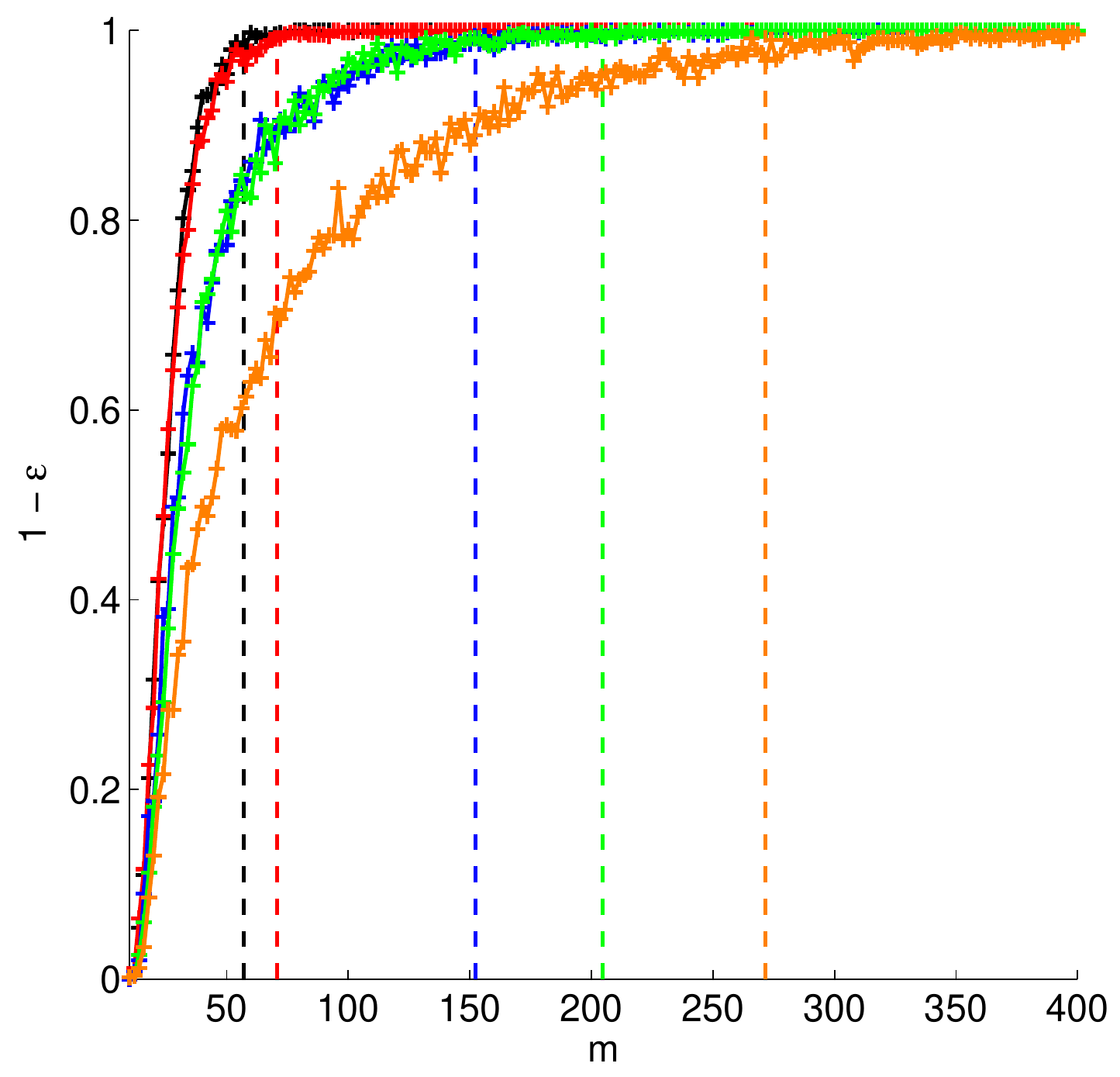}
\includegraphics[width=.32\linewidth, keepaspectratio]{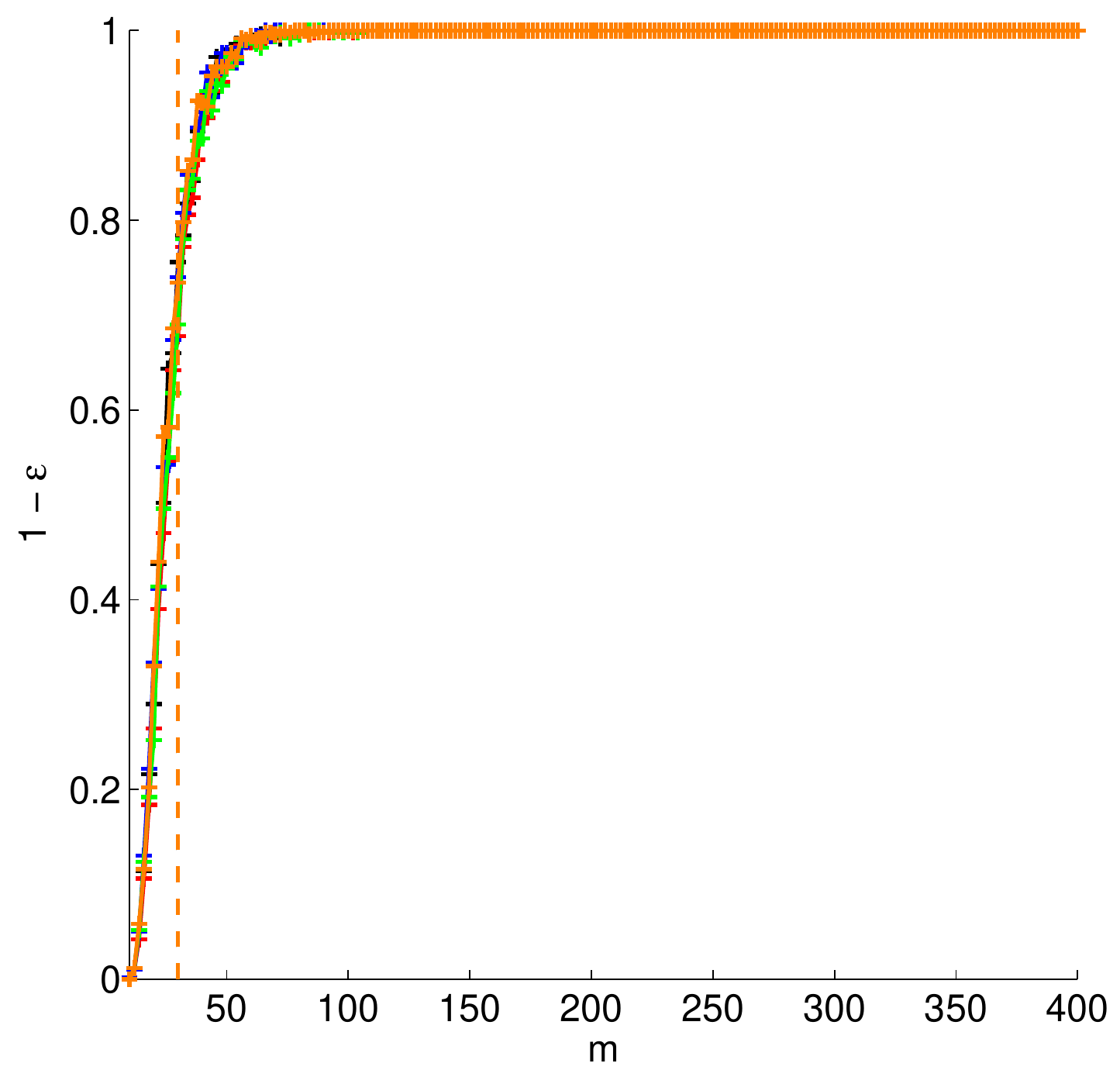}
\includegraphics[width=.32\linewidth, keepaspectratio]{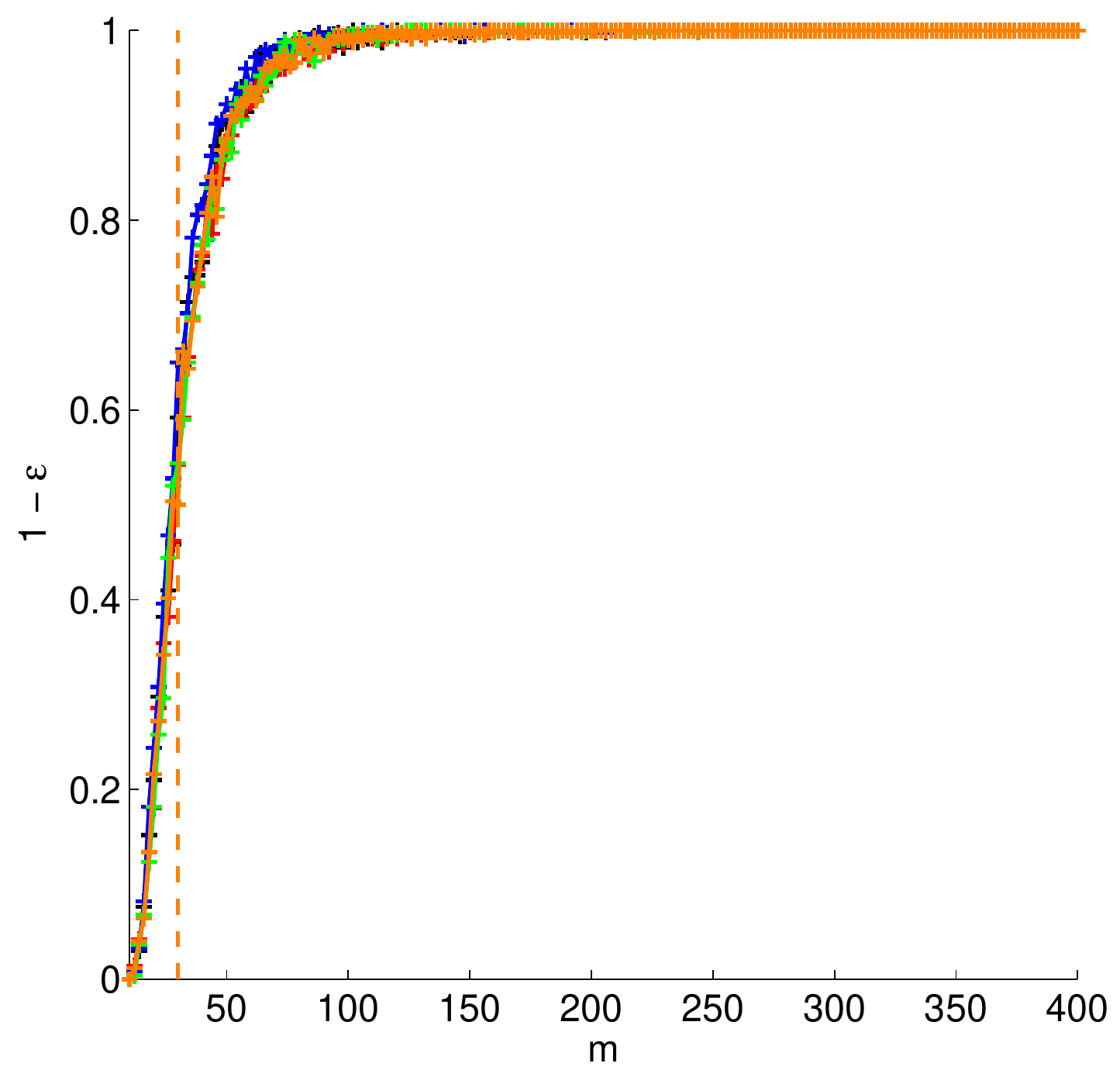}
\caption{\label{fig:eig_community} Probability that $\underline{\delta}_{10}$ is less than $0.995$ as a function of $\nbVertRed$ for $5$ different types of community graph: $\CCal_1$ in black, $\CCal_2$ in red, $\CCal_3$ in blue, $\CCal_4$ in green, $\CCal_5$ in orange. Left panel: the dashed vertical lines indicate the value of $3 \cdot (\cumCoh_{\vec{\pi}}^{10})^2$ for each type of graph. Middle and right panels: the dashed vertical lines indicate the value $3 \cdot (\cumCoh_{\vec{p}^*}^{10})^2 = 3 \cdot 10$.}
\end{figure}
%

\subsubsection{Using the Minnesota and bunny graphs}

To confirm the results observed above, we repeat the same experiments but using \MODIF{four other graphs: the Minnesota, the bunny and path graphs, and the binary tree. For the first three graphs, the experiments are performed for $k$-bandlimited signals with band-limits $10$ and $100$, \ie, we compute $\underline{\delta}_{k}$ - defined as in \refeq{eq:lower_rip_constant} - with $\Fou_{k}$.} \ADDED{For the binary tree, we set the band-limits at $16$ and $64$. These choices are due to the fact that some eigenvalues have a multiplicity larger than $1$ for this last graph. These choices ensure that $\eig_{\nbClass} < \eig_{\nbClass+1}$, as required in our assumptions.}

\MODIF{We present in Fig.~\ref{fig:eig_minne_bunny} the probability that $\underline{\delta}_{\nbClass}$ is less than $0.995$, estimated over $500$ draws of $\Meas$, as a function of $\nbVertRed$.}

\MODIF{For the Minnesota, the bunny and path graphs at $\nbClass = 10$, we remark that all distributions yield essentially the same result. The advantage of using the distributions $\vec{p}^*$ or $\tilde{\vec{p}}$ is more obvious at $k=100$ for the Minnesota and the bunny graphs. Note that for the bunny graph, we reach only a probability of $0.036$ at $\nbVertRed=2000$ with the uniform distribution, whereas $\nbVertRed=600$ measurements are sufficient to reach a probability $1$ with $\vec{p}^*$. Uniform sampling is not working for the bunny graph at $k=100$ because there exist few eigenmodes whose energy is highly concentrated on few nodes. In other words, we have $\norm{\Fou_{100}^\adjoint \delta_i}_2 \approx 1$ for few nodes $i$. Finally, for the path graph at $k=100$ and for the binary tree, we notice that the uniform distribution and the optimal distribution $\vec{p}^*$ have the same performance while the estimated optimal distribution $\tilde{\vec{p}}$ performs slightly worse. However, we noticed that these differences decrease when increasing the number $L$ of random signals in Algorithm~\eqref{alg:optimal_distribution}.}

\begin{figure}
\centering
\begin{minipage} {.242\linewidth} \centering \scriptsize \hspace{2mm} Minnesota - $k = 10$ \end{minipage}
\begin{minipage} {.242\linewidth} \centering \scriptsize \hspace{2mm} Bunny - $k=10$ \end{minipage}
\begin{minipage} {.242\linewidth} \centering \scriptsize \hspace{2mm} Path graph - $k=10$ \end{minipage}
\begin{minipage} {.242\linewidth} \centering \scriptsize \hspace{2mm} Binary tree - $k=16$ \end{minipage} \\
\includegraphics[width=.242\linewidth, keepaspectratio]{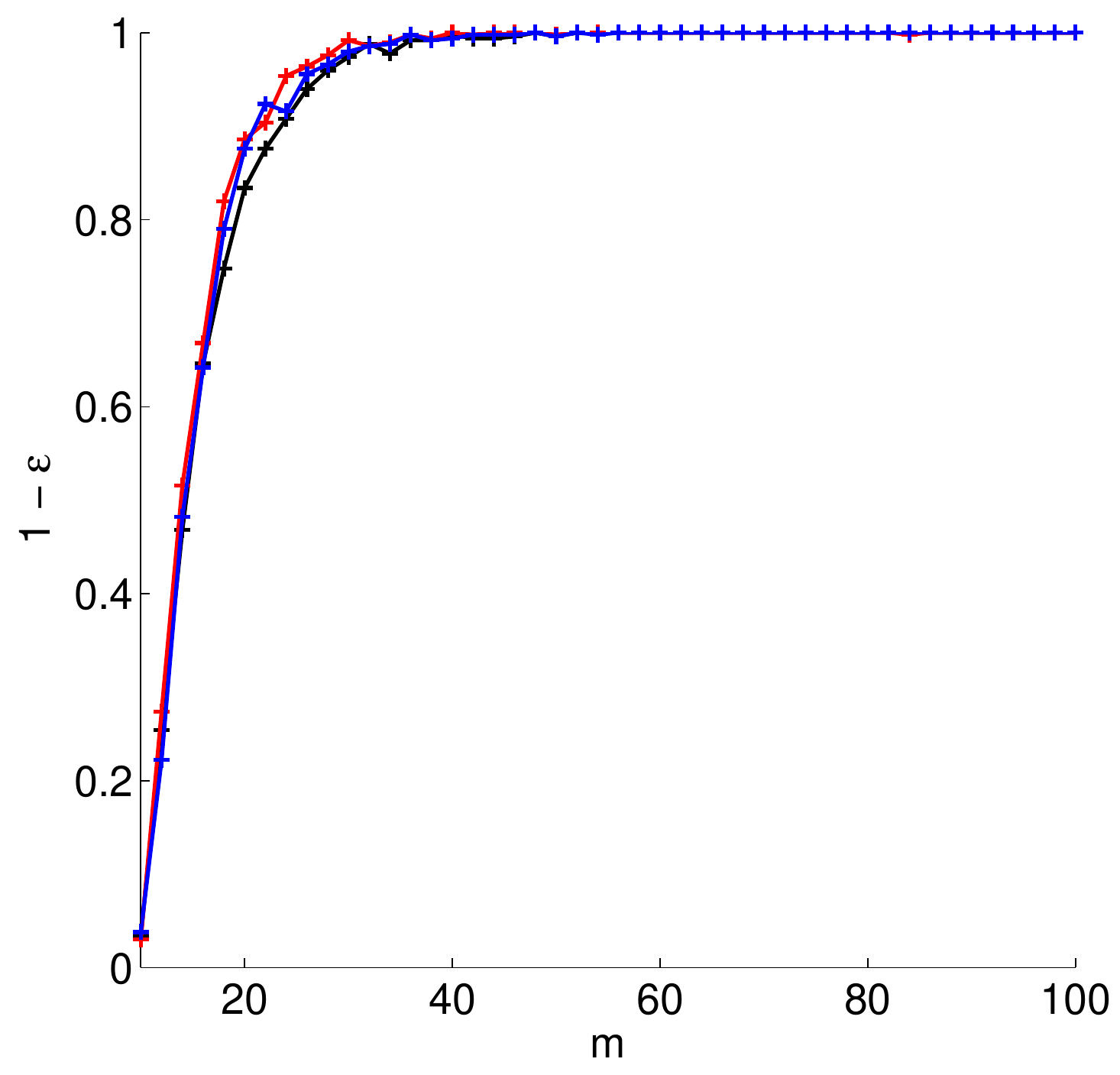}
\includegraphics[width=.242\linewidth, keepaspectratio]{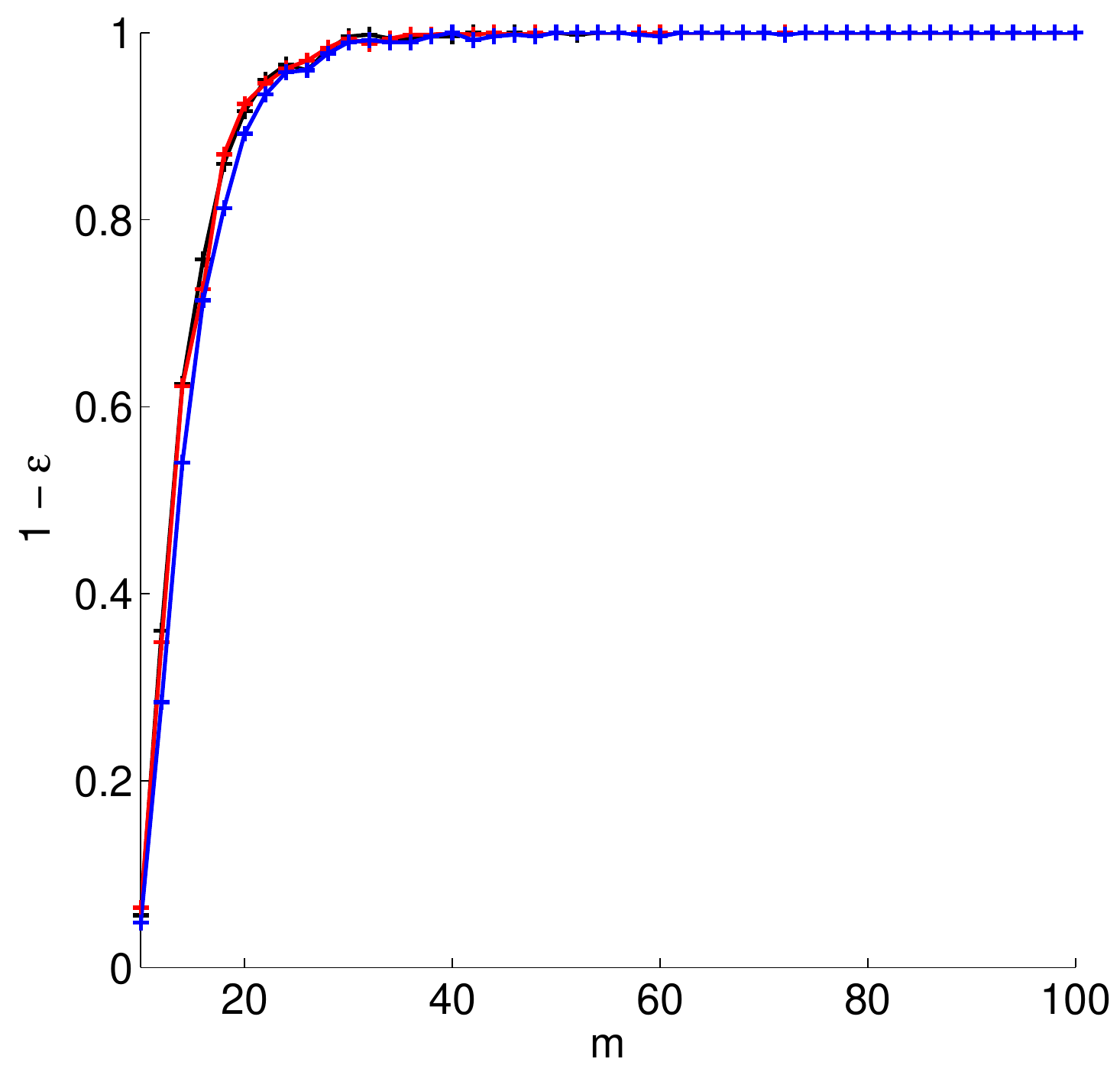}
\includegraphics[width=.242\linewidth, keepaspectratio]{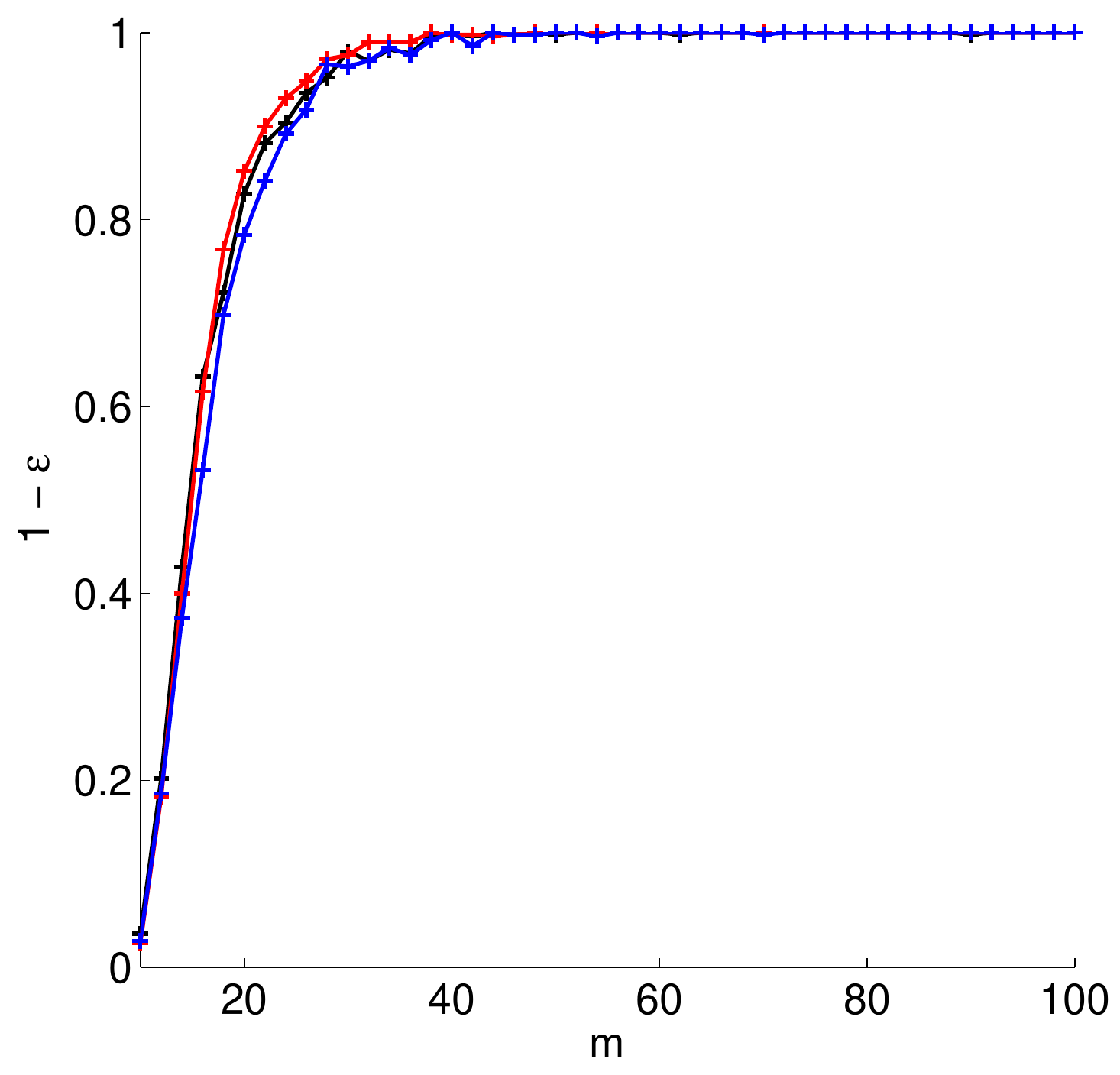}
\includegraphics[width=.242\linewidth, keepaspectratio]{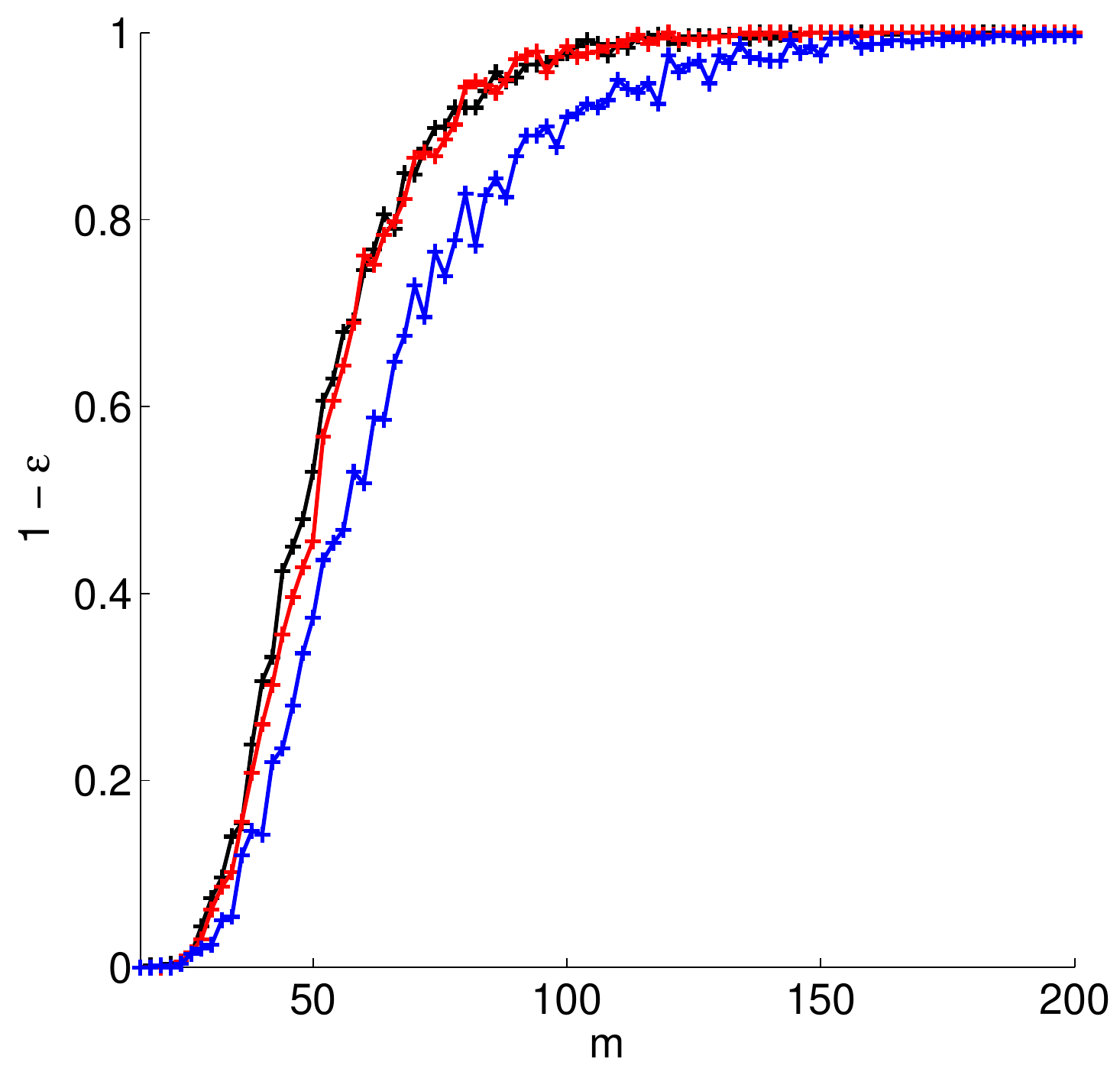} \\
\begin{minipage} {.242\linewidth} \centering \scriptsize \hspace{2mm} Minnesota - $k = 100$ \end{minipage}
\begin{minipage} {.242\linewidth} \centering \scriptsize \hspace{2mm} Bunny - $k=100$ \end{minipage}
\begin{minipage} {.242\linewidth} \centering \scriptsize \hspace{2mm} Path graph - $k=100$ \end{minipage}
\begin{minipage} {.242\linewidth} \centering \scriptsize \hspace{2mm} Binary tree - $k=64$ \end{minipage}\\
\includegraphics[width=.242\linewidth, keepaspectratio]{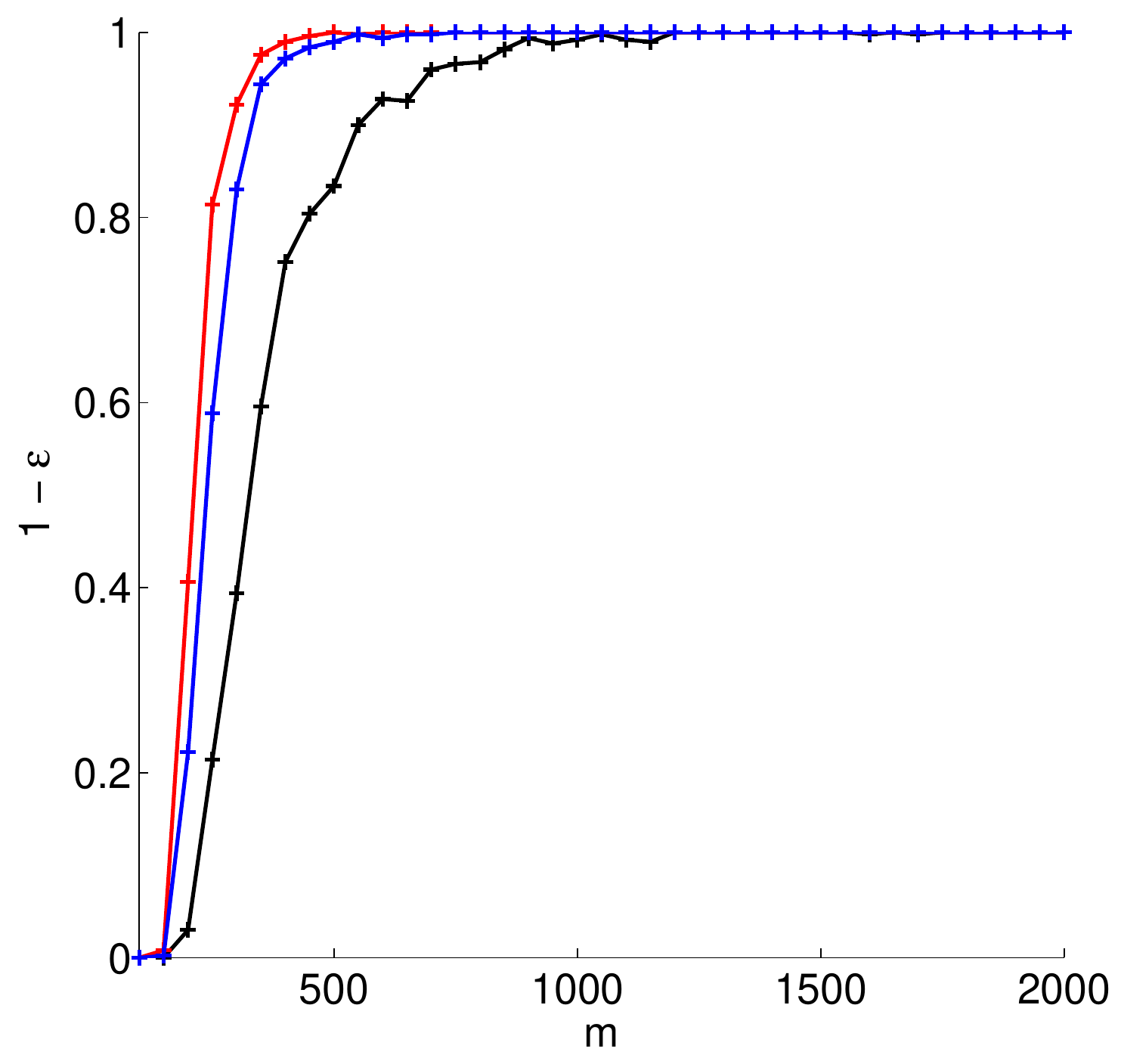}
\includegraphics[width=.242\linewidth, keepaspectratio]{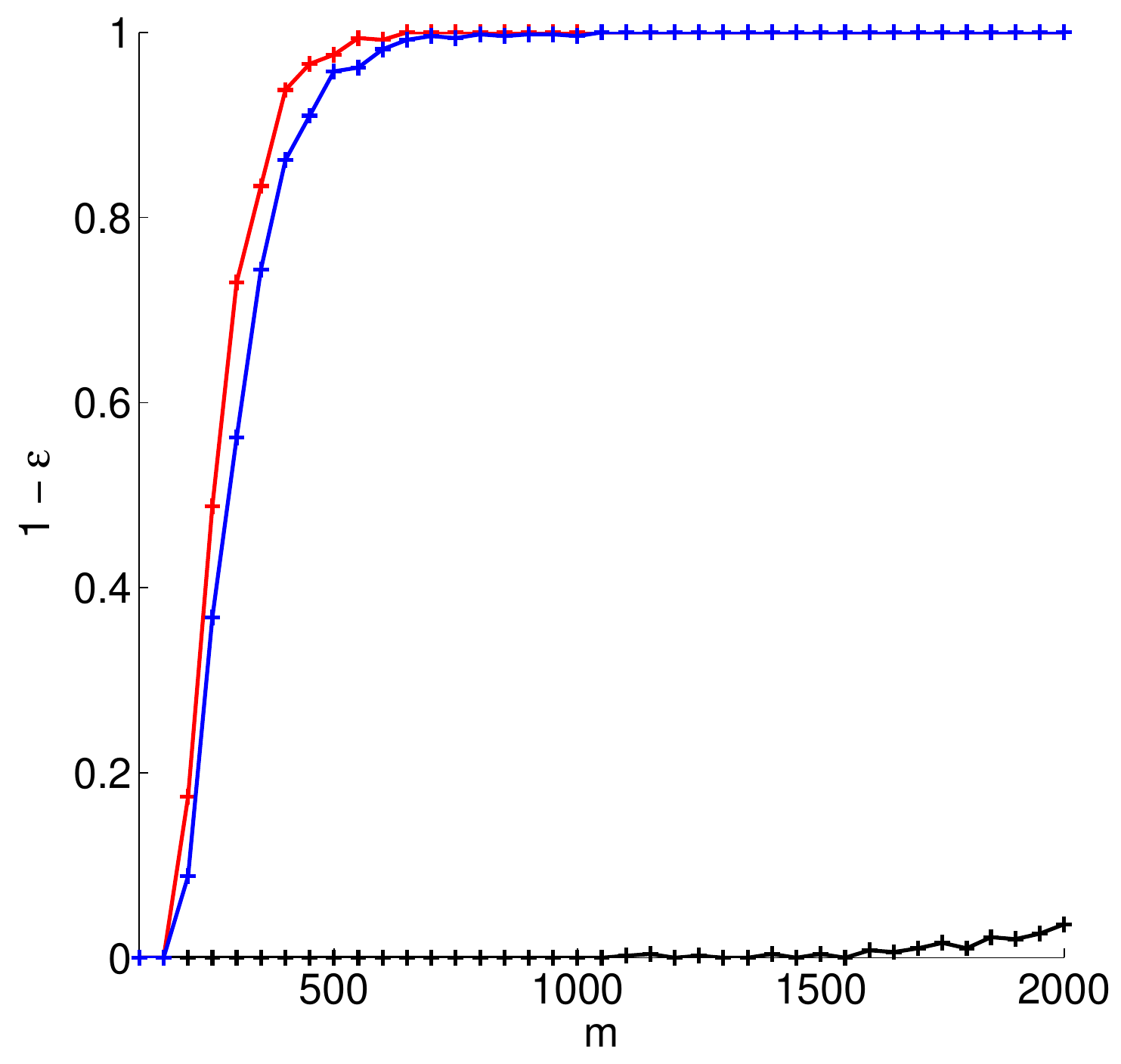}
\includegraphics[width=.242\linewidth, keepaspectratio]{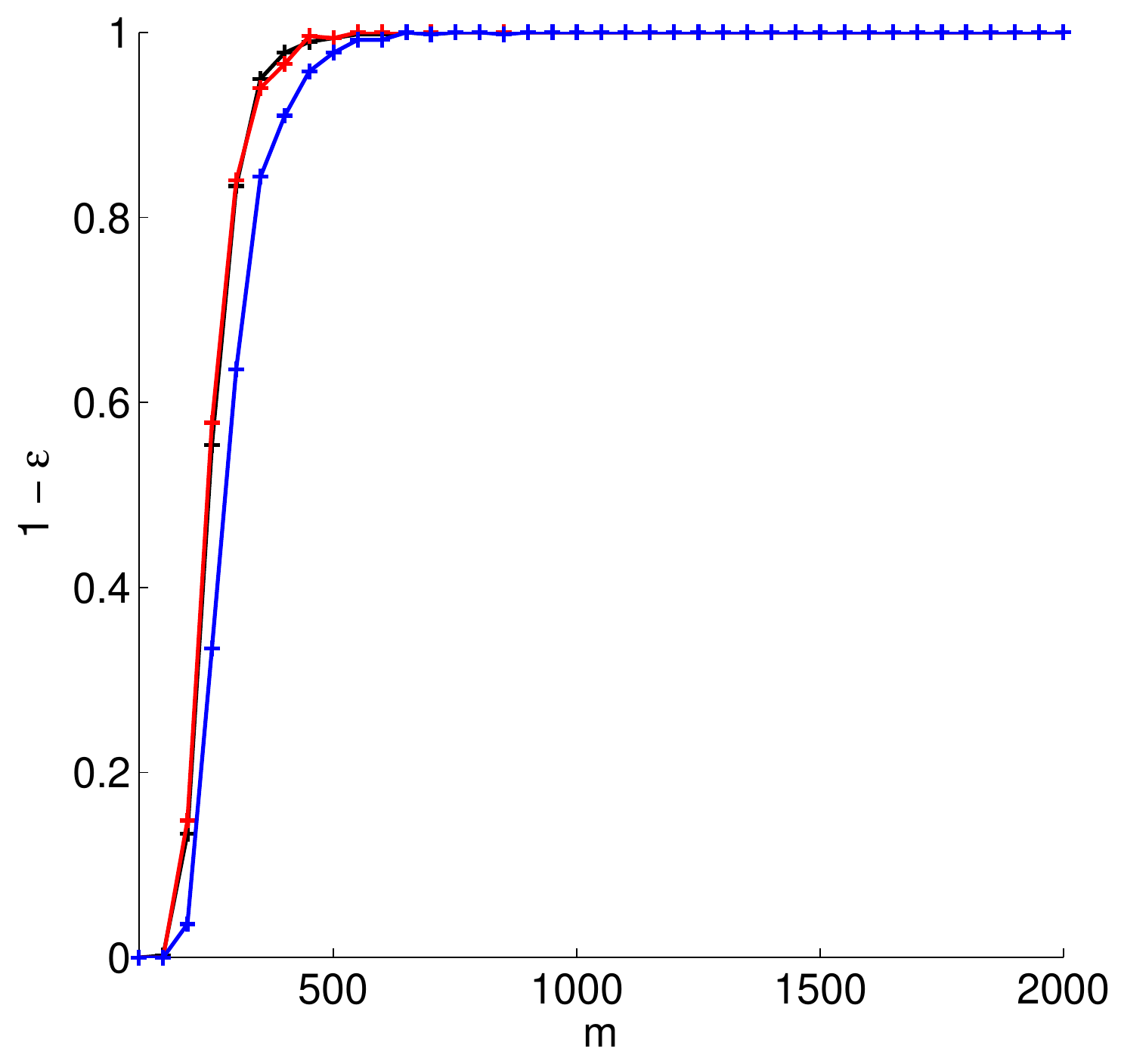}
\includegraphics[width=.242\linewidth, keepaspectratio]{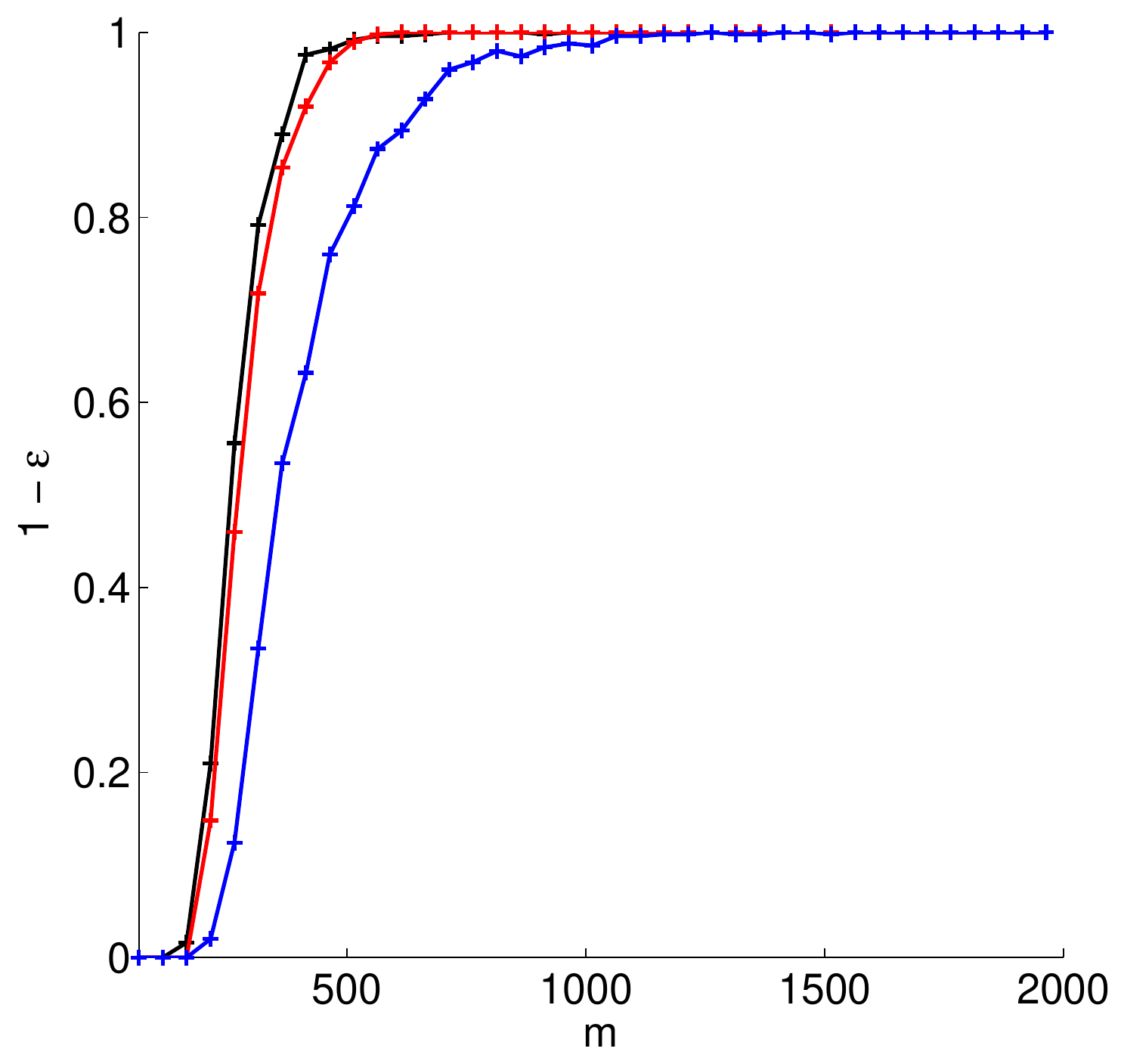}
\caption{\label{fig:eig_minne_bunny} Probability that $\underline{\delta}_k$ is less than $0.995$ as a function of $\nbVertRed$. The curve in black indicates the result for the uniform distribution. The curve in red indicates the result for the optimal distribution. The curve in blue indicates the result for the estimated optimal distribution. \MODIF{The panels on the top row show the results at $k=10$ for the Minnesota, bunny and path graphs, and at $k=16$ for the binary tree. The panels on the bottom row show the results at $k=100$ for the Minnesota, bunny and path graphs, and at $k=64$ for the binary tree.}}
\end{figure}
%

\subsubsection{Examples of optimal and estimated sampling distributions}

For illustration, we present some examples of sampling distributions in Fig.~\ref{fig:sampling_distribution} for \MODIF{five} of the graphs used above. The top panels in Fig.~\ref{fig:sampling_distribution} show the optimal sampling distribution computed with $\Fou_k$. The bottom panels show the estimated sampling distribution obtained with Algorithm \ref{alg:optimal_distribution}. 

\ADDED{For the community, Minnesota and bunny graphs, we notice that the estimated sampling distribution $\tilde{\vec{p}}$ and the optimal one $\vec{p}^*$ are quite similar. We observe more differences for the path graph and the binary tree. This explains the slight differences of performance in the previous experiments. We recall that these differences decrease when increasing the number $L$ of random signals in Algorithm~\ref{alg:optimal_distribution}.}

\ADDED{It is interesting to notice that for the path graph, the optimal sampling distribution is essentially constant except for the nodes at the boundaries that are less connected than the other nodes and need to be sampled with higher probability. For the binary tree, the probability of sampling a node is only determined by its depth in the tree - all nodes at a given depth are equally important - and the deeper the node is, the higher the probability of sampling this node should be - it is easier to predict the value at one node from the values bore by its children than its parents.}

\begin{figure}
\centering
\hspace{1.5mm}
\begin{minipage} {.19\linewidth}\centering \hspace{2mm} \scriptsize Community graph $\CCal_5$ \\ $k = 10$ \end{minipage}
\begin{minipage} {.19\linewidth} \centering \scriptsize \hspace{0mm} Minnesota graph \\ $k = 100$\end{minipage}
\begin{minipage} {.19\linewidth} \centering \scriptsize \hspace{0mm} Bunny graph \\ \hspace{0mm} $k = 100$\end{minipage}
\begin{minipage} {.19\linewidth} \centering \scriptsize Path graph \\ $k = 100$\end{minipage}
\begin{minipage} {.19\linewidth} \centering \scriptsize Binary tree \\ $k = 64$\end{minipage}\\
\begin{sideways} \hspace{1mm}\tiny Optimal sampling  \end{sideways}
\includegraphics[height=25mm, keepaspectratio]{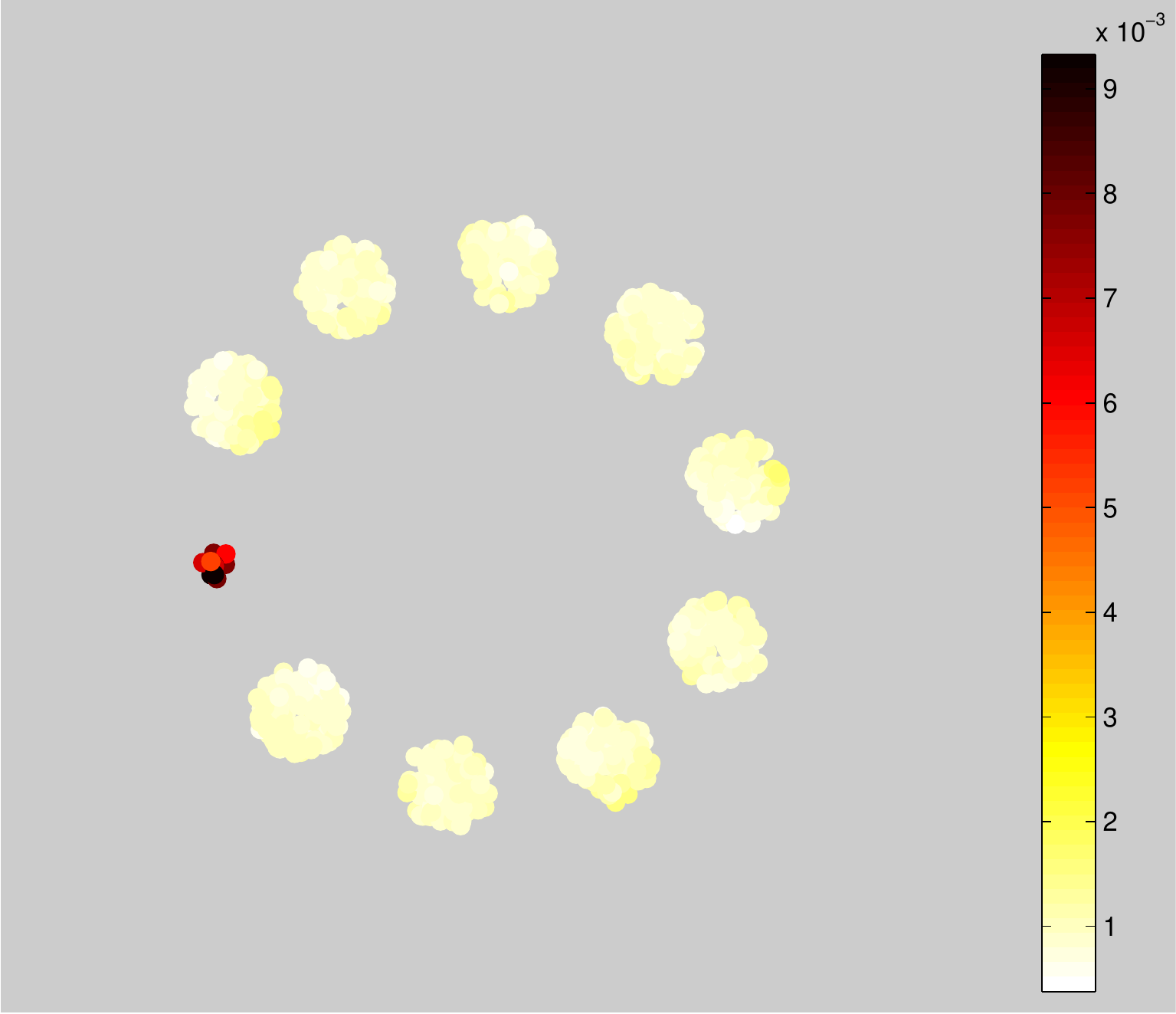}
\includegraphics[height=25mm, keepaspectratio]{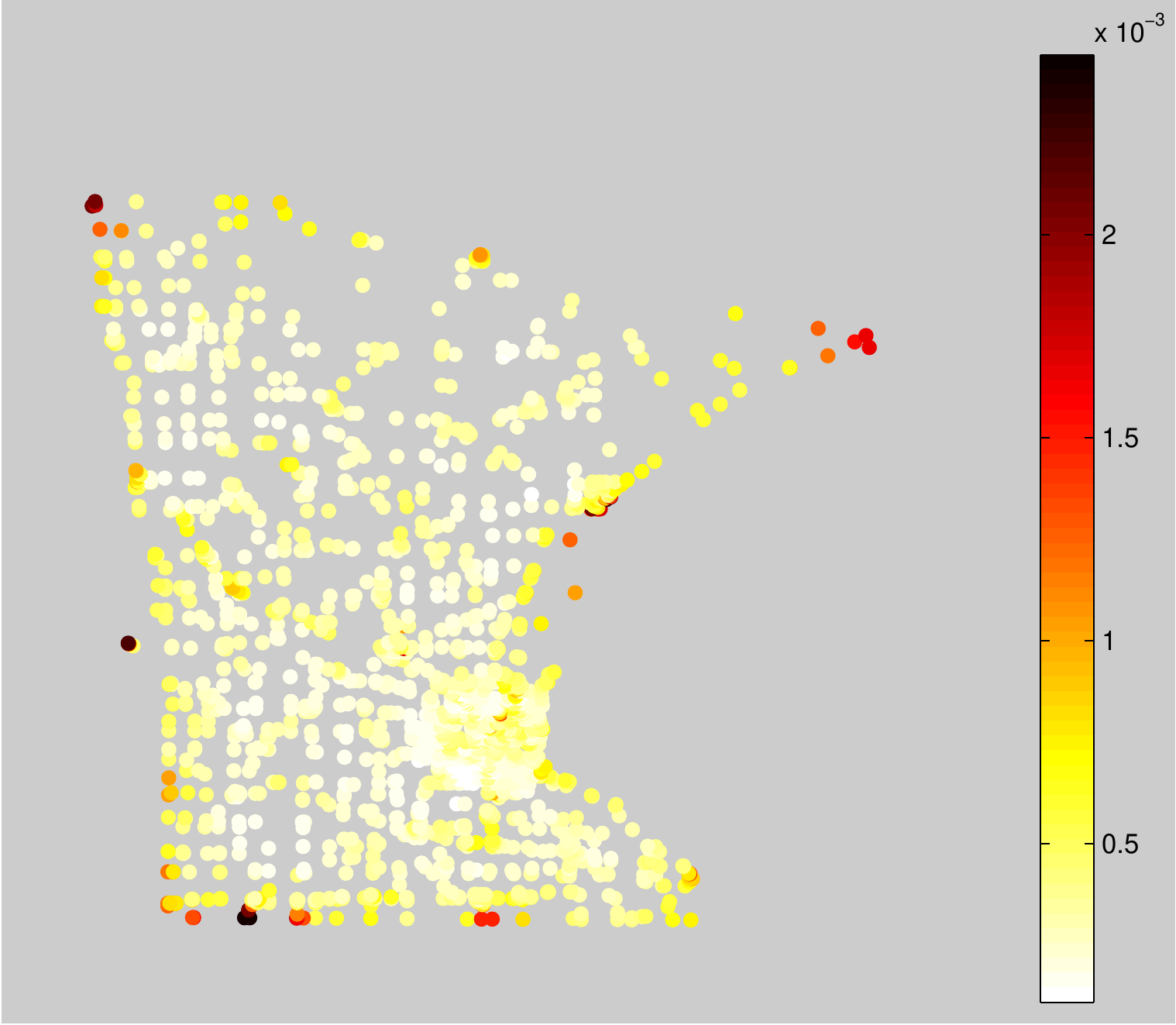}
\includegraphics[height=25mm, keepaspectratio]{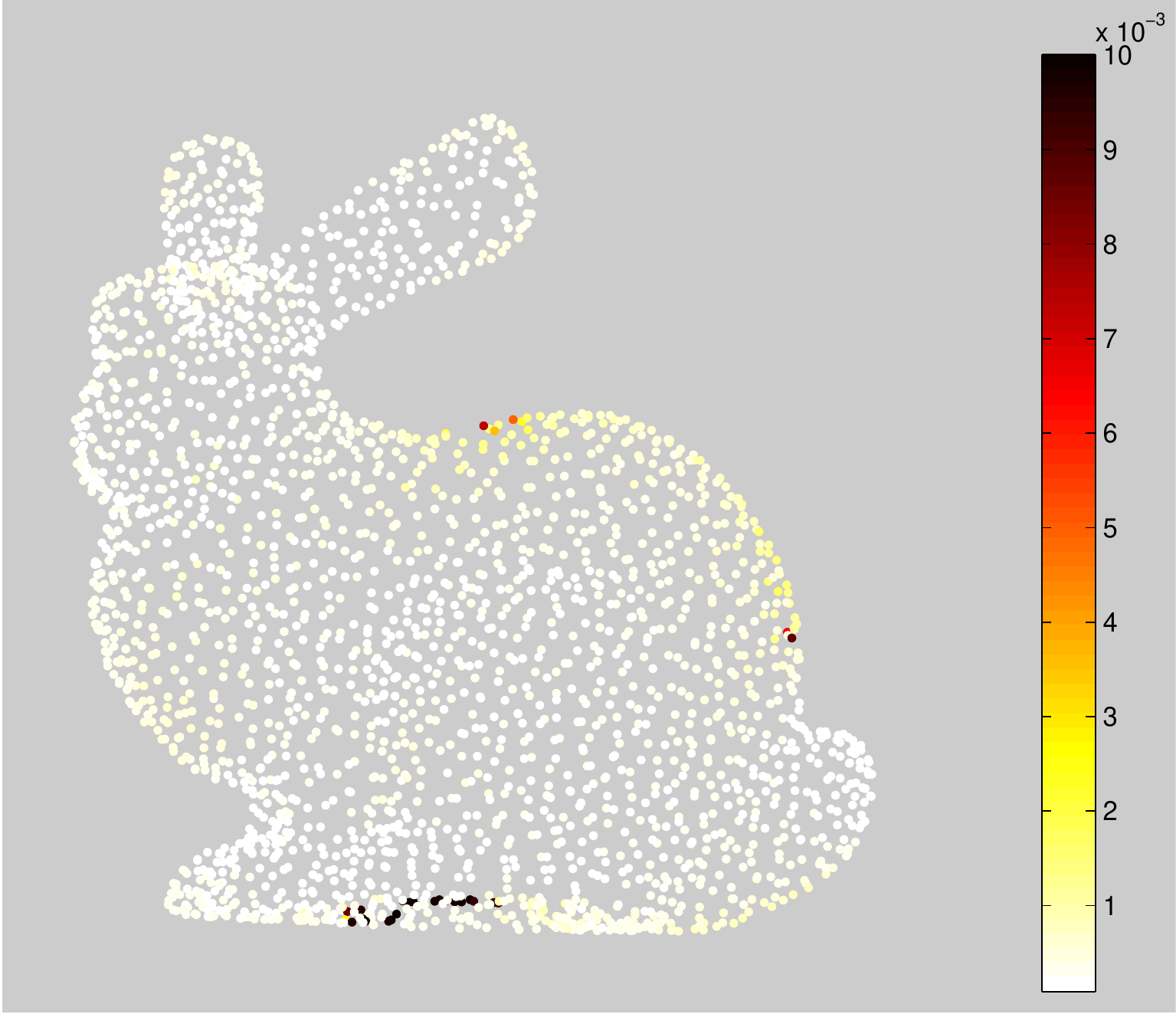}
\includegraphics[height=25mm, keepaspectratio]{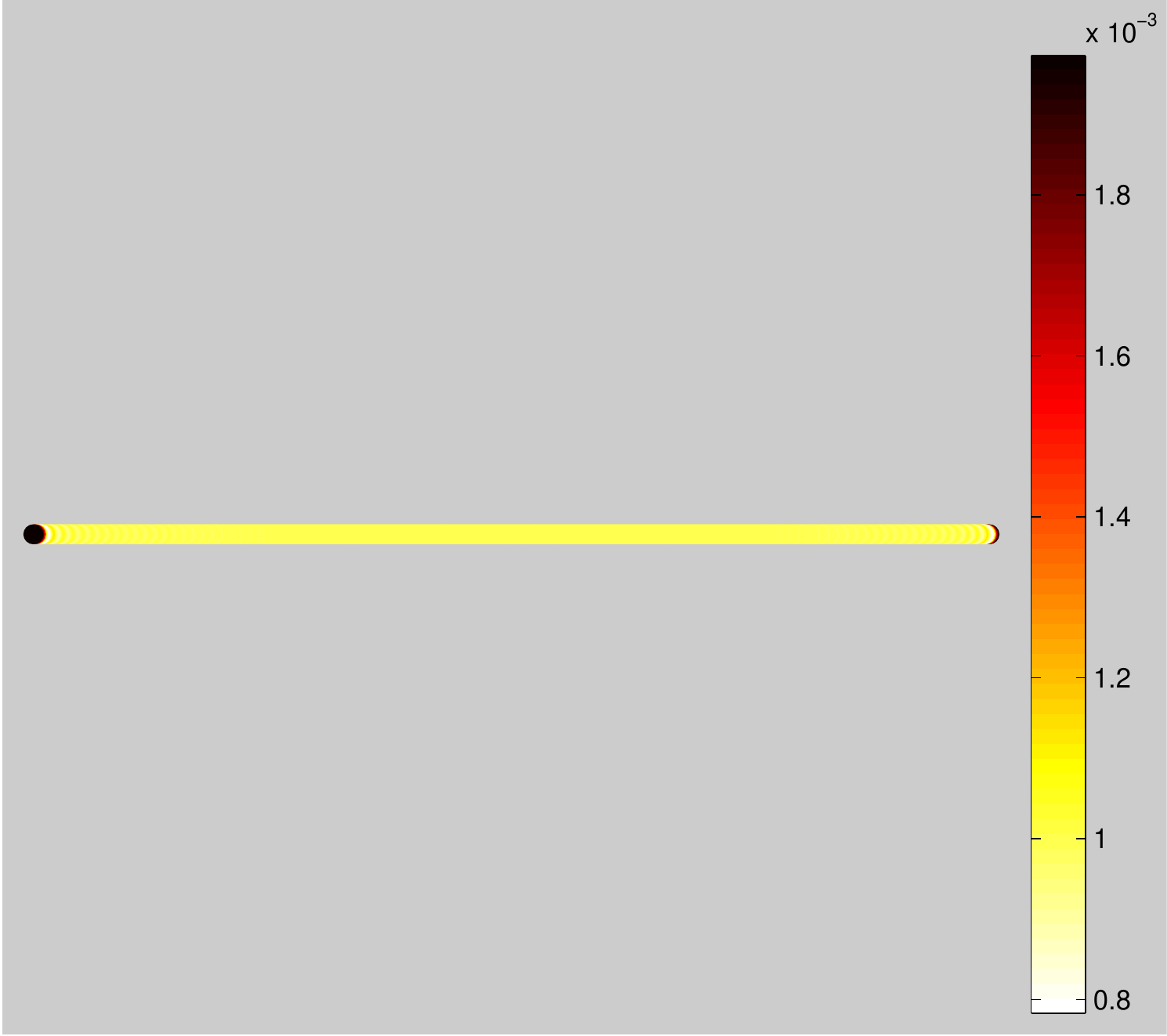}
\includegraphics[height=25mm, keepaspectratio]{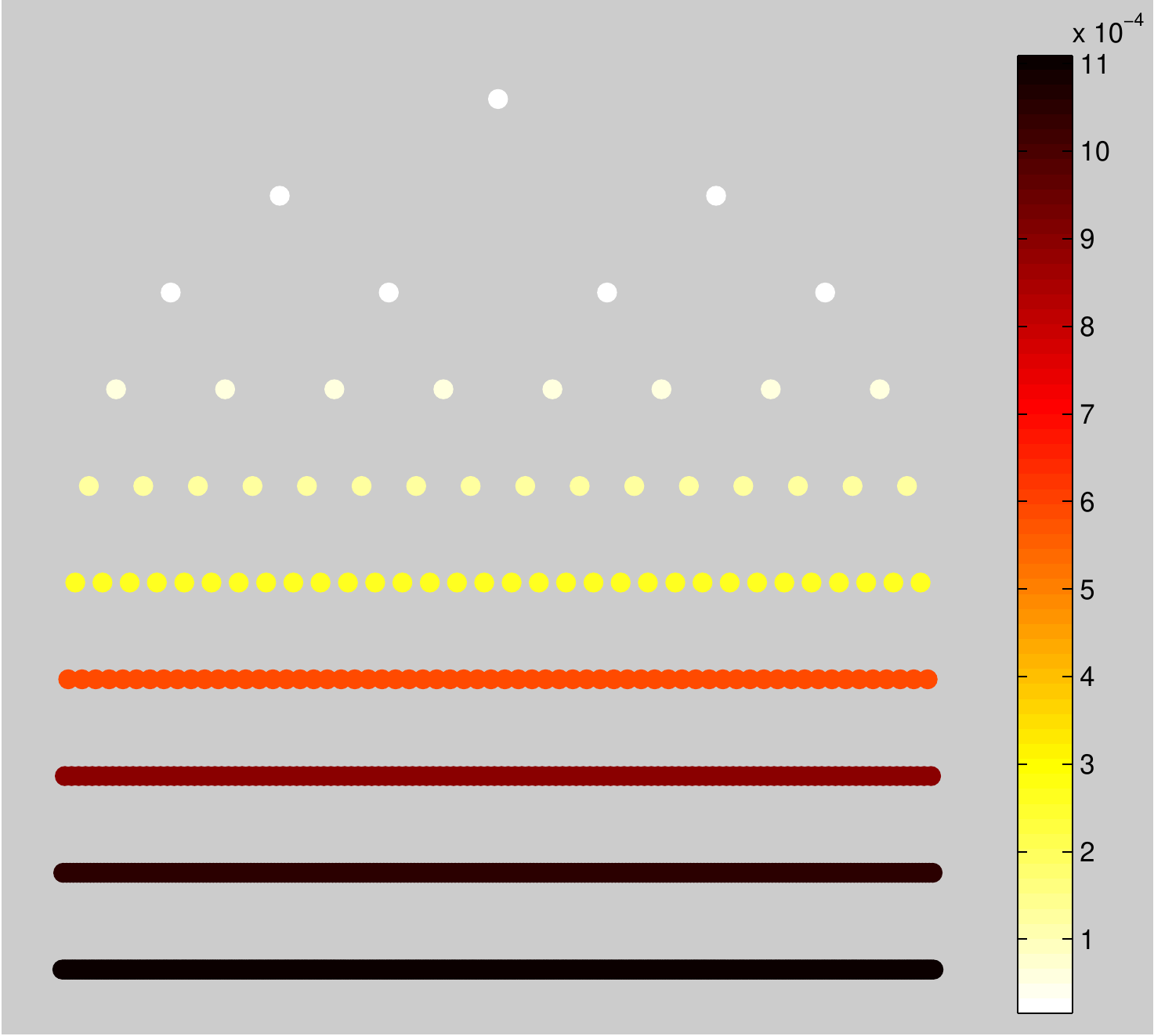}\\
\begin{sideways} \tiny  Estimated sampling \end{sideways}
\includegraphics[height=25mm, keepaspectratio]{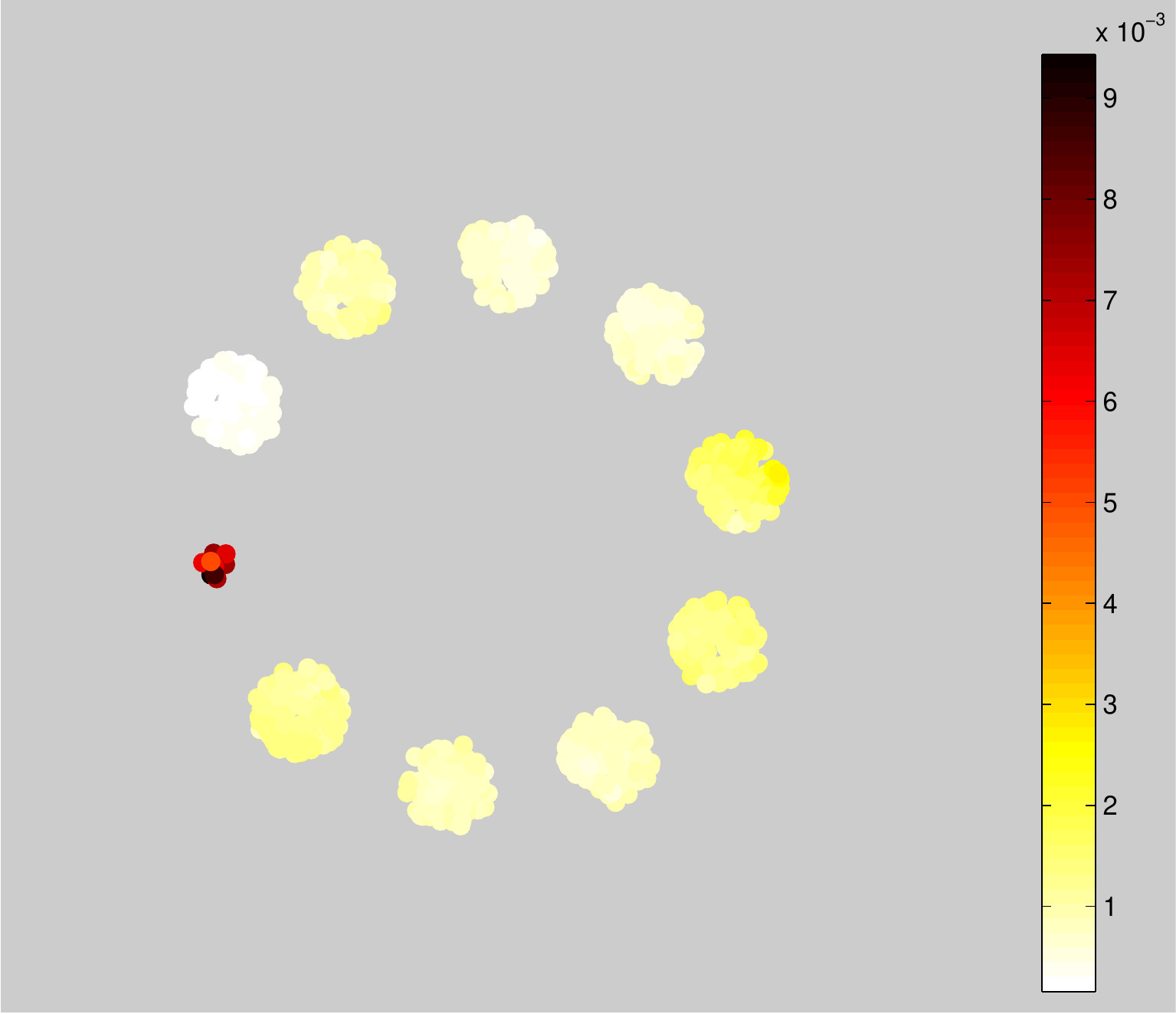}
\includegraphics[height=25mm, keepaspectratio]{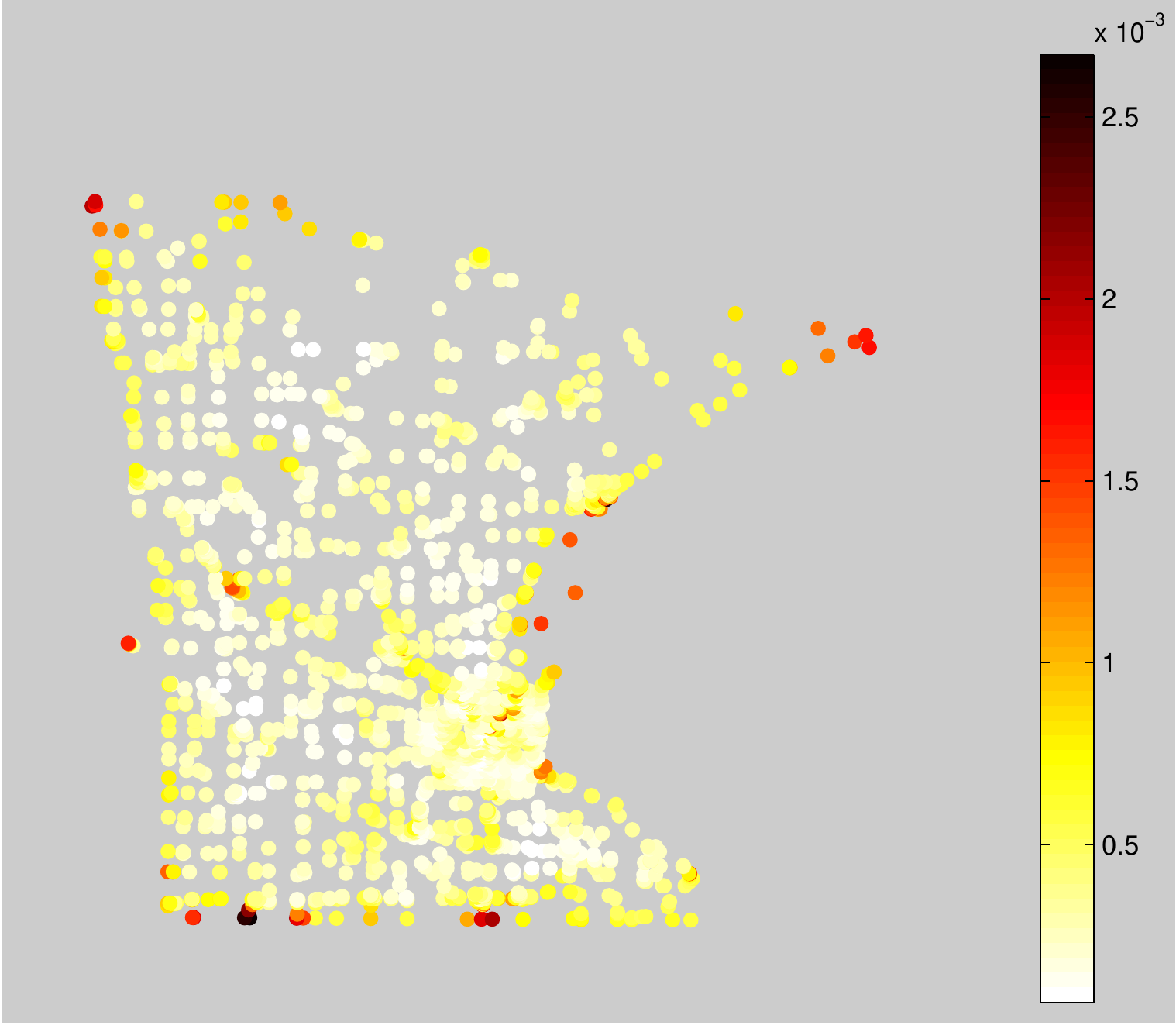}
\includegraphics[height=25mm, keepaspectratio]{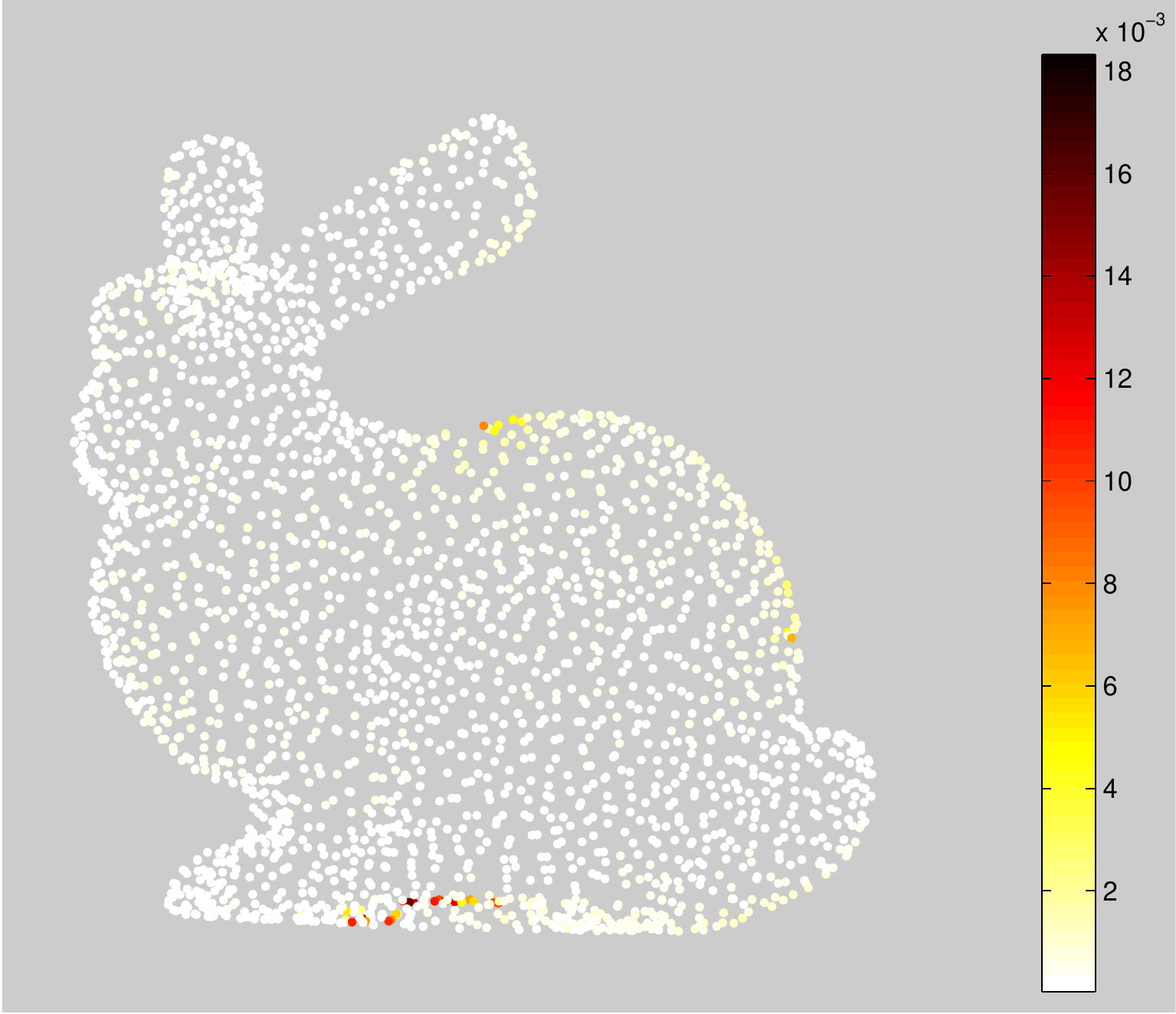}
\includegraphics[height=25mm, keepaspectratio]{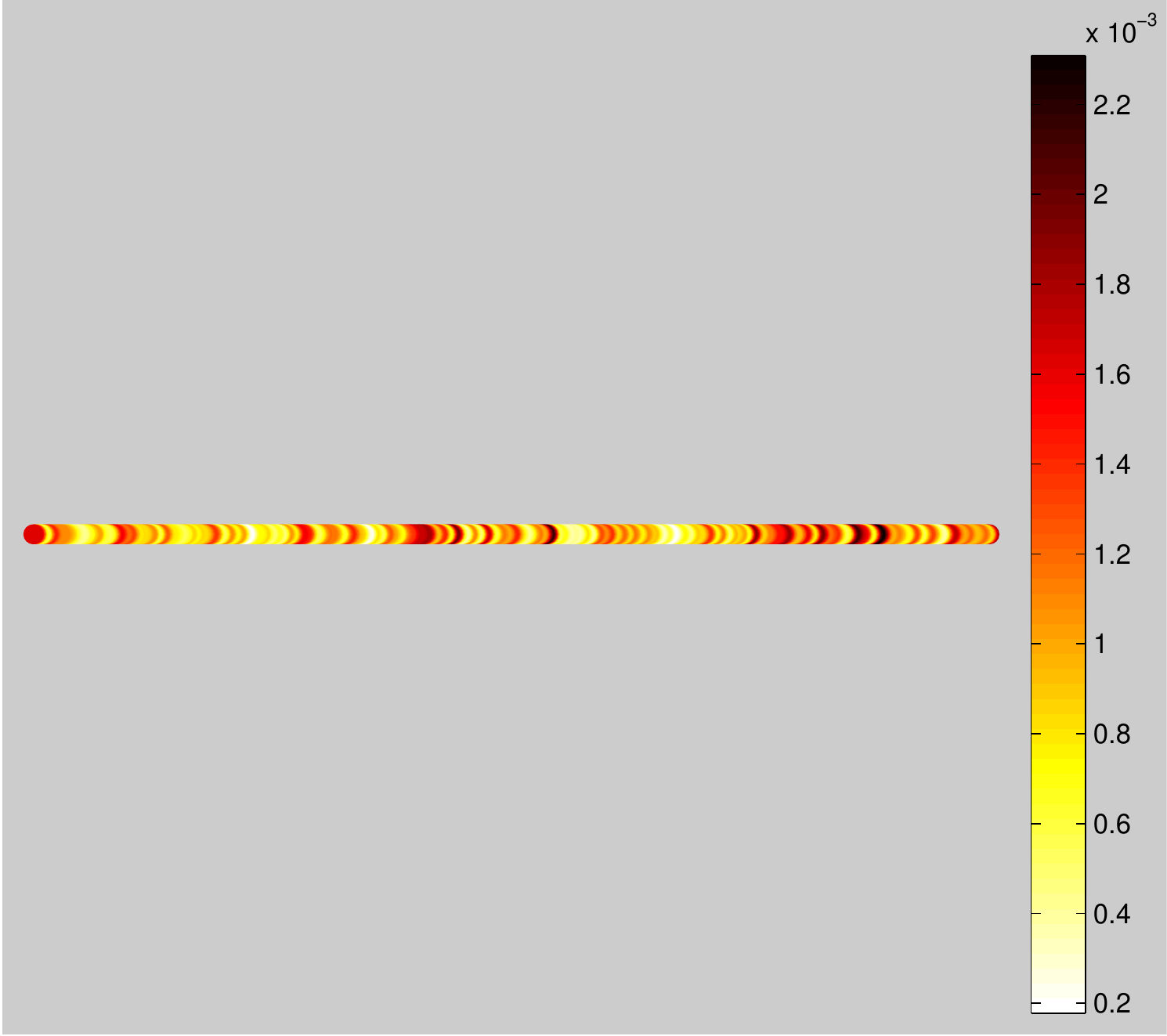}
\includegraphics[height=25mm, keepaspectratio]{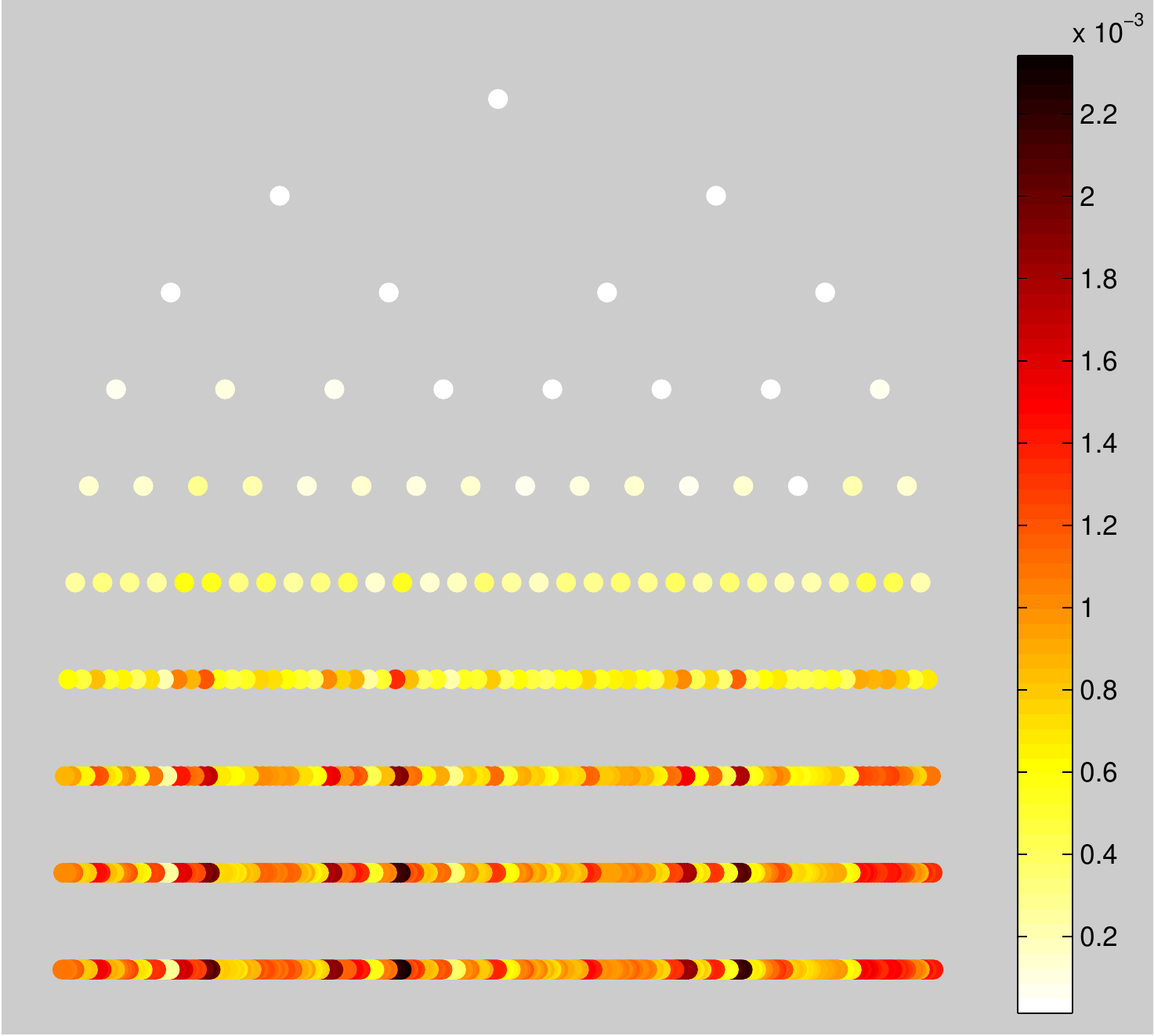}
\caption{\label{fig:sampling_distribution} Optimal and estimated optimal sampling distributions for \MODIF{five} different graphs. Nodes in black are sampled with a higher probability than nodes in white.}
\end{figure}
%

\subsection{Reconstruction of $k$-bandlimited signals}

In this second part, we study experimentally the performance of the decoder \refeq{eq:practical_decoder}. All experiments are repeated for 3 different graphs: a community graph of type $\CCal_5$, the Minnesota graph and the bunny graph. We consider the recovery of $k$-bandlimited signals with band-limit $\nbClass=10$. We take $\nbVertRed = 200$ measurements using the estimated optimal distribution $\tilde{\vec{p}}$. The experiments are conducted with and without noise on the measurements. In the presence of noise, the random noise vector $\err$ follows a zero-mean Gaussian distribution\footnote{For nodes sampled multiple times, the realisation of the noise is thus different each time the same node is sampled. The noise vector $\err$ contains no duplicated entry.} of variance $\sigma^2$. The values of $\sigma$ used are $\{0, 1.5 \cdot 10^{-3}, 3.7 \cdot 10^{-3}, 8.8 \cdot 10^{-3}, 2.1 \cdot 10^{-2}, 5.0 \cdot 10^{-2} \}$. The signals are reconstructed by solving \refeq{eq:practical_decoder} for different values of the regularisation parameter $\reg$ and different functions $g$. For the community graph and the bunny graph, the regularisation parameter $\reg$ varies between $10^{-3}$ and $10^2$. For the Minnesota graph, it varies between $10^{-1}$ and $10^{10}$. For each $\sigma$, $10$ independent random signals of unit norm are drawn, sampled and reconstructed using all possible pairs $(\reg, g)$. Then, we compute the mean reconstruction errors\footnote{See Theorem \ref{th:practical_decoder} for the definition of $\vec{\alpha}^*$ and $\vec{\beta}^*$.} $\norm{\sig^* - \sig}_2$, $\norm{\vec{\alpha}^* - \sig}_2$ and $\norm{\vec{\beta}^*}_2$ over these $10$ signals. In our experiments, the distribution $\tilde{\vec{p}}$ is re-estimated with Algorithm~\ref{alg:optimal_distribution} each time a new signal $\sig$ is drawn.

We present the mean reconstruction errors obtained in the absence of noise on the measurements in Fig.~\ref{fig:reconstruction_nonoise}. In this set of experiments, we reconstruct the signals using $g(\Lap) = \Lap$, then $g(\Lap) = \Lap^2$, and finally $g(\Lap) = \Lap^4$. Before describing these results, we recall that the ratio $g(\eig_{10})/g(\eig_{11})$ decreases as the power of $\Lap$ increases. We observe that all reconstruction errors, $\norm{\sig^* - \sig}_2$, $\norm{\vec{\alpha}^* - \sig}_2$ and $\norm{\vec{\beta}^*}_2$ decrease when the ratio $g(\eig_\nbClass)/g(\eig_{\nbClass+1})$ in the range of small $\reg$, as predicted by the upper bounds on these errors in Theorem~\ref{th:practical_decoder}.

\begin{figure}
\centering
\begin{minipage} {.32\linewidth} \centering \scriptsize \hspace{2mm} Community graph $\CCal_5$ \end{minipage}
\begin{minipage} {.32\linewidth} \centering \scriptsize \hspace{2mm} Bunny graph \end{minipage}
\begin{minipage} {.32\linewidth} \centering \scriptsize \hspace{2mm} Minnesota graph \end{minipage}\\
\includegraphics[width=.32\linewidth, keepaspectratio]{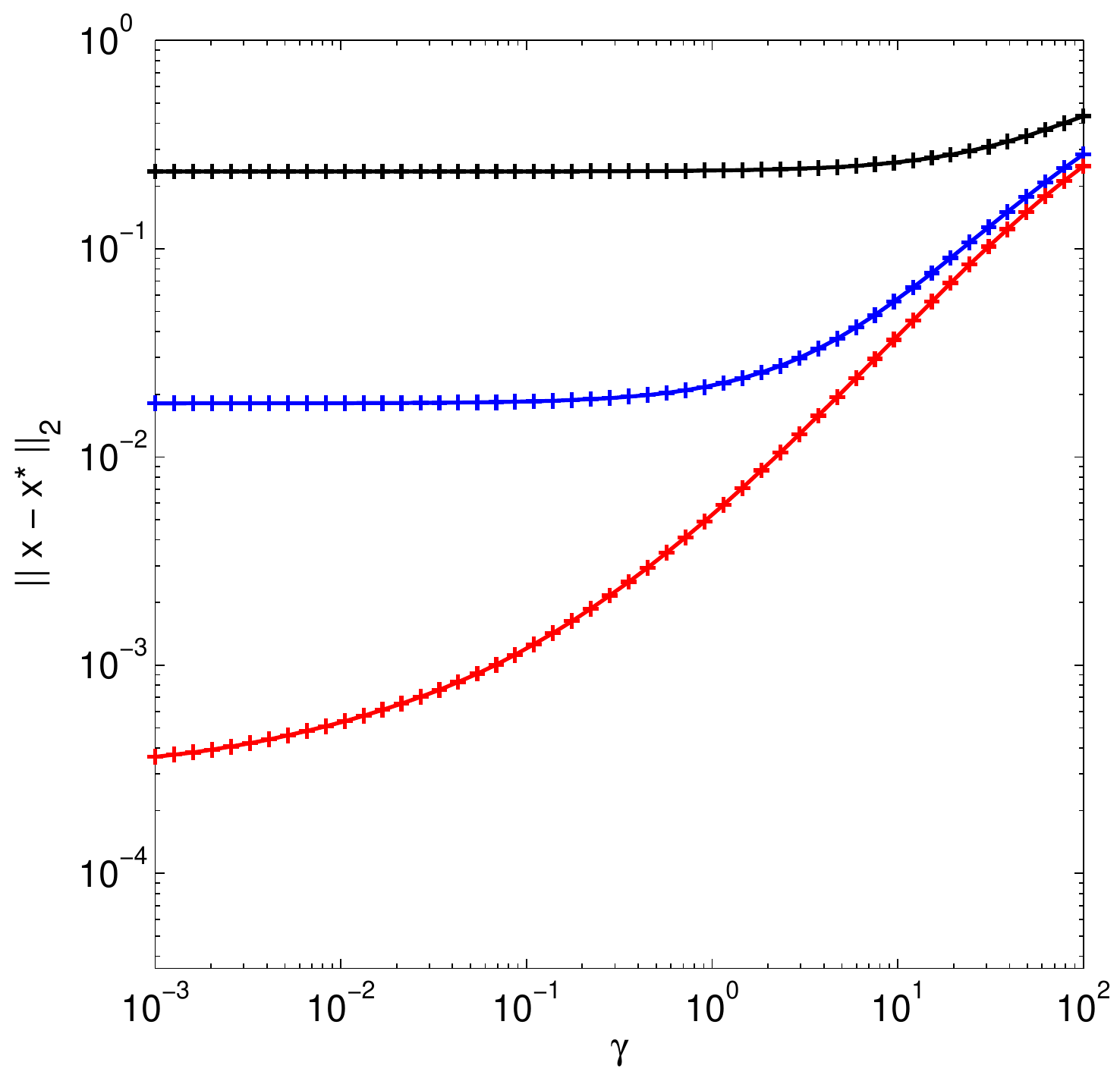}
\includegraphics[width=.32\linewidth, keepaspectratio]{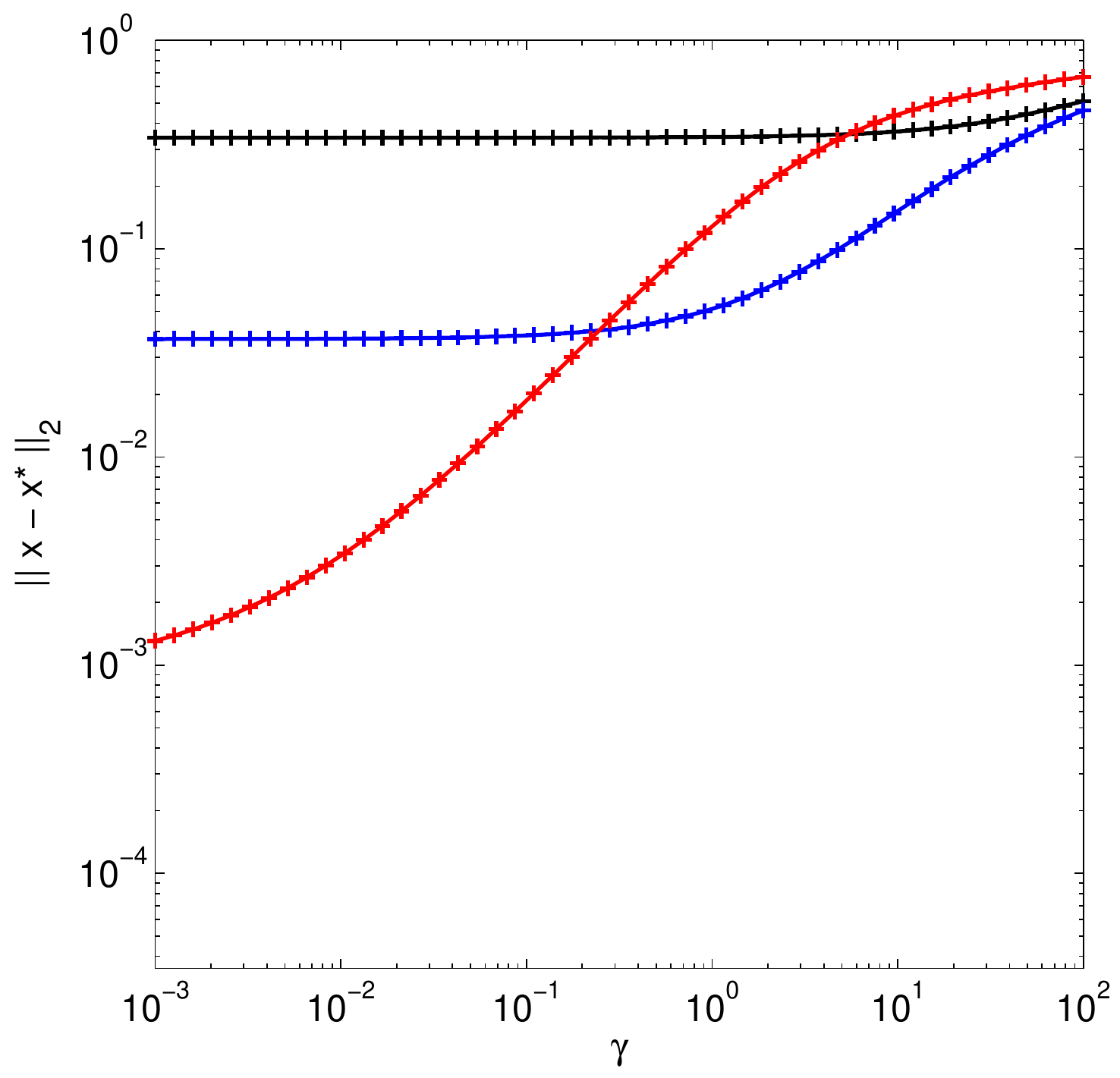}
\includegraphics[width=.32\linewidth, keepaspectratio]{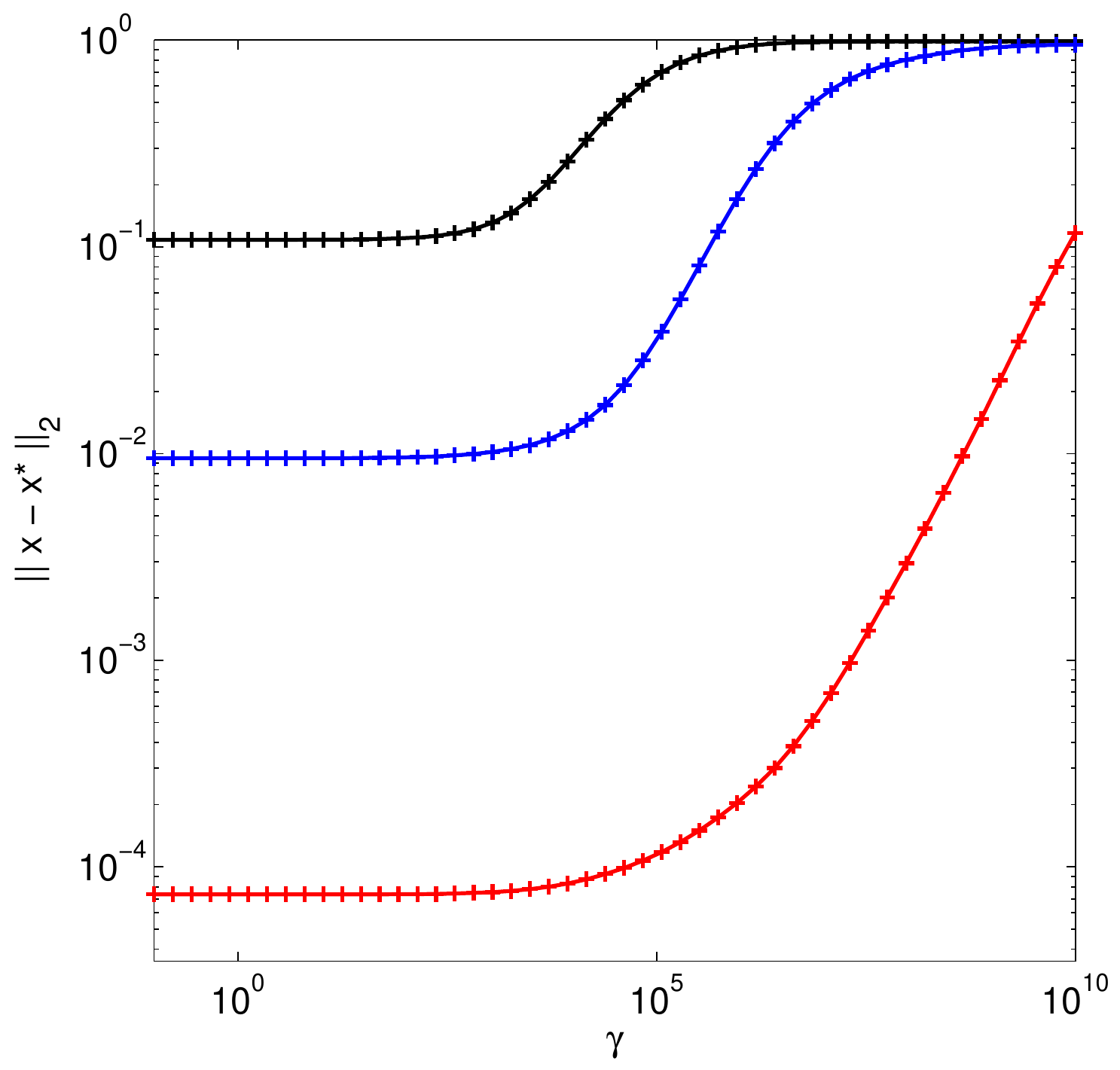}\\
\includegraphics[width=.32\linewidth, keepaspectratio]{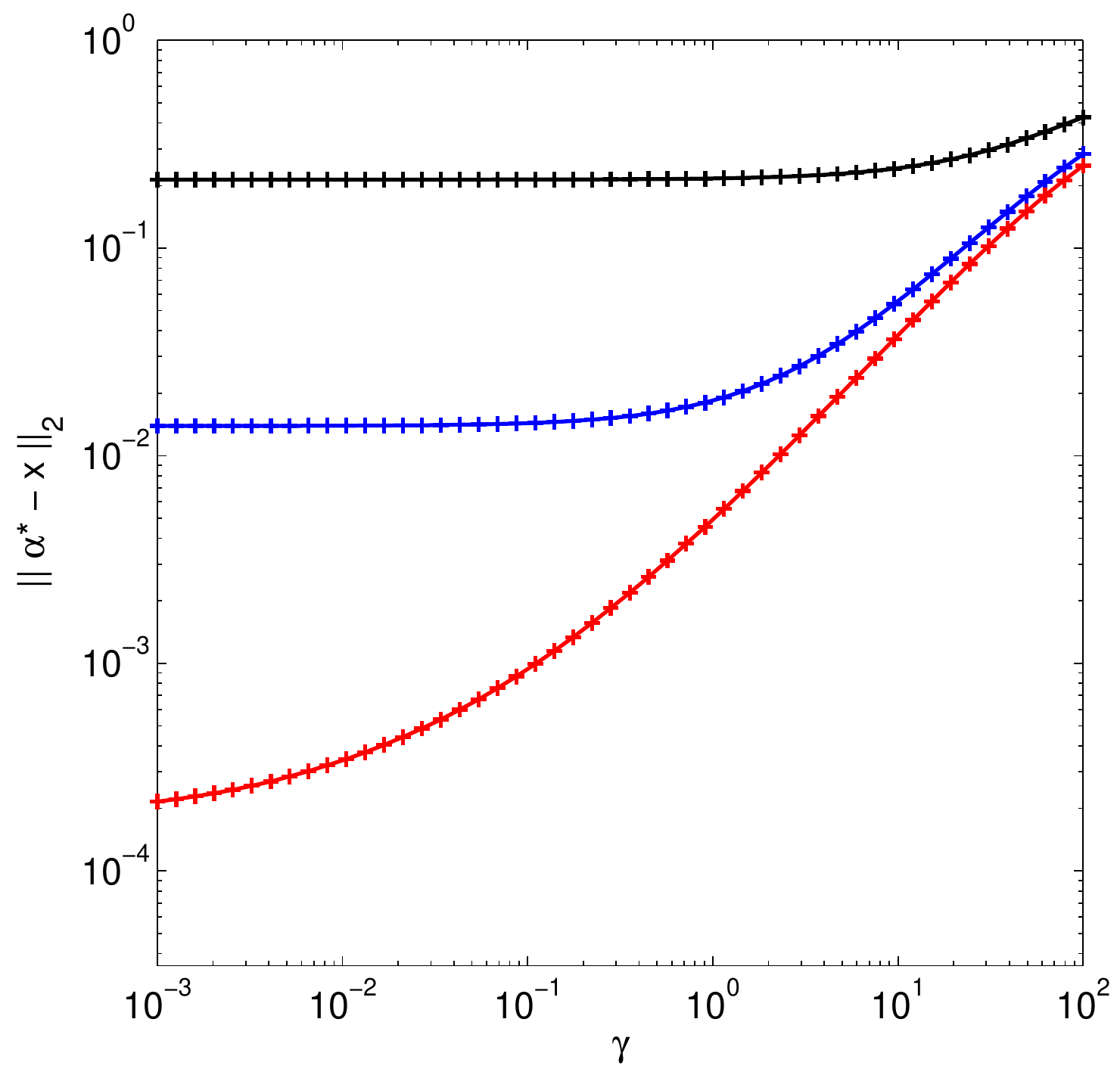}
\includegraphics[width=.32\linewidth, keepaspectratio]{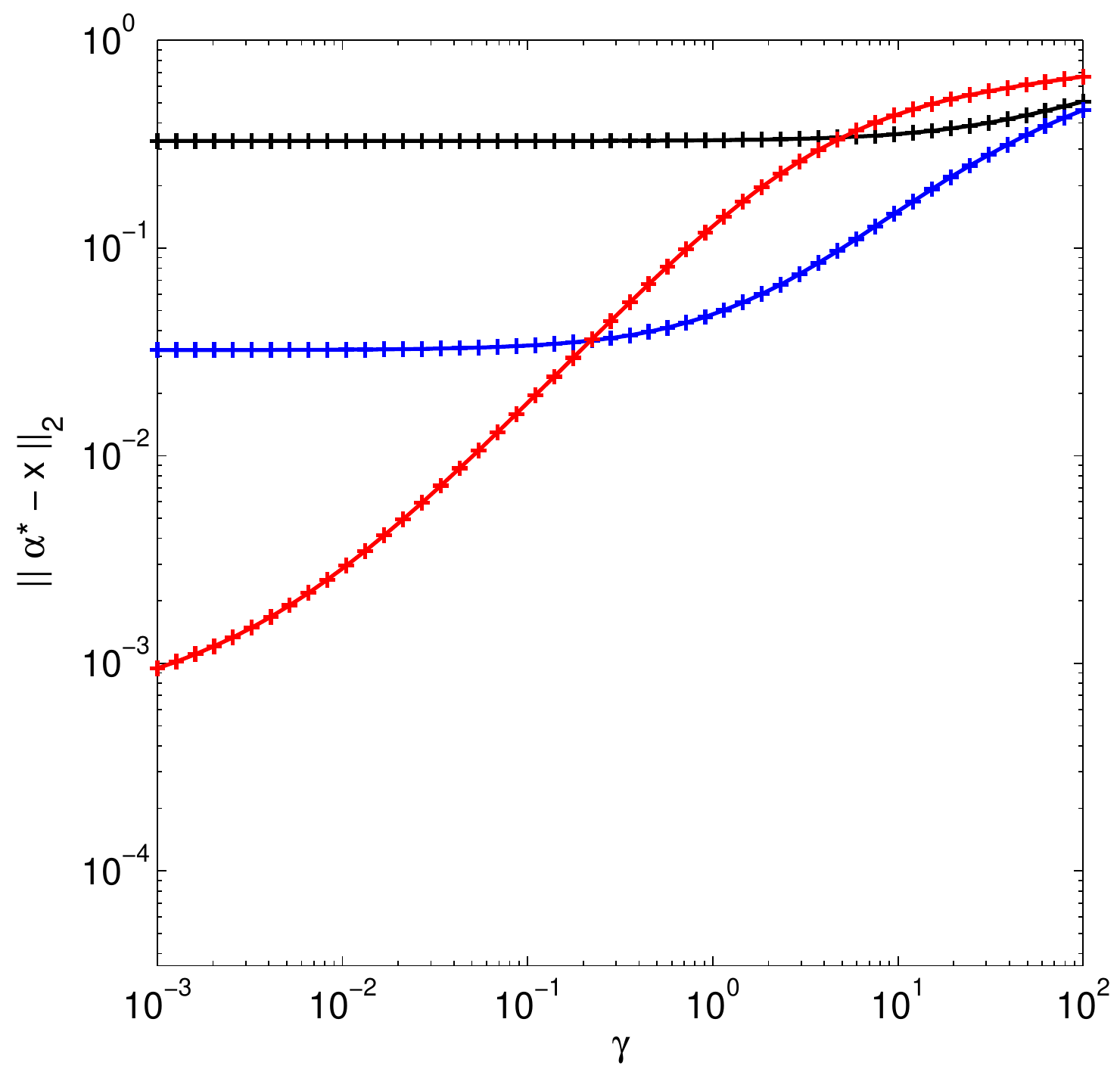}
\includegraphics[width=.32\linewidth, keepaspectratio]{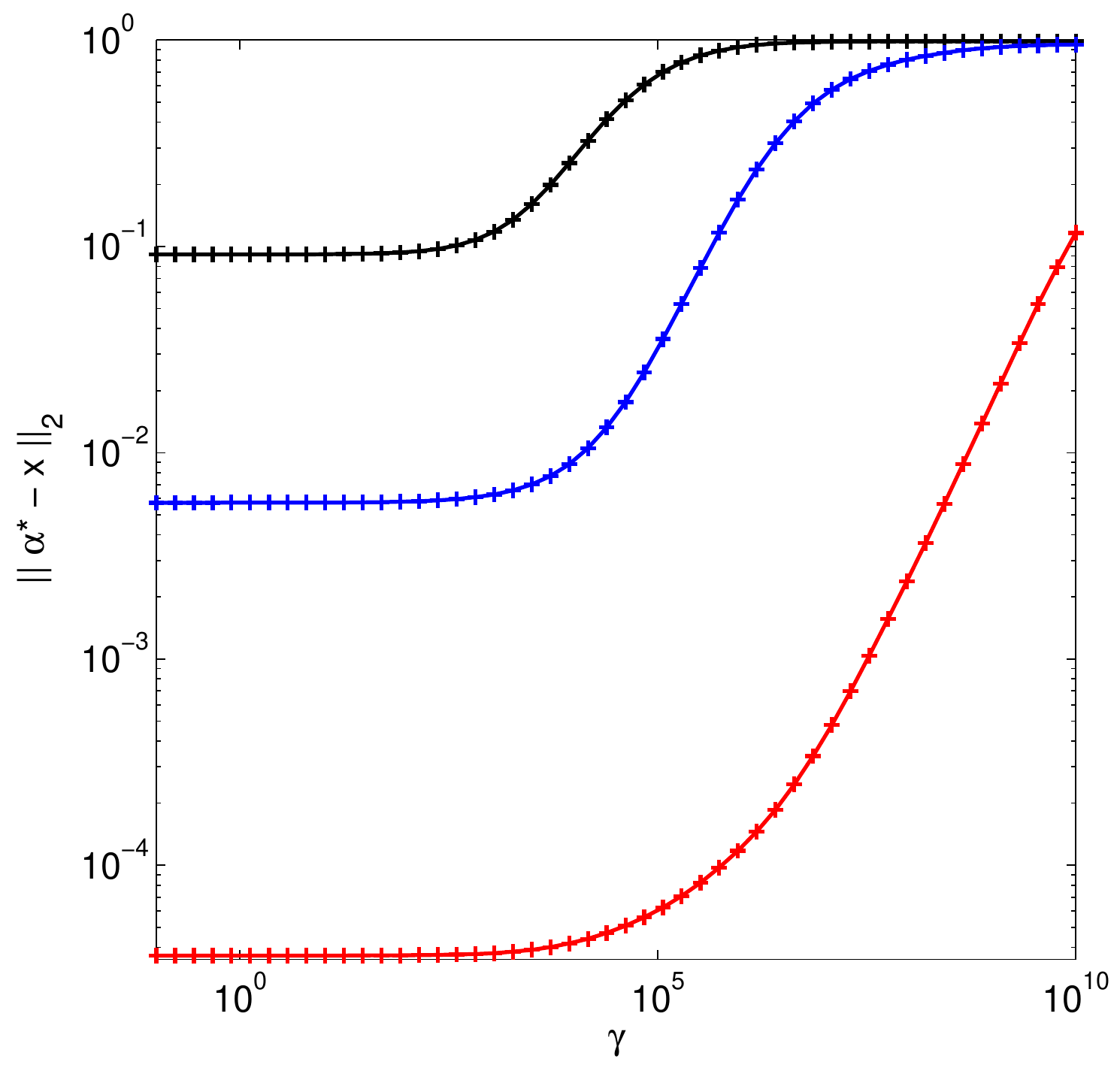}\\
\includegraphics[width=.32\linewidth, keepaspectratio]{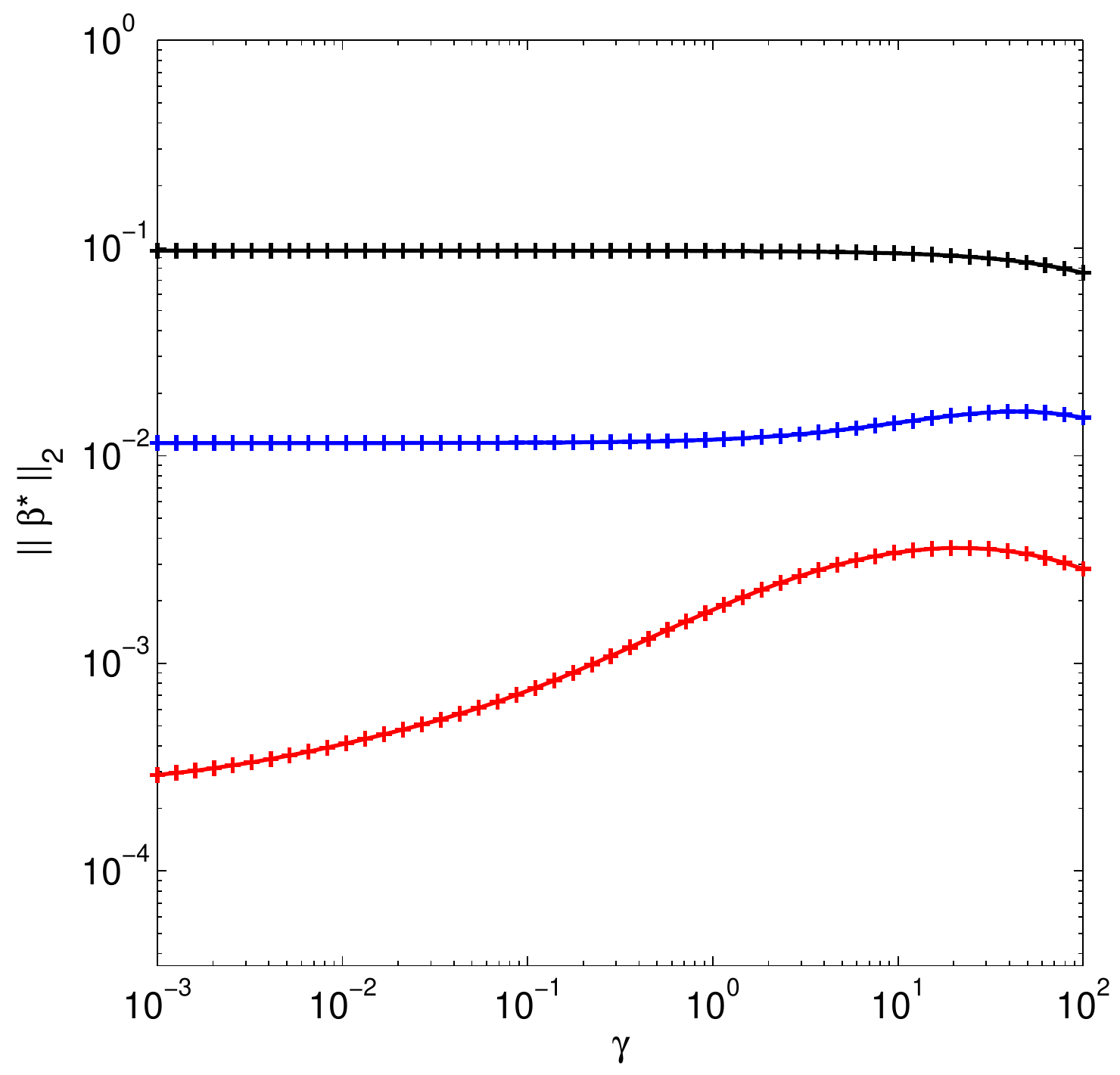}
\includegraphics[width=.32\linewidth, keepaspectratio]{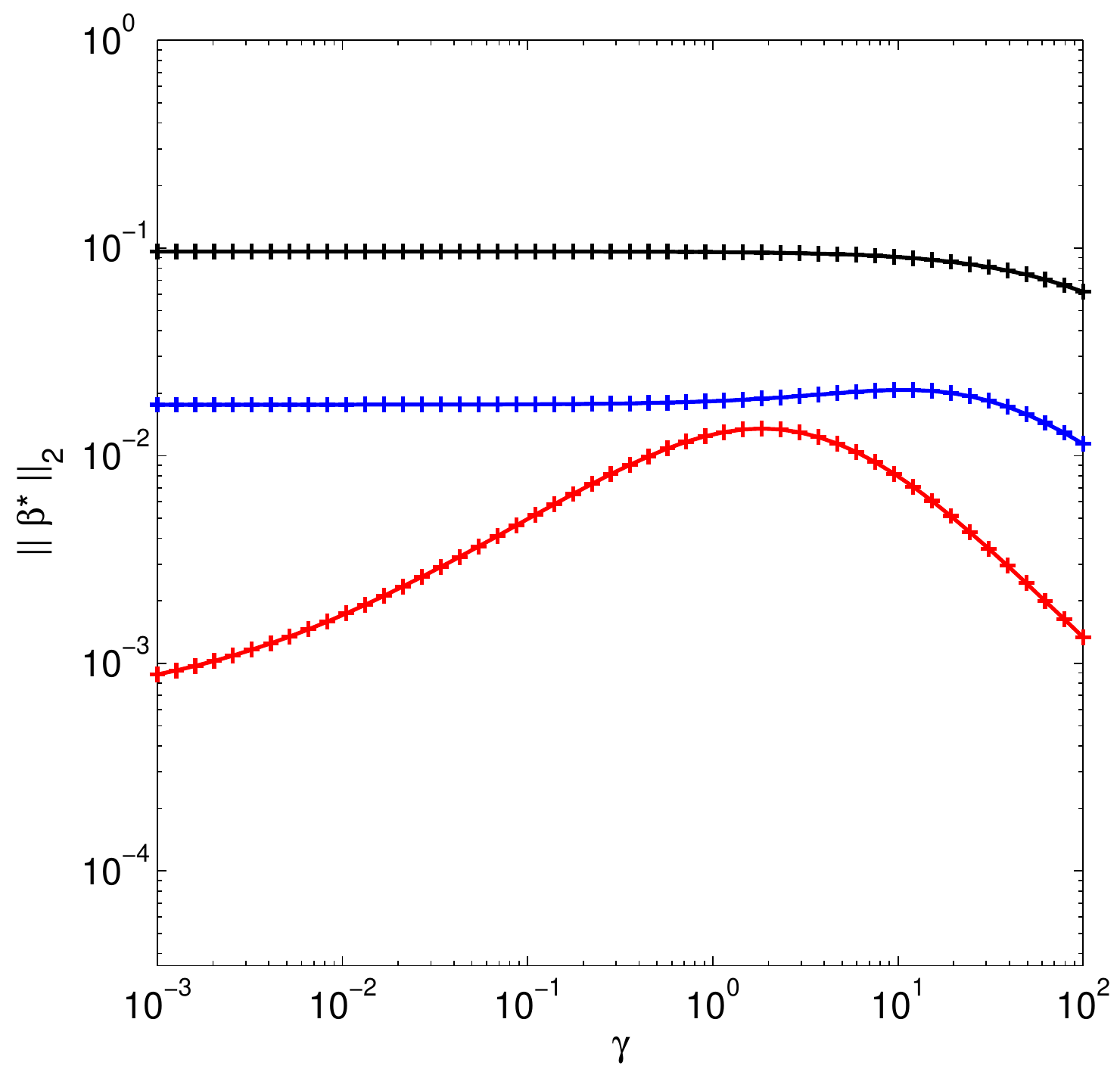}
\includegraphics[width=.32\linewidth, keepaspectratio]{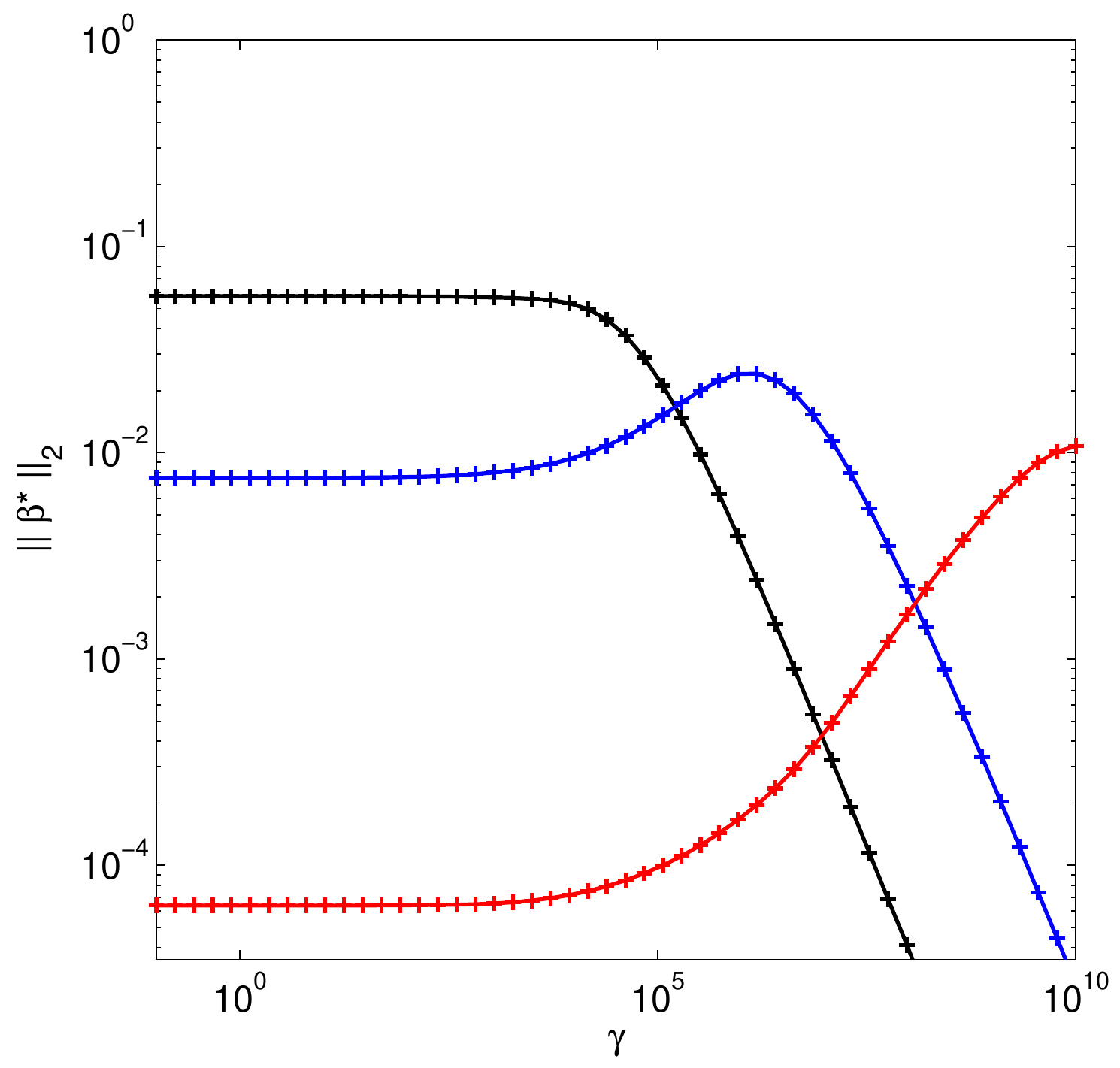}
\caption{\label{fig:reconstruction_nonoise} Mean reconstruction errors of $10$-bandlimited signals as a function of $\reg$. The simulations are performed in the absence of noise. The black curves indicate the results with $g(\Lap) = \Lap$. The blue curves indicate the results with $g(\Lap) = \Lap^2$. The red curves indicate the results with $g(\Lap) = \Lap^4$. The first, second and third columns show the results for a community graph of type $\CCal_5$, the bunny graph, and the Minnesota graph, respectively. The first, second and third rows show the mean reconstruction errors $\norm{\sig^* - \sig}_2$, $\norm{\vec{\alpha}^* - \sig}_2$ and $\norm{\vec{\beta}^*}_2$, respectively.}
\end{figure}

We present the mean reconstruction errors obtained in the presence of noise on the measurements in Fig.~\ref{fig:reconstruction_noise}. In this set of experiments, we reconstruct the signals using $g(\Lap) = \Lap^4$. As expected the best regularisation parameter $\reg$ increases with the noise level.

\begin{figure}
\centering
\begin{minipage} {.32\linewidth} \centering \scriptsize \hspace{2mm} Community graph $\CCal_5$ \end{minipage}
\begin{minipage} {.32\linewidth} \centering \scriptsize \hspace{2mm} Bunny graph \end{minipage}
\begin{minipage} {.32\linewidth} \centering \scriptsize \hspace{2mm} Minnesota graph \end{minipage}\\
\includegraphics[width=.32\linewidth, keepaspectratio]{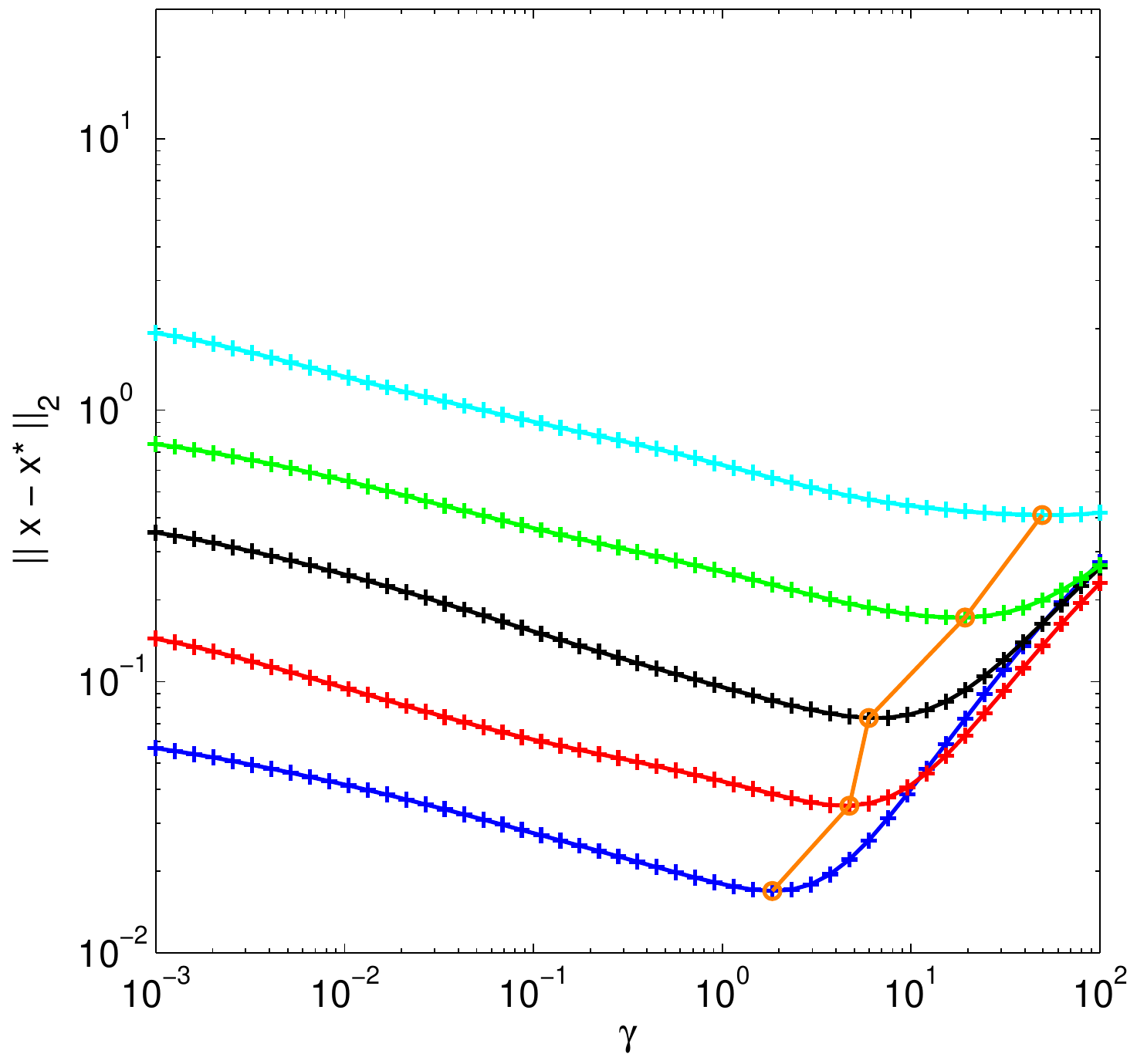}
\includegraphics[width=.32\linewidth, keepaspectratio]{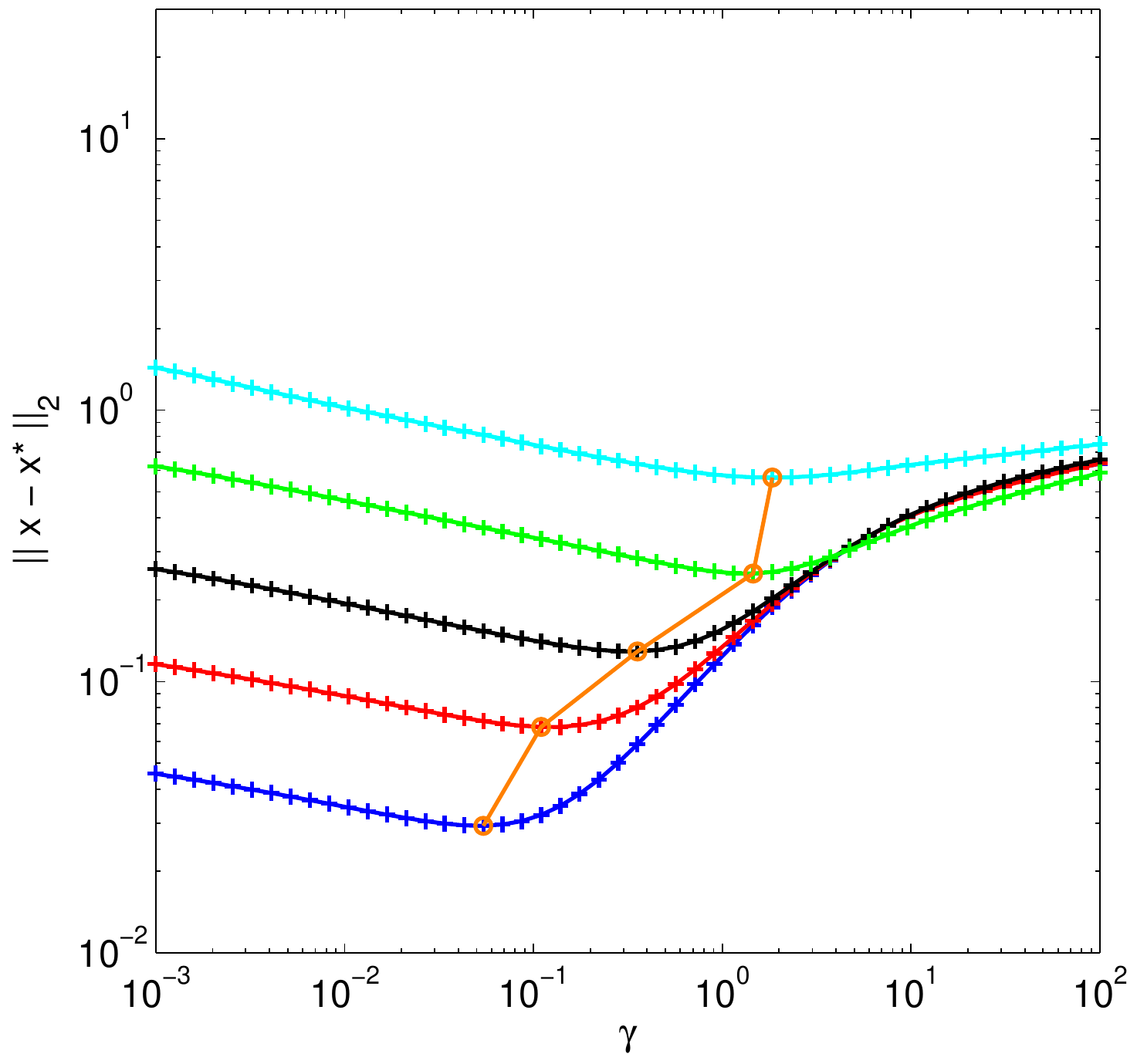}
\includegraphics[width=.32\linewidth, keepaspectratio]{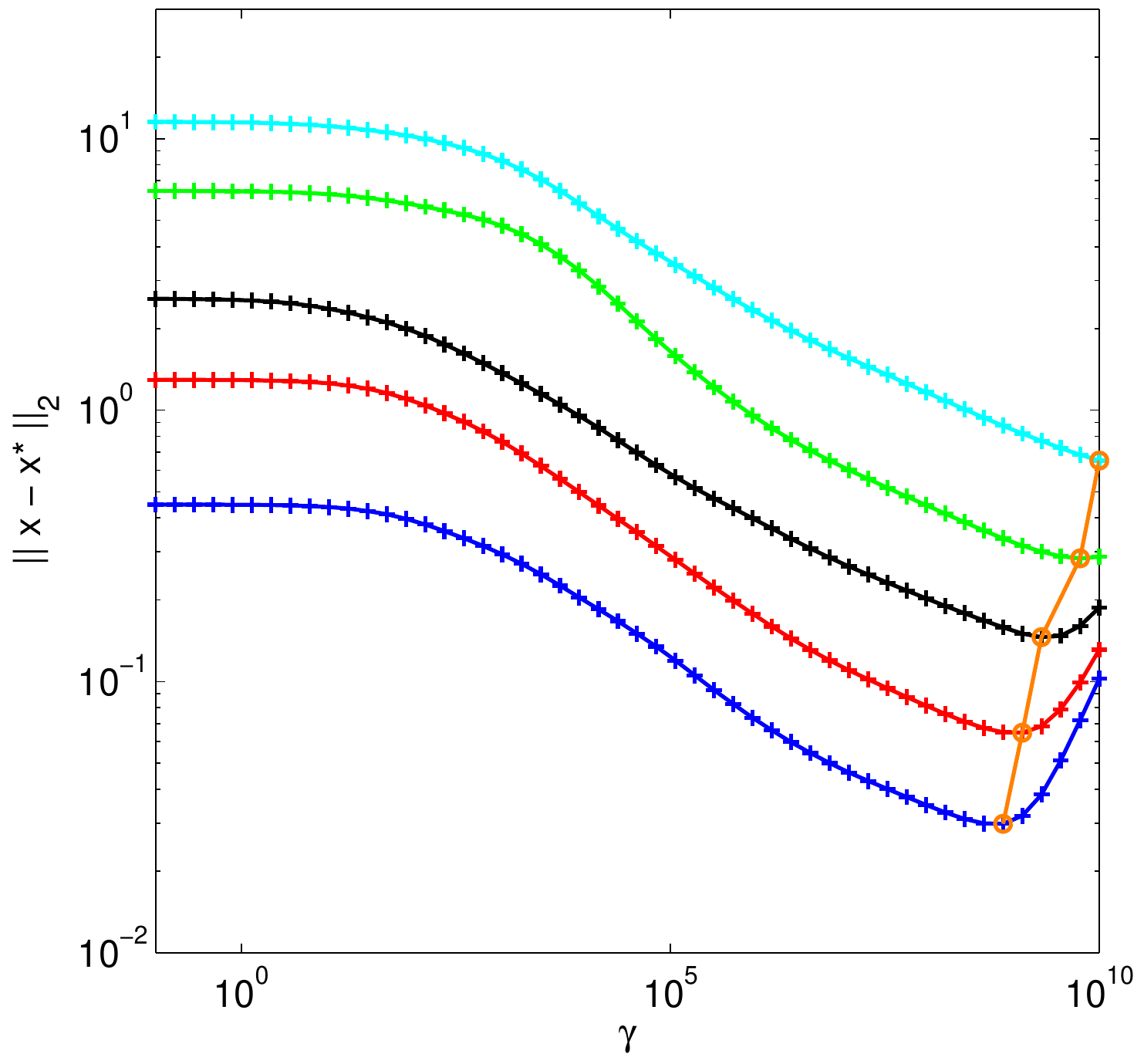}\\
\caption{\label{fig:reconstruction_noise} Mean reconstruction error $\norm{\sig^* - \sig}_2$ of $10$-bandlimited signals as a function of $\reg$ with $g(\Lap) = \Lap^4$. The simulations are performed in presence of noise. The standard deviation of the noise is $0.0015$ (blue), $0.0037$ (red), $0.0088$ (black), $0.0210$ (green), $0.0500$ (cyan). The best reconstruction errors are indicated by orange circles. The first, second and third columns show the results for a community graph of type $\CCal_5$, the bunny graph, and the Minnesota graph, respectively.}
\end{figure}
%

\subsection{Illustration: sampling of a real image}

We finish this experimental section with an example of image sampling using the developed theory.

For this illustration, we use the photo of \emph{Lac d'Emosson} in Switzerland presented in Fig.~\ref{fig:original_image} This RGB image contains $4288 \times 2848$ pixels. We divide this image into patches of $8 \times 8$ pixels, thus obtaining $536 \times 356$ patches of $64$ pixels per RGB channel. Let us denote each patch by $\vec{q}_{i,j,l} \in \Rbb^{64}$ with $i \in \{1, \ldots, 536\}$, $j \in \{1, \ldots, 356\}$, and $l \in \{1, 2, 3\}$. The pair of indices $(i,j)$ encodes the spatial location of the patch and $l$ encodes the color channel. Using these patches, we build the following matrix
\begin{align*}
\ma{X} := \left(
\begin{array}{cccccc}
\vec{q}_{1,1,1} & \vec{q}_{1,2,1} & \ldots & \vec{q}_{2,1,1} & \ldots & \vec{q}_{536,356,1} \\ 
\vec{q}_{1,1,2} & \vec{q}_{1,2,2} & \ldots & \vec{q}_{2,1,2} & \ldots & \vec{q}_{536,356,2} \\ 
\vec{q}_{1,1,3} & \vec{q}_{1,2,3} & \ldots & \vec{q}_{2,1,3} & \ldots & \vec{q}_{536,356,3} \\ 
\end{array}
\right) \in \Rbb^{ 192 \times \nbVert},
\end{align*}
where $\nbVert = 190816$. Each column of $\ma{X}$ represents a color patch of the original image at a given position.

\begin{figure}
\centering
\subfigure[]{\label{fig:original_image}\includegraphics[height=.34\linewidth, keepaspectratio]{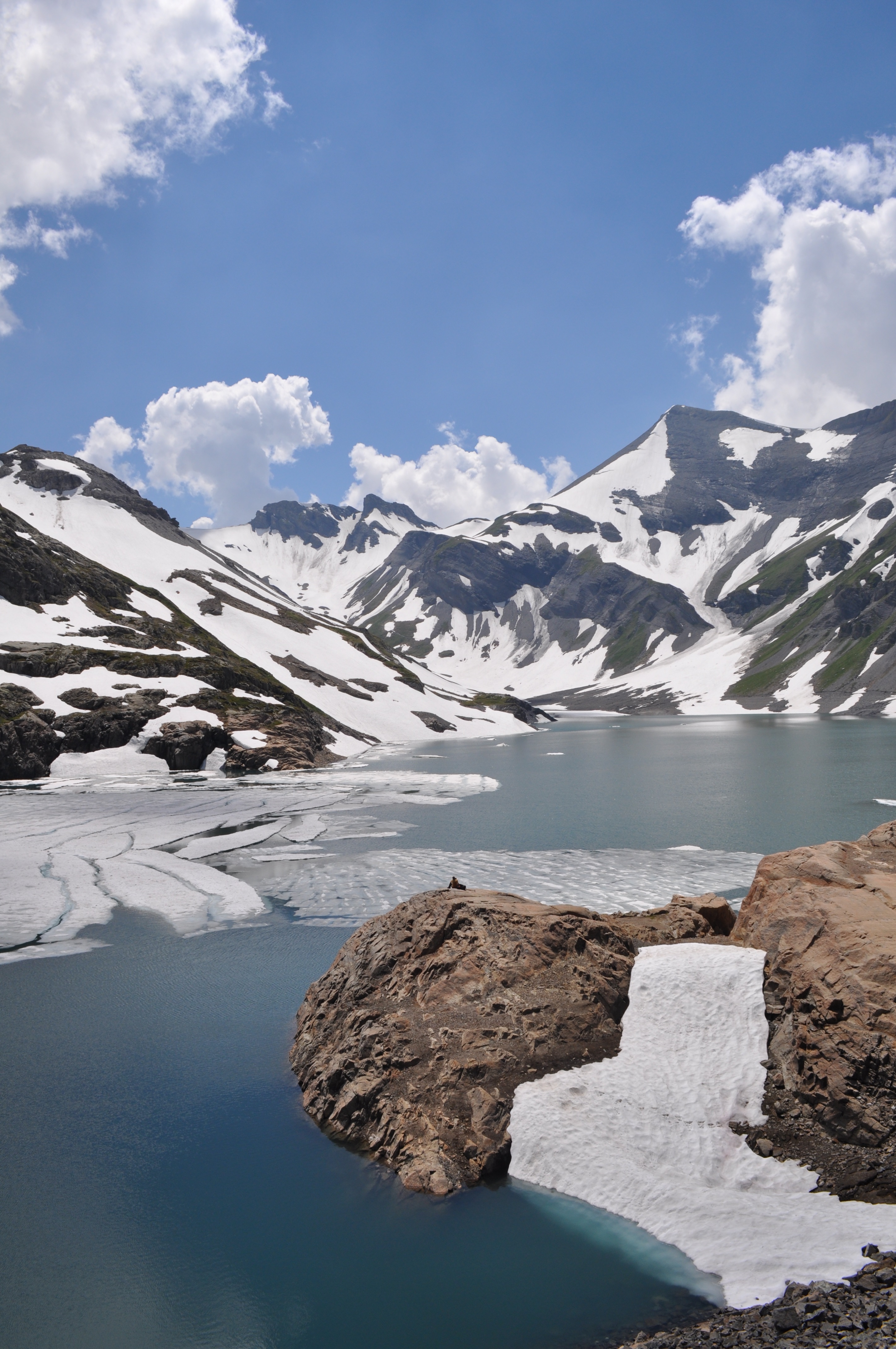}}
\subfigure[]{\label{fig:sampling_distribution_image}\includegraphics[height=.35\linewidth, keepaspectratio]{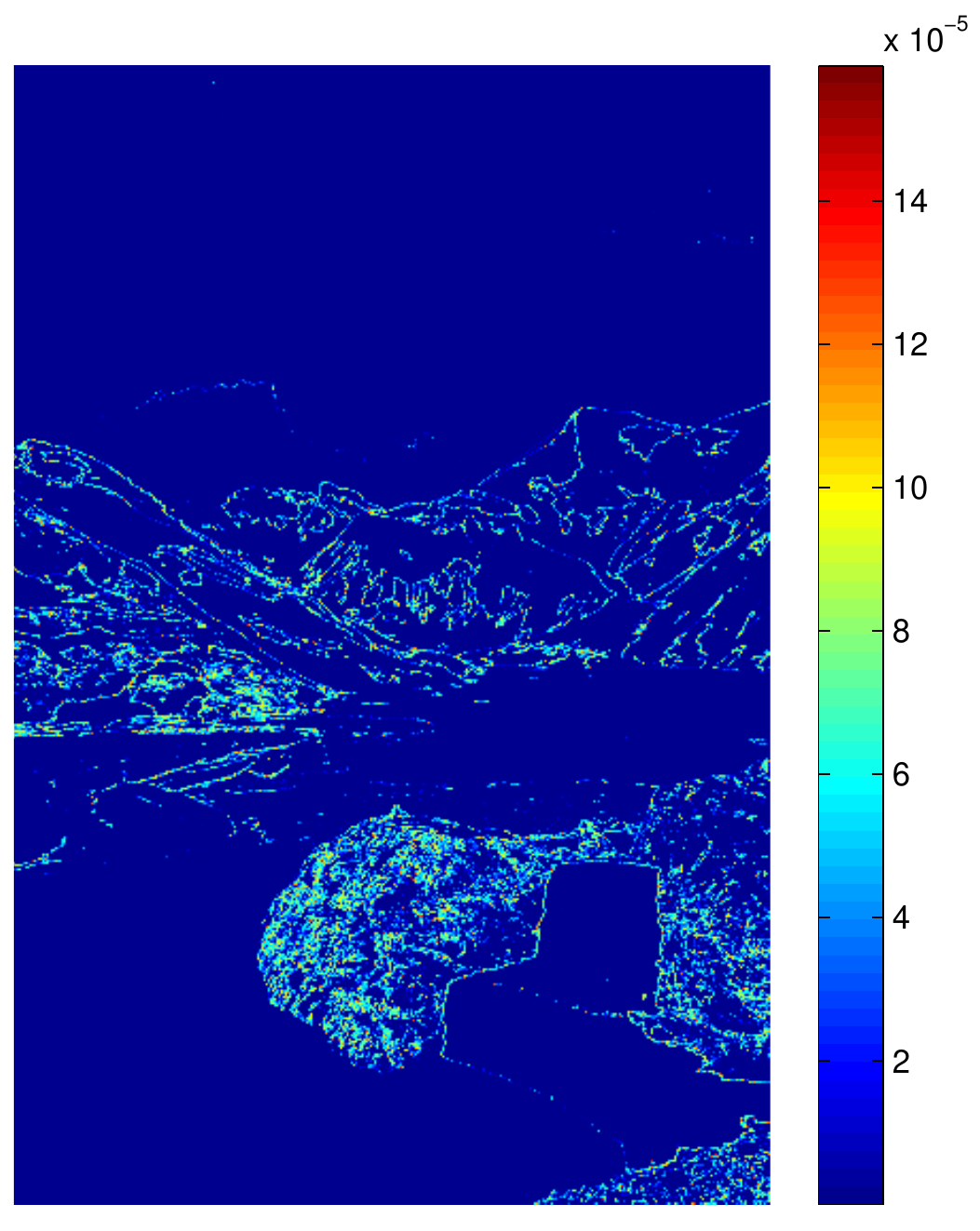}}
\subfigure[]{\label{fig:sampled_image_nonunif}\includegraphics[height=.34\linewidth, keepaspectratio]{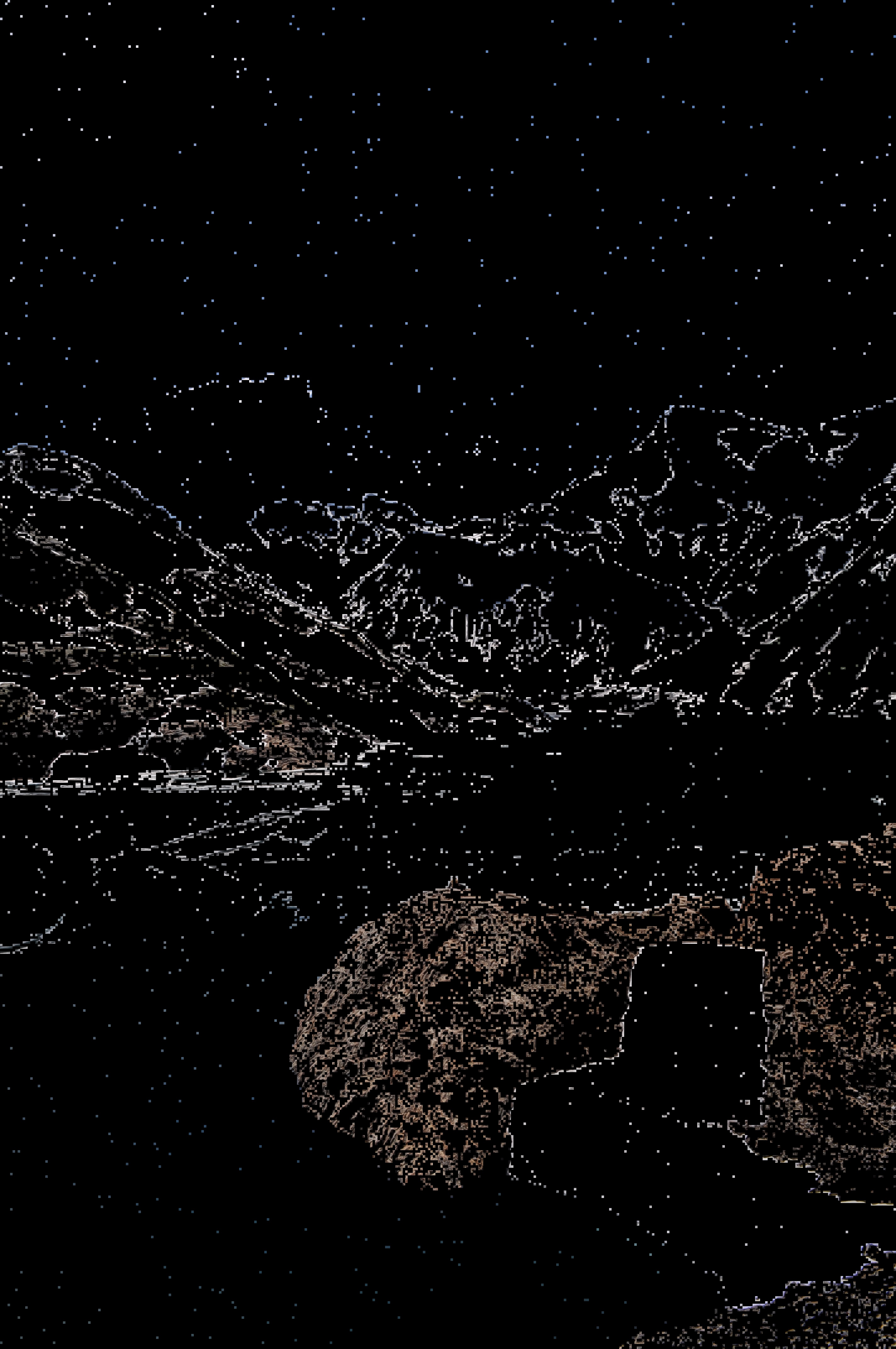}}
\subfigure[]{\label{fig:sampled_image_unif}\includegraphics[height=.34\linewidth, keepaspectratio]{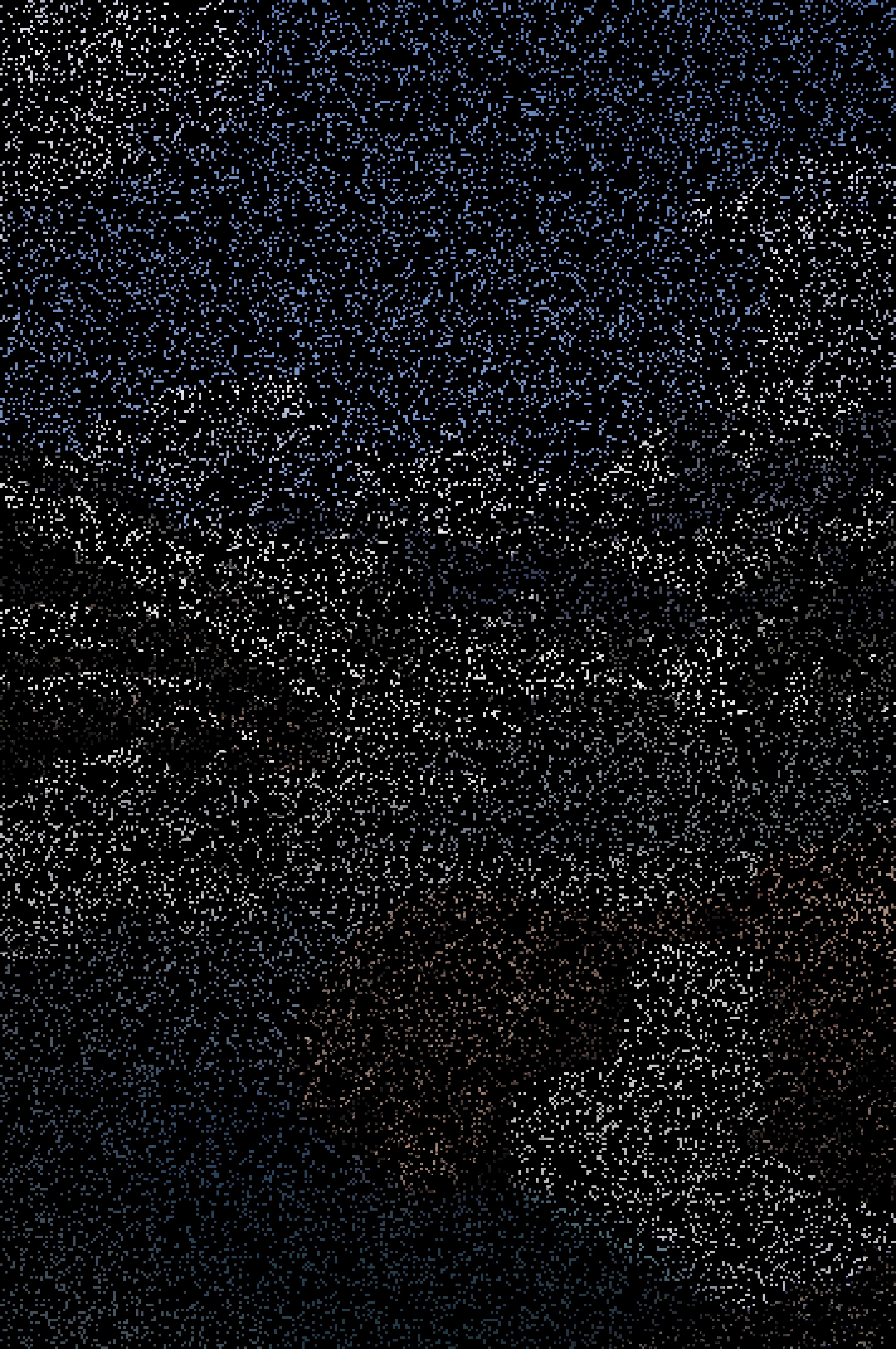}}
\caption{a) original image; b) estimated optimal sampling distribution $\tilde{\vec{p}}$; c) sampled image using $\tilde{\vec{p}}$; d) sampled image using the uniform sampling distribution $\vec{\pi}$. The sampled images are obtained using the same number of measurements.}
\end{figure}

We continue by building a graph modelling the similarity between the columns of $\ma{X}$. Let $\sig_i \in \Rbb^{192}$ be the $i^\th$ column-vector of the matrix $\ma{X}$. For each $i \in \{1, \ldots, \nbVert \}$, we search for the $20$ nearest neighbours of $\sig_i$ among all other columns of $\ma{X}$. Let $\sig_j \in \Rbb^{192}$ be a vector connected to $\sig_i$. The weight $\ma{W}_{ij}$ of the weighted adjacency matrix $\ma{W} \in \Rbb^{\nbVert \times \nbVert}$ satisfies
\begin{align*}
\ma{W}_{ij} := \exp \left(- \frac{\norm{\sig_i  - \sig_j}_2^2}{2\sigma^2} \right),
\end{align*}
where $\sigma > 0$ is the standard deviation of all Euclidean distances between pairs of connected columns (patches). We then symmetrise the matrix $\ma{W}$. Each column of $\ma{X}$ is thus connected to at least $20$ other columns after symmetrisation. We finish the construction of the graph by computing the combinatorial Laplacian $\Lap \in \Rbb^{\nbVert \times \nbVert}$ associated to $\ma{W}_{ij}$.

We sample $\ma{X}$ by measuring about $15\%$ of its $\nbVert$ columns: $\nbVertRed = 28622 \approx 0.15 \, \nbVert$. First, we estimate the optimal sampling distribution $\tilde{\vec{p}}$ for $\nbClass = 9541 \approx \nbVertRed/3$ with Algorithm~\ref{alg:optimal_distribution}. It takes about $4$ minutes to compute $\tilde{\vec{p}}$ using Matlab on a laptop with a 2,8 GHz Intel Core i7 and 16GB of RAM. In comparison, we tried to compute $\vec{p}^*$ exactly by computing $\Fou_{9541}$ but stopped Matlab after $30$ minutes of computations. The estimated sampling distribution is presented in Fig.~\ref{fig:sampling_distribution_image}. Then, we build the sampling matrix $\Meas$ by drawing at random $\nbVertRed$ independent indices from $\{1, \ldots, \nbVert \}$ according to $\tilde{\vec{p}}$. Note that the \emph{effective sampling rate}, \ie, once the indices sampled multiple times are removed, is about $7.6\%$, only. The sampled columns are denoted by $\ma{Y} \in \Rbb^{192 \times \nbVertRed}$ and satisfy $\ma{Y}^\adjoint = \Meas \ma{X}^\adjoint$. 

We present in Fig.~\ref{fig:sampled_image_nonunif} the sampled image, where all non-sampled pixels appear in black. We remark that the regions where many patches are similar (sky, lake, snow) are very sparsely sampled. This can be explained as follows. The patches in such a region being all similar, one can fill this region by copying a single representative patch. In practice this is done via the Laplacian matrix, which encodes the similarities between the patches, by solving \refeq{eq:practical_decoder}.

We reconstruct the image by solving \refeq{eq:practical_decoder} for each column of $\ma{Y}$ with $\reg = 1$ and $g(\Lap) = \Lap$. Reconstructing the image takes about $3$ minutes by solving~\refeq{eq:exact_solution} using the {\sf mldivide} function of Matlab. We show the reconstructed image in Fig.~\ref{fig:reconstruction_image}. One can notice that we obtain a very accurate reconstruction of the original image. The SNR between the original and the reconstructed images is $27.76$ dB. As a comparison, we present in Fig.~\ref{fig:reconstruction_image} the reconstructed image from, again, $\nbVertRed = 28622 \approx 0.15 \, \nbVert$ measurements but obtained using the uniform sampling distribution. The effective sampling ratio in this case is about $14\%$. The associated sampled image is presented in Fig.~\ref{fig:sampled_image_unif}. The SNR between the original and the reconstructed images is $27.10$ dB. The estimated optimal sampling distribution $\tilde{\vec{p}}$ allows us to attain a better image quality with an effective sampling ratio almost twice smaller.

\begin{figure}
\centering
\begin{minipage}{.32\linewidth} \centering \scriptsize Original \end{minipage}
\begin{minipage}{.32\linewidth}  \centering \scriptsize Reconstructed (sampling with $\tilde{\vec{p}}$)\end{minipage}
\begin{minipage}{.32\linewidth}  \centering \scriptsize Reconstructed (sampling with $\vec{\pi}$)\end{minipage}
\includegraphics[width=.32\linewidth, keepaspectratio]{original.jpg}
\includegraphics[width=.32\linewidth, keepaspectratio]{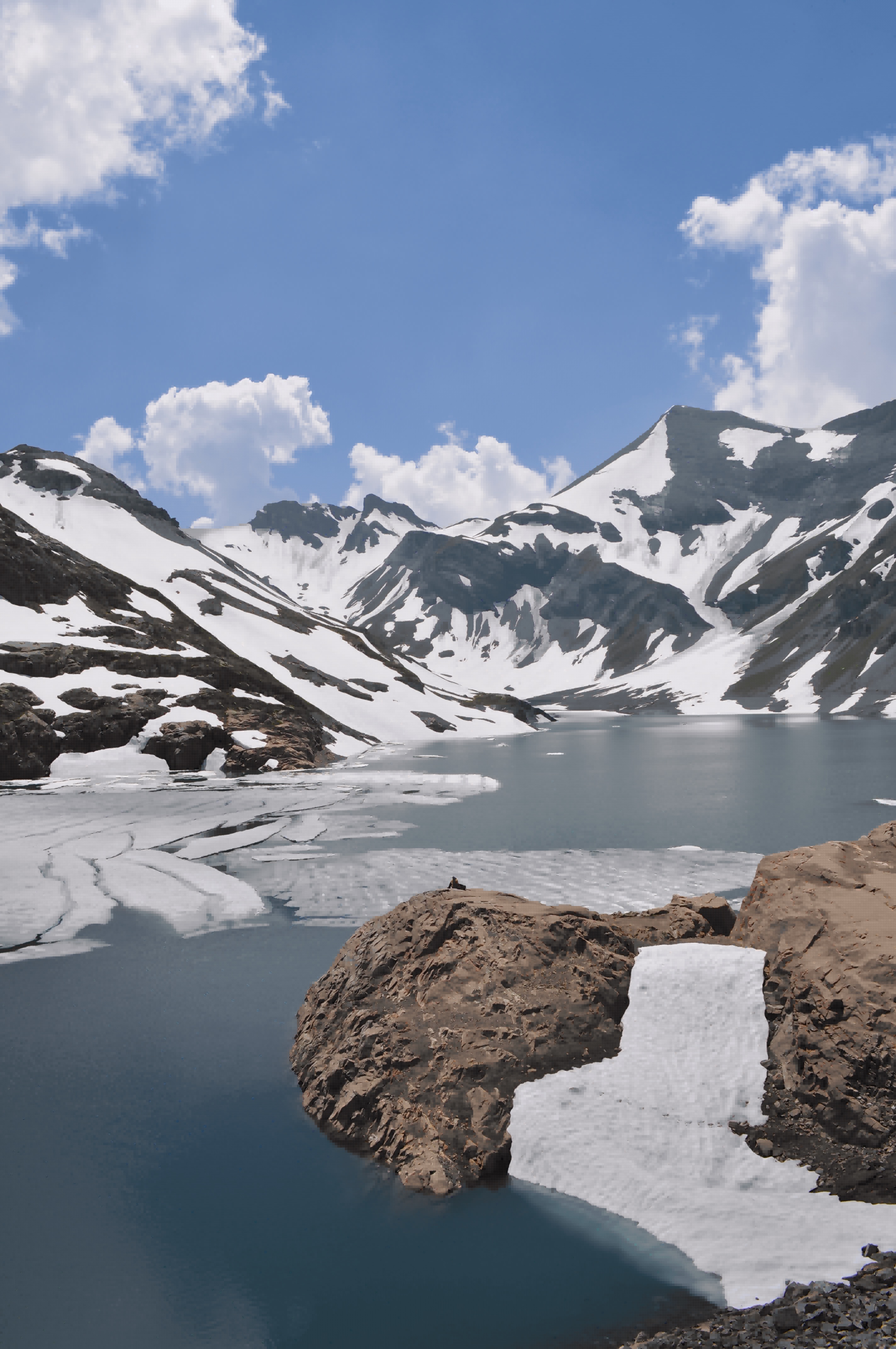}
\includegraphics[width=.32\linewidth, keepaspectratio]{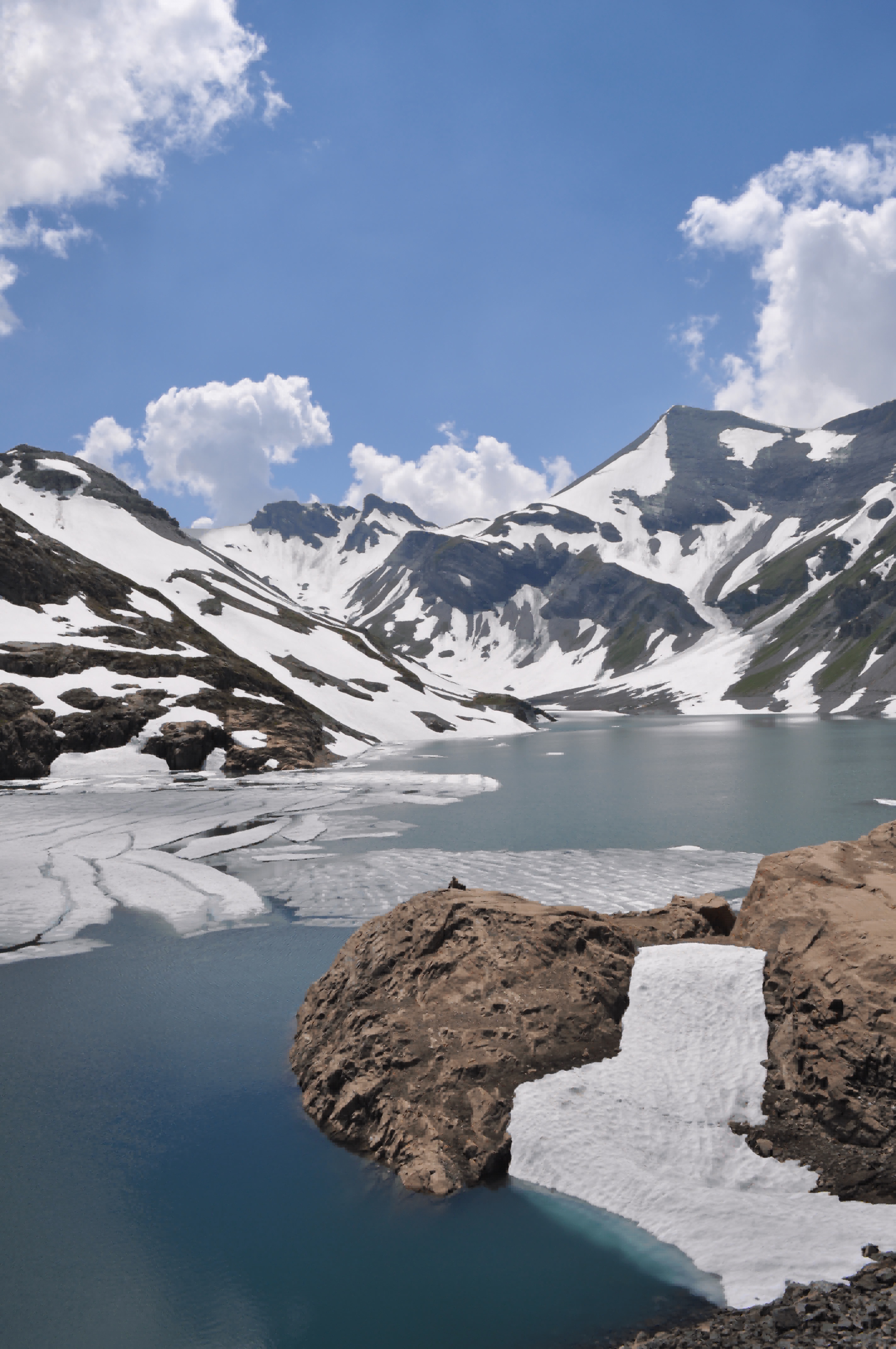}
\caption{\label{fig:reconstruction_image} From left to right: original image; reconstructed image from the measurements obtained with $\tilde{\vec{p}}$ (the reconstruction SNR is $27.76$ dB); reconstructed image from the measurements obtained with $\vec{\pi}$ (the reconstruction SNR is $27.10$ dB).}
\end{figure}
%

\section{Conclusion and perspectives}
\label{sec:conclusion}
We proposed two efficient sampling procedures for $\nbClass$-bandlimited signals defined on the nodes of a graph $\Graph$. The performance of these sampling techniques is governed by the graph weighted coherence, which characterises the interaction between the sampling distribution and the localisation of the first $\nbClass$ Fourier modes over the nodes of $\Graph$. For regular graph with non-localised Fourier modes and a uniform sampling distribution, we proved that $O(\nbClass\log{\nbClass})$ samples are sufficient to embed the set of $\nbClass$-bandlimited signals. For arbitrary graphs, uniform sampling might perform very poorly. In such cases, we proved that it is always possible to adapt the sampling distribution to the structure of the graph and reach optimal sampling conditions. We designed an algorithm to estimate the optimal sampling distribution rapidly. Finally, we proposed an efficient decoder that provides accurate and stable reconstruction of $k$-bandlimited signals from their samples.

We believe that the sampling method developed in this work can be used to speed up computations in multiple applications using graph models. Let us take the example of the fast robust PCA method proposed in \cite{shahid15}. In this work, the authors consider the case where one has access to two graphs $\Graph_1$ and $\Graph_2$ that respectively model the similarities between the rows and the columns of a matrix $\ma{X}$. In this context, they propose an optimisation technique that provides a low-rank approximation of $\ma{X}$. We denote this low-rank approximation by $\ma{X}^*$. The intuition is that the left singular vectors and the right singular vectors of $\ma{X}^*$ live respectively in the span of the first eigenvectors of $\Lap_1$ and $\Lap_2$, the Laplacians associated to $\Graph_1$ and $\Graph_2$. Therefore, the singular vectors of $\ma{X}^*$ can be drastically subsampled using our sampling method. The low-rank matrix $\ma{X}^*$ can be reconstructed from a subset of its rows and columns. Instead of estimating $\ma{X}^*$ from the entire matrix $\ma{X}$, one could thus first reduce the dimension of the problem by selecting a small subset of the rows and columns of $\ma{X}$. 

In semi-supervised learning, a small subset of nodes are labeled and the goal is to infer the label of all nodes. Advances in sampling of graph signals give insight on which nodes should be preferentially observed to infer the labels on the complete graphs. Similarly, in spectral graph clustering, cluster assignments are well approximated by $k$-bandlimited signals and can therefore be heavily subsampled. This leaves the possibility to initially cluster a small subset of the nodes and infer the clustering solution on the complete graphs afterwards, as we have proposed in~\cite{Tremblay_ICML2016}. 

Sensor networks provide other applications of our sampling methods. Indeed, if signals measured by a network of sensors are smooth, one can deduce \textit{beforehand} from the structure of the network which sensors to sample in priority in an active sampling strategy, using the optimal or estimated sampling distribution.

%% file: manuscript/appendix.tex

\ifels
\appendix 
\else
\appendix  \titleformat{\section}[hang]{\color{blue1}\large\bfseries\centering}{Appendix \thesection}{0mm}{}[]
\fi

%
\section{ \ifels\else - \fi Proof of the theorems in Section \ref{sec:rip}}
\label{app:proof_rip}

We start with the proof of Theorem \ref{th:rip}. For this proof, we need the following result obtained by Tropp in \cite{tropp12}.

\begin{lemma}[Theorem $1.1$, \cite{tropp12}]
\label{th:matrix_chernoff}
Consider a finite sequence $\{ \ma{X}_i \}$ of independent, random, self-adjoint, positive semi-definite matrices of dimension $d \times d$. Assume that each random matrix satisfies
\begin{align*}
\lambda_{\rm max}(\ma{X}_i) \leq R\quad \text{almost surely}.
\end{align*}
Define
\begin{align*}
\mu_{\rm min} := \lambda_{\rm min} \left( \sum_i \Ebb \, \ma{X}_i \right)
\quad \text{and} \quad
\mu_{\rm max} := \lambda_{\rm max} \left( \sum_i \Ebb \, \ma{X}_i \right).
\end{align*}
Then
\begin{align*}
\Pbb \left\{ \lambda_{\rm min} \left( \sum_i \ma{X}_i \right)  \leq (1 - \delta) \mu_{\rm min} \right\} 
& \; \leq \; d \cdot \left[ \frac{\ee^{-\delta}}{(1-\delta)^{1-\delta}}\right]^{\mu_{\rm min}/R} \; \text{ for } \delta \in [0, 1], \; \text { and}\\
\Pbb \left\{ \lambda_{\rm max} \left( \sum_i \ma{X}_i \right)  \geq (1 + \delta) \mu_{\rm max} \right\} 
& \; \leq \; d \cdot \left[ \frac{\ee^{\delta}}{(1+\delta)^{1+\delta}}\right]^{\mu_{\rm max}/R} \; \text{ for } \delta \geq 0.
\end{align*}
\end{lemma}

We also need the following facts. For all $\delta \in [0, 1]$, we have
\begin{align*}
\left[ \frac{\ee^{-\delta}}{(1-\delta)^{1-\delta}}\right]^{\mu_{\rm min}/R} 
\; \leq \; \exp \left( - \frac{\delta^2 \mu_{\rm min}}{3 \, R}\right)
\; \text{and} \;
\left[ \frac{\ee^{\delta}}{(1+\delta)^{1+\delta}}\right]^{\mu_{\rm max}/R} 
\; \leq \; \exp \left( - \frac{\delta^2 \mu_{\rm max}}{3 \, R}\right).\\
\end{align*}
\begin{proof}[Proof of Theorem \ref{th:rip}]
As the $i^\th$ row-vector of $\Meas \ma{P}^{-1/2} \Fou_\nbClass$ is $\vec{\delta}_{\omega_i}^\adjoint \Fou_\nbClass/\sqrt{\prob_{\omega_i}}$, we have
\begin{align*}
\inv{\nbVertRed} \; \Fou_\nbClass^\adjoint \ma{P}^{-1/2} \Meas^\adjoint \Meas \ma{P}^{-1/2} \Fou_\nbClass  
= \sum_{i = 1}^\nbVertRed \frac{\left(\Fou_\nbClass^\adjoint \vec{\delta}_{\omega_i} \right) \left(\vec{\delta}_{\omega_i}^\adjoint \Fou_\nbClass \right)}{\nbVertRed \prob_{\omega_i}}.
\end{align*}
Let us define
\begin{align*}
\ma{X}_i := \frac{1}{\nbVertRed \prob_{\omega_i}} \Fou_\nbClass^\adjoint \vec{\delta}_{\omega_i} \vec{\delta}_{\omega_i}^\adjoint \Fou_\nbClass,
\end{align*}
and
\begin{align*}
\ma{X} := \sum_{i=1}^\nbVertRed \ma{X}_i = \nbVertRed^{-1} \; \Fou_\nbClass^\adjoint \ma{P}^{-1/2} \Meas^\adjoint \Meas \ma{P}^{-1/2} \Fou_\nbClass.
\end{align*}
The matrix $\ma{X}$ is thus a sum of $\nbVertRed$ of independent, random, self-adjoint, positive semi-definite matrices. We are in the setting of Lemma \ref{th:matrix_chernoff}. We continue by computing $\Ebb \,\ma{X}_i$ and $\lambda_{\rm max}(\ma{X}_i)$.

The expected value of each $\ma{X}_i$ is
\begin{align*}
\Ebb \, \ma{X}_i 
= 
\Ebb \left[ \frac{\left(\Fou_\nbClass^\adjoint \vec{\delta}_{\omega_i} \right) \left(\vec{\delta}_{\omega_i}^\adjoint \Fou_\nbClass \right)}{\nbVertRed \prob_{\omega_i}} \right] 
= 
\inv{\nbVertRed} \; \Fou_\nbClass^\adjoint \left(\sum_{i=1}^\nbVert \prob_i \frac{\vec{\delta}_{i} \vec{\delta}_{i}^\adjoint}{\prob_i}  \right) \Fou_\nbClass 
=  
\inv{\nbVertRed} \; \Fou_\nbClass^\adjoint \Fou_\nbClass = \inv{\nbVertRed} \ma{I}
\end{align*}
where $\ma{I} \in \Rbb^{\nbClass \times \nbClass}$ is the identity matrix. Therefore, 
\begin{align*}
\mu_{\rm min} := \lambda_{\rm min} \left( \sum_i \Ebb \, \ma{X}_i \right) = 1
\quad \text{and} \quad
\mu_{\rm max} := \lambda_{\rm max} \left( \sum_i \Ebb \, \ma{X}_i \right) = 1.
\end{align*}
Furthermore, for all $i = 1, \ldots, \nbVert$, we have
\begin{align*}
\lambda_{\rm max} (\ma{X}_i)
=
\norm{\ma{X}_i}_2 
\leq 
\max_{1 \leq j \leq \nbVert} \norm{ \frac{\left(\Fou_\nbClass^\adjoint \vec{\delta}_{j} \right) \left(\vec{\delta}_{j}^\adjoint \Fou_\nbClass \right)}{\nbVertRed \prob_{j}}}_2
=
\frac{1}{\nbVertRed} \; \max_{1 \leq j \leq \nbVert} \left\{ \frac{\norm{\Fou_\nbClass^\adjoint \vec{\delta}_{j}}_2^2}{\prob_j} \right\}
= \frac{(\cumCoh_{\prob}^{\nbClass})^2}{\nbVertRed}.
\end{align*}
Lemma \ref{th:matrix_chernoff} yields, for any $\delta \in (0, 1)$,
\begin{align*}
\Pbb \left\{ \lambda_{\rm min} \left( \ma{X} \right)  \leq (1 - \delta) \right\} 
& \; \leq \; 
\nbClass \cdot \left[ \frac{\ee^{-\delta}}{(1-\delta)^{1-\delta}}\right]^{m/(\cumCoh_{\prob}^{\nbClass})^2} 
\; \leq \; 
\nbClass \; \exp \left( - \frac{\delta^2 m}{3 \, (\cumCoh_{\prob}^{\nbClass})^2} \right)  \; \text { and}\\
\Pbb \left\{ \lambda_{\rm max} \left( \ma{X} \right)  \geq (1 + \delta) \right\} 
& \; \leq \; 
\nbClass \cdot \left[ \frac{\ee^{\delta}}{(1+\delta)^{1+\delta}}\right]^{m/(\cumCoh_{\prob}^{\nbClass})^2} 
\; \leq \; 
\nbClass \; \exp \left( - \frac{\delta^2 m}{3 \, (\cumCoh_{\prob}^{\nbClass})^2} \right).
\end{align*}
Therefore, for any $\delta \in (0, 1)$, we have, with probability at least $1 - \epsilon$,
\begin{align}
\label{eq:hidden_rip}
1 - \delta \leq \lambda_{\rm min} \left( \ma{X} \right) 
\quad \text{and} \quad
\lambda_{\rm max} \left( \ma{X} \right) \leq 1+\delta
\end{align}
provided that
\begin{align*}
\nbVertRed 
\geq 
\frac{3}{\delta^2} \; (\cumCoh_{\prob}^{\nbClass})^2 \; \log\left( \frac{2k}{\epsilon} \right).
\end{align*}
Noticing that \refeq{eq:hidden_rip} implies that
\begin{align*}
(1+\delta) \norm{\vec{\alpha}}_2^2
\leq
\norm{ \Meas \ma{P}^{-1/2} \Fou_\nbClass \vec{\alpha} }_2^2 
\leq 
(1+\delta) \norm{\vec{\alpha}}_2^2,
\end{align*}
for all $\vec{\alpha} \in \Rbb^k$, which is equivalent to
\begin{align*}
(1+\delta) \norm{\sig}_2^2
\leq
\norm{ \Meas \ma{P}^{-1/2} \sig }_2^2 
\leq 
(1+\delta) \norm{\sig}_2^2,
\end{align*}
for all $\sig \in \spann(\Fou_k)$, terminates the proof.
\end{proof}

The proof of Theorem \ref{th:rip_bis} is based on the following results, also obtained by Tropp.
\begin{lemma}[Theorem $2.2$, \cite{tropp11}]
Let $\mathcal{X}$ be a finite set of positive-semidefinite matrices of dimension $d \times d$, and suppose that
\begin{align*}
\max_{\ma{X} \in \mathcal{X}}\lambda_{\rm max}(\ma{X}) \leq R.
\end{align*}
Sample $\{ \ma{X}_1, \ldots, \ma{X}_l  \}$ uniformly at random from $\mathcal{X}$ without replacement. Compute
\begin{align*}
\mu_{\rm min} := l \cdot \lambda_{\rm min} \left( \Ebb \, \ma{X}_1 \right)
\quad \text{and} \quad
\mu_{\rm max} := l \cdot \lambda_{\rm max} \left( \Ebb \, \ma{X}_1 \right).
\end{align*}
Then
\begin{align*}
\Pbb \left\{ \lambda_{\rm min} \left( \sum_i \ma{X}_i \right)  \leq (1 - \delta) \mu_{\rm min} \right\} 
& \; \leq \; d \cdot \left[ \frac{\ee^{-\delta}}{(1-\delta)^{1-\delta}}\right]^{\mu_{\rm min}/R} \; \text{ for } \delta \in [0, 1], \; \text { and}\\
\Pbb \left\{ \lambda_{\rm max} \left( \sum_i \ma{X}_i \right)  \geq (1 + \delta) \mu_{\rm max} \right\} 
& \; \leq \; d \cdot \left[ \frac{\ee^{\delta}}{(1+\delta)^{1+\delta}}\right]^{\mu_{\rm max}/R} \; \text{ for } \delta \geq 0.
\end{align*}
\end{lemma}

Using Lemma \ref{th:matrix_chernoff}, one can notice that the above probability bounds would be identical if the matrices $\{ \ma{X}_1, \ldots, \ma{X}_l  \}$ were sampled uniformly at random from $\mathcal{X}$ \emph{with} replacement. It is thus not necessary to detail the complete proof which is entirely similar to the one of Theorem~\ref{th:rip}, at the exception of the sampling procedure.

%
\section{ \ifels\else - \fi Proof of the theorems in Section \ref{sec:reconstruction}}
\label{app:proof_decoder}

\begin{proof}[Proof of Theorem \ref{th:ideal_decoder}]
We recall that $\sig^*$ is a solution to \refeq{eq:optimal_decoder}. By optimality of $\sig^*$, we have
\begin{align*}
\norm{ \ma{P}^{-1/2}_\Omega \Meas \sig^* -  \ma{P}^{-1/2}_\Omega \meas }_2 \leq \norm{ \ma{P}^{-1/2}_\Omega \Meas \vec{z} - \ma{P}^{-1/2}_\Omega \meas }_2
\end{align*}
for any $\vec{z} \in \spann(\Fou_\nbClass)$. In particular for $\vec{z} = \sig$, we obtain
\begin{align*}
\norm{ \ma{P}^{-1/2}_\Omega \Meas \sig^* -  \ma{P}^{-1/2}_\Omega \meas }_2 \leq \norm{ \ma{P}^{-1/2}_\Omega \Meas \sig - \ma{P}^{-1/2}_\Omega \meas}_2,
\end{align*}
which yields
\begin{align}
\label{eq:upper_bound_optim_decoder}
\norm{ \ma{P}^{-1/2}_\Omega \Meas \sig^* - \ma{P}^{-1/2}_\Omega \Meas \sig - \ma{P}^{-1/2}_\Omega \err}_2 \leq \norm{\ma{P}^{-1/2}_\Omega \err}_2.
\end{align}
Then, the triangle inequality and \refeq{eq:RIP} yields
\begin{align}
\norm{ \ma{P}^{-1/2}_\Omega \Meas (\sig^* -  \sig) - \ma{P}^{-1/2}_\Omega \err}_2 
& \geq 
\norm{ \ma{P}^{-1/2}_\Omega \Meas (\sig^* -  \sig)} - \norm{\ma{P}^{-1/2}_\Omega \err}_2 \nonumber \\
& =
\norm{\Meas \ma{P}^{-1/2} (\sig^* -  \sig)}_2 - \norm{\ma{P}^{-1/2}_\Omega \err}_2 \nonumber \\
& \geq
\label{eq:lower_bound_optim_decoder}
\sqrt{\nbVertRed \, (1 - \delta)} \norm{\sig^* -  \sig}_2 - \norm{\ma{P}^{-1/2}_\Omega \err}_2.
\end{align}
In the second step, we used the fact that $\Meas \ma{P}^{-1/2} = \ma{P}_{\Omega}^{-1/2} \Meas$. Combining \refeq{eq:upper_bound_optim_decoder} and \refeq{eq:lower_bound_optim_decoder} directly yields \refeq{eq:error_reconstruction_optimal}, the first bound in Theorem \ref{th:ideal_decoder}.

To prove the second bound, let us choose $\err_0 = \Meas \vec{z}_0$ with $\vec{z}_0 \in \spann(\Fou_k)$. Therefore, $\meas = \Meas (\sig + \vec{z}_0)$ and $\sig^* = \sig + \vec{z}_0$ is an obvious solution to \refeq{eq:optimal_decoder} in this case. To finish the proof, we use \refeq{eq:RIP} which yields
\begin{align*}
\norm{\sig^* -  \sig}_2
= 
\norm{\vec{z}_0}_2
& \geq
\inv{\sqrt{m(1+\delta)}} \, \norm{\Meas \ma{P}^{-1/2} \vec{z}_0}_2
=
\inv{\sqrt{m(1+\delta)}} \, \norm{\ma{P}^{-1/2}_\Omega \Meas \vec{z}_0}_2 \\
& = 
\inv{\sqrt{m(1+\delta)}} \, \norm{\ma{P}^{-1/2}_\Omega \err_0}_2.
\end{align*}
\end{proof}
\begin{proof}[Proof of Theorem \ref{th:practical_decoder}]
As $\sig^*$ is a solution to \refeq{eq:practical_decoder}, we have
\begin{align}
\label{eq:opt_cond_reg_problem}
\norm{ \ma{P}^{-1/2}_\Omega \left(\Meas \sig^* - \meas \right) }_2^2 + \reg \; (\sig^*)^\adjoint g(\Lap) \sig^*
\leq
\norm{ \ma{P}^{-1/2}_\Omega \left(\Meas \vec{z} - \meas \right) }_2^2 + \reg \; \vec{z}^\adjoint g(\Lap) \vec{z},
\end{align}
for all $\vec{z} \in \Rbb^\nbVert$. We also have $\sig^* = \vec{\alpha}^*  + \vec{\beta}^*$ with $\vec{\alpha}^* \in \spann(\Fou_\nbClass)$ and $\vec{\beta}^* \in \spann(\bar{\Fou}_\nbClass)$. Let us define the matrix 
\begin{align*}
\bar{\Fou}_\nbClass := \left( \fou_{\nbClass+1}, \ldots, \fou_\nbVert \right) \in \Rbb^{\nbVert \times (\nbVert - \nbClass)}.
\end{align*}
Choosing $\vec{z} = \sig$ in \refeq{eq:opt_cond_reg_problem} and using the facts that $\bar{\Fou}_\nbClass^\adjoint \vec{\alpha}^* = \vec{0}$, $\Fou_\nbClass^\adjoint \vec{\beta}^* = \vec{0}$, $\bar{\Fou}_\nbClass^\adjoint \sig = \vec{0}$, and that $g(\Lap) = \Fou \, g(\Lap) \, \Fou^\adjoint$, we obtain
\begin{align*}
& \norm{ \ma{P}^{-1/2}_\Omega \left(\Meas \sig^* - \meas \right) }_2^2 
+ \reg \; (\Fou_\nbClass^\adjoint \vec{\alpha}^*)^\adjoint \; \ma{G}_\nbClass \; (\Fou_\nbClass^\adjoint \vec{\alpha}^*) 
+ \reg \; (\bar{\Fou}_\nbClass^\adjoint \vec{\beta}^*)^\adjoint \; \bar{\ma{G}}_\nbClass \; (\bar{\Fou}_\nbClass^\adjoint \vec{\beta}^*) \\
& \hspace{7cm} \leq \;
\norm{ \ma{P}^{-1/2}_\Omega \err }_2^2 
+ \reg \; (\Fou_\nbClass^\adjoint \sig)^\adjoint \; \ma{G}_\nbClass \; (\Fou_\nbClass^\adjoint \sig),
\end{align*}
where
\begin{align*}
\ma{G}_\nbClass := \diag\left( g(\eig_{1}), \ldots, g(\eig_{\nbClass}) \right) \in \Rbb^{\nbClass \times \nbClass}
\text{ and }
\bar{\ma{G}}_\nbClass := \diag\left( g(\eig_{\nbClass+1}), \ldots, g(\eig_\nbVert) \right) \in \Rbb^{(\nbVert - \nbClass) \times (\nbVert - \nbClass)}.
\end{align*}
We deduce that
\begin{align*}
\norm{ \ma{P}^{-1/2}_\Omega \left(\Meas \sig^* - \meas \right) }_2^2 
+ \reg \; g(\eig_{\nbClass+1}) \norm{\vec{\beta}^*}_2^2
\; \leq \;
\norm{ \ma{P}^{-1/2}_\Omega \err }_2^2  
+ \reg \; g(\eig_\nbClass) \norm{ \sig  }_2^2,
\end{align*}
where we used the fact that $\norm{ \bar{\Fou}_\nbClass^\adjoint \vec{\beta}^* }_2 = \norm{ \vec{\beta}^* }_2$ and $\norm{ \Fou_\nbClass^\adjoint \sig }_2 = \norm{ \sig }_2$. As the left hand side of the last inequality is a sum of two positive quantities, we also have
\begin{align}
\label{eq:bound_fidelity_term}
\norm{ \ma{P}^{-1/2}_\Omega \left(\Meas \sig^* - \meas \right) }_2 
& \; \leq \;
\norm{ \ma{P}^{-1/2}_\Omega \err }_2 + \sqrt{\reg g(\eig_\nbClass)} \norm{ \sig  }_2
\text{ and }\\
\label{eq:bound_non_model_term}
\sqrt{\reg g(\eig_{\nbClass+1})} \norm{\vec{\beta}^*}_2
& \; \leq \;
\norm{ \ma{P}^{-1/2}_\Omega \err }_2 + \sqrt{\reg g(\eig_\nbClass)} \norm{ \sig  }_2.
\end{align}
Inequality \refeq{eq:bound_non_model_term} proves \refeq{eq:bound_beta}, the second inquality in Theorem~\ref{th:practical_decoder}. It remains to prove \refeq{eq:bound_alpha}. To prove this inequality, we continue by using \refeq{eq:RIP}, which yields
\begin{align}
\norm{ \ma{P}^{-1/2}_\Omega \left(\Meas \sig^* - \meas \right)  }_2 
& =
\norm{ \ma{P}^{-1/2}_\Omega \Meas (\vec{\alpha}^* - \sig) + \ma{P}^{-1/2}_\Omega \err + \ma{P}^{-1/2}_\Omega \Meas \vec{\beta}^*}_2 \nonumber \\
& \geq
\norm{\ma{P}^{-1/2}_\Omega \Meas (\vec{\alpha}^* - \sig)}_2 - \norm{\ma{P}^{-1/2}_\Omega \err} - \norm{\ma{P}^{-1/2}_\Omega \Meas \vec{\beta}^*}_2 \nonumber \\
& =
\norm{\Meas \ma{P}^{-1/2} (\vec{\alpha}^* - \sig)}_2 - \norm{\ma{P}^{-1/2}_\Omega \err} - \norm{\Meas \ma{P}^{-1/2} \vec{\beta}^*}_2 \nonumber \\
& \geq
\label{eq:lower_bound_reg_prob}
\sqrt{\nbVertRed (1-\delta)} \; \norm{ \vec{\alpha}^* - \sig }_2 - \norm{\ma{P}^{-1/2}_\Omega \err}_2 - M_{\rm max} \norm{\vec{\beta}^*}_2.
\end{align}
Finally, combining \refeq{eq:bound_fidelity_term}, \refeq{eq:bound_non_model_term} and \refeq{eq:lower_bound_reg_prob} gives
\begin{align*}
\norm{ \vec{\alpha}^* - \sig }_2
& \leq
\inv{\sqrt{\nbVertRed (1-\delta)}} \left(\norm{ \ma{P}^{-1/2}_\Omega \left(\Meas \sig^* - \meas \right) }_2 + \norm{\ma{P}^{-1/2}_\Omega \err}_2 + M_{\rm max} \norm{\vec{\beta}^*}_2 \right)\\
& \leq
\inv{\sqrt{\nbVertRed (1-\delta)}} 
\left( \norm{\ma{P}^{-1/2}_\Omega \err}_2 + \sqrt{\reg g(\eig_\nbClass)} \norm{ \sig  }_2 \right. \\
& \hspace{2.5cm} + \norm{\ma{P}^{-1/2}_\Omega \err}_2 \\
& \hspace{2.5cm} \left. + \frac{M_{\rm max} \norm{\ma{P}^{-1/2}_\Omega \err}_2}{\sqrt{\reg g(\eig_{\nbClass+1})}} + M_{\rm max} \sqrt{\frac{g(\eig_\nbClass)}{g(\eig_{\nbClass+1})}} \norm{ \sig  }_2 \right) \\
& \leq 
\inv{\sqrt{\nbVertRed (1-\delta)}} 
\left( 2 +  \frac{M_{\rm max}}{\sqrt{\reg g(\eig_{\nbClass+1})}}\right)\norm{\ma{P}^{-1/2}_\Omega \err}_2 \\
& \hspace{1cm} + \inv{\sqrt{\nbVertRed (1-\delta)}} \left( M_{\rm max} \sqrt{\frac{g(\eig_\nbClass)}{g(\eig_{\nbClass+1})}} + \sqrt{\reg g(\eig_\nbClass)}\right) \norm{ \sig  }_2.
\end{align*}
This terminates the proof.
\end{proof}
%

%
\section{ \ifels\else - \fi Proof of the theorem in Section \ref{sec:estimation_distribution}}
\label{app:proof_jl}

We use the classical technique to prove the Johnson-Lindenstrauss lemma (see, \eg, \cite{baraniuk08}).

\begin{proof}
Each filtered signal $\vec{r}_{\hat{c}_\lambda}^l$, $l \in \{1, \ldots, L\}$, satisfies
\begin{align*}
\vec{r}_{\hat{c}_\lambda}^l = \Fou \, \ma{C}_\lambda \, \Fou^\adjoint \vec{r}^l,
\end{align*}
where $\ma{C}_\lambda := \diag(\hat{c}_\lambda(\eig_1), \ldots, \hat{c}_\lambda(\eig_\nbVert))$. Let $i$ be fixed for the moment. We have
\begin{align*}
\sum_{l=1}^L \, (\vec{r}_{\hat{c}_\lambda}^l)_i^2
\; = \;
\sum_{l=1}^L \, (\vec{\delta}_i^\adjoint \Fou \, \ma{C}_\lambda \, \Fou^\adjoint \vec{r}^l)^2,
\end{align*}
The expected value of this sum is $\norm{\ma{C}_\lambda \Fou^\adjoint \vec{\delta}_i}_2^2$. Indeed,
\begin{align*}
\Ebb \left[ \sum_{l=1}^L \, (\vec{r}_{\hat{c}_\lambda}^l)_i^2 \right]
& \; = \;
\sum_{l=1}^L \, \vec{\delta}_i^\adjoint \Fou \, \ma{C}_\lambda \, \Fou^\adjoint \; \Ebb \left[ \vec{r}^l (\vec{r}^l)^\adjoint \right] \; \Fou \ma{C}_\lambda \Fou^\adjoint \vec{\delta}_i
\; = \;
L^{-1} \; \sum_{l=1}^L \, \vec{\delta}_i^\adjoint \Fou \, \ma{C}_\lambda \, \Fou^\adjoint \Fou \ma{C}_\lambda \Fou^\adjoint \vec{\delta}_i \\
& \; = \;
\vec{\delta}_i^\adjoint \Fou \, \ma{C}_\lambda^2 \Fou^\adjoint \vec{\delta}_i
\; = \;
\norm{\ma{C}_\lambda \Fou^\adjoint \vec{\delta}_i}_2^2.
\end{align*}
Let us define 
\begin{align*}
X_i := \sum_{l=1}^L \, \left[ (\vec{r}_{\hat{c}_\lambda}^l)_i^2 - L^{-1} \norm{\ma{C}_\lambda \Fou^\adjoint \vec{\delta}_i}_2^2 \right].
\end{align*}
This is a sum of $L$ independent centered random variables. Furthermore, as each $\vec{r}^l$ is a zero-mean Gaussian random vector with covariance matrix $L^{-1} \, \ma{I}$, the variables $(\vec{r}_{\hat{c}_\lambda}^l)_i$ are subgaussian with subgaussian bounded by $C \, L^{-1/2} \, \norm{\ma{C}_\lambda \Fou^\adjoint \vec{\delta}_i}_2$, where $C\geq1$ is an absolute constant. We let the reader refer to, \eg, \cite{vershynin12} for more information on the definition and properties of subgaussian random variables. Using Lemma $5.14$ and Remark $5.18$ in \cite{vershynin12}, one can prove that each summand of $X$ is a centered subexponential random variable with subexponentinal norm bounded by $4C^2 \, L^{-1} \, \norm{\ma{C}_\lambda \Fou^\adjoint \vec{\delta}_i}_2^2$. Corollary $5.17$ in \cite{vershynin12} shows that there exists an absolute contant $c>0$ such that
\begin{align*}
\Pbb \left( \abs{X_i} \geq t \, L  \right) \leq 2 \exp \left( - \frac{c \, L \, t^2}{16C^4 \, L^{-2} \, \norm{\ma{C}_\lambda \Fou^\adjoint \vec{\delta}_i}_2^4} \right)
\end{align*}
for all $t \in (0, 4C^2 \, L^{-1} \, \norm{\ma{C}_\lambda \Fou^\adjoint \vec{\delta}_i}_2^2)$, or, equivalently, that
\begin{align*}
\Pbb \left( \abs{X_i} \geq \delta \, \norm{\ma{C}_\lambda \Fou^\adjoint \vec{\delta}_i}_2^2  \right) \leq 2 \exp \left( - \frac{c \, L \,\delta^2}{16C^4} \right),
\end{align*}
for all $\delta \in (0, 4C^2)$. 

Then, using the union bound, we obtain
\begin{align*}
\Pbb \left( \max_{i \in \{1, \ldots, \nbVert\} } \abs{X_i} \geq \delta \, \norm{\ma{C}_\lambda \Fou^\adjoint \vec{\delta}_i}_2^2  \right) \leq 2 \nbVert \exp \left( - \frac{c \, L \,\delta^2}{16C^4} \right).
\end{align*}
This proves that, with probability at least $1 - \epsilon$,
\begin{align*}
(1-\delta) \norm{\ma{C}_\lambda \Fou^\adjoint \vec{\delta}_i}_2^2
\; \leq \;
\sum_{l=1}^L \, (\vec{r}_{\hat{c}_\lambda}^l)_i^2 
\; \leq \;
(1+\delta) \norm{\ma{C}_\lambda \Fou^\adjoint \vec{\delta}_i}_2^2,
\end{align*}
for all $i \in \{1, \ldots, \nbVert \}$, provided that
\begin{align*}
L \geq  \frac{16C^4} {c \, \delta^2}  \; \log \left( \frac{2 \nbVert}{\epsilon} \right).
\end{align*}
To finish the the proof, one just needs to remark that
\begin{align*}
\ma{C}_\lambda \Fou^\adjoint \vec{\delta}_i = \Fou_{j^*}^\adjoint \vec{\delta}_i + \ma{E}_\lambda \Fou^\adjoint \vec{\delta}_i,
\end{align*}
by definition of $\hat{c}_\lambda$ (see \refeq{eq:def_c_lambda}) and use the triangle inequality.
\end{proof}